\documentclass[preprint, 3p, times]{elsarticle}

\usepackage{custom-preamble}

\crefname{appendix}{}{}
\biboptions{sort,square}

\begin{document}

\begin{frontmatter}


\title{Exact and Heuristic Computation of the Scanwidth of Directed Acyclic Graphs\tnoteref{fund1}}

\author[inst1]{Niels Holtgrefe\corref{cor1}}
\ead{N.A.L.Holtgrefe@tudelft.nl}

\author[inst1]{Leo van Iersel}
\ead{L.J.J.vanIersel@tudelft.nl}

\author[inst2]{Mark Jones}
\ead{M.Jones@mdx.ac.uk}

\affiliation[inst1]{
            organization={Delft Institute of Applied Mathematics, Delft University of Technology},
            city={Delft},
            country={The Netherlands}}

\affiliation[inst2]{
            organization={Department of Computer Science, Middlesex University},
            city={London},
            country={United Kingdom}}

\tnotetext[fund1]{This paper received funding from the Dutch Research Council (NWO) under projects OCENW.M.21.306 and OCENW.KLEIN.125.}

\cortext[cor1]{Corresponding author}

\begin{abstract}
To measure the tree-likeness of a directed acyclic graph (DAG), a new width parameter that considers the directions of the arcs was recently introduced: \emph{scanwidth}. We present the first algorithm that efficiently computes the exact scanwidth of general DAGs. For DAGs with one root and scanwidth $k$ it runs in $O(k \cdot n^k \cdot m)$ time. The algorithm also functions as an FPT algorithm with complexity $O(2^{4 \ell - 1} \cdot \ell \cdot n + n^2)$ for phylogenetic networks of level-$\ell$, a type of DAG used to depict evolutionary relationships among species. Our algorithm performs well in practice, being able to compute the scanwidth of synthetic networks up to 30 reticulations and 100 leaves within 500 seconds. Furthermore, we propose a heuristic that obtains an average practical approximation ratio of 1.5 on these networks. While we prove that the scanwidth is bounded from below by the treewidth of the underlying undirected graph, experiments suggest that for networks the parameters are close in practice.
\end{abstract}



\end{frontmatter}

\section{Introduction}
Treewidth is an extensively researched width measure that quantifies the `tree-likeness' of an undirected graph. The parameter comes with a corresponding tree decomposition, which represents the graph in a tree-like way (see e.g.~\cite{diestel2017graph}). When considering directed acyclic graphs (DAGs), tree decompositions become less natural since they do not preserve the natural top-to-bottom structure of a DAG. With this in mind, Berry, Scornavacca, and Weller~\cite{Berry2022b} recently developed a DAG-specific measure of tree-likeness: scanwidth. Contrary to treewidth, scanwidth is not agnostic to the directions of the arcs, thus making it a more natural parameter for DAGs.

Berry, Scornavacca, and Weller motivate the use of scanwidth from the perspective of phylogenetics: the study of evolutionary relationships among different genes, species, or populations. These relationships can be represented by a type of DAG with a single root: a phylogenetic network. The leaves of such a network represent the studied set of species (or taxa), whereas the root is a common ancestor. The arcs and their directions depict how species evolve over time, while the internal vertices are either reticulate events (where multiple species converge), or speciation events (where a species diverges into multiple species).

Phylogenetics is a research area with many computationally hard problems such as \textsc{Network Inference} \cite{rabier2021inference}, \textsc{Tree Containment} \cite{iersel2023embedding, vanIersel2017}, \textsc{Small Parsimony} \cite{Scornavacca2022} and \textsc{Hybridization Number} \cite{bernardini2023constructing, Borst2022}. A common way to solve such problems is by using so-called parameterized algorithms, where the time complexity is expressed as a function of both the input size and an additional parameter that is small in practice. In phylogenetics these algorithms tend to exploit the observation that reticulate events are fairly rare for most practical phylogenetic networks (when compared to the overall size of the network). In some sense, these networks thus still remain somewhat tree-like. Two well-known parameters in phylogenetics that measure this tree-likeness are the reticulation number and level of a network (see e.g. \cite{van2010phylogenetic}). Both measures have seen a vast amount of successful uses as a parameter in parameterized algorithms (see e.g. \cite{Borst2022, vanIersel2017} and the survey \cite{bulteau2019parameterized}).

Since the treewidth is smaller than both the reticulation number and the level, it has recently attracted interest from the phylogenetics community. Although treewidth has already successfully been applied for parameterized algorithms in phylogenetics \cite{iersel2023embedding, Scornavacca2022}, using scanwidth seems to be more natural, thus allowing for a more intuitive design of parameterized algorithms \cite{Scornavacca2022}. Indeed, algorithms in phylogenetics relying on scanwidth are already starting to appear \cite{rabier2021inference}. 

Berry, Scornavacca, and Weller named scanwidth after the informal concept of `scanning' a DAG. Imagine a `scanner line' for each leaf of a DAG. As these scanner lines move up through the DAG, they scan the arcs of the DAG. A scanner line merges with another scanner line when they meet at a vertex. The order in which the arcs and vertices are scanned is determined by a tree extension: a tree on the same vertex set as the original DAG, with the constraint that it maintains the natural ordering of the DAG. Therefore, this tree extension functions as a route for the scanner lines. The goal is now to find a tree extension that minimizes the maximum number of arcs that are cut by a line during the scanning. This number is referred to as the scanwidth of the DAG. \cref{fig:scanning a graph} provides an illustration of the concept of scanning a DAG.

\begin{figure}[htb]
     \centering
     \begin{subfigure}[b]{0.47\textwidth}
         \centering
		\begin{tikzpicture}[scale=0.46]
	\begin{pgfonlayer}{nodelayer}
		\node [style={graph_node}, label=right:{$b$}] (41) at (2, 0) {};
		\node [style={graph_node}, label=right:{$c$}] (42) at (4, 0) {};
		\node [style={graph_node}, label=right:{$a$}] (43) at (0, 1.5) {};
		\node [style={graph_node}, label=right:{$x$}] (44) at (2, 1.5) {};
		\node [style={graph_node}, label=right:{$y$}] (45) at (4, 1.5) {};
		\node [style={graph_node}, label=right:{$u$}] (46) at (1, 3) {};
		\node [style={graph_node}, label=right:{$v$}] (47) at (3, 3) {};
		\node [style={graph_node}, label=right:{$q$}] (48) at (2, 4.5) {};
		\node [style={graph_node}, label=right:{$w$}] (49) at (5, 3) {};
		\node [style={graph_node}, label=right:{$\rho$}] (50) at (3, 6) {};
		\node [style={graph_node}, label=right:{$d$}] (51) at (6, 1.5) {};
	\end{pgfonlayer}
	\begin{pgfonlayer}{edgelayer}
		\draw [style={graph_edge}] (50) to (48);
		\draw [style={graph_edge}] (48) to (46);
		\draw [style={graph_edge}] (46) to (43);
		\draw [style={graph_edge}] (48) to (47);
		\draw [style={graph_edge}] (46) to (44);
		\draw [style={graph_edge}] (47) to (44);
		\draw [style={graph_edge}] (44) to (41);
		\draw [style={graph_edge}] (47) to (45);
		\draw [style={graph_edge}] (45) to (42);
		\draw [style={graph_edge}] (50) to (49);
		\draw [style={graph_edge}] (49) to (45);
		\draw [style={graph_edge}] (49) to (51);
	\end{pgfonlayer}
\end{tikzpicture}
         \caption{DAG with single root}
         \label{subfig:scanning_network}
     \end{subfigure}
     \hfill
     \begin{subfigure}[b]{0.47\textwidth}
         \centering
         \begin{tikzpicture}[scale=0.46]
	\begin{pgfonlayer}{nodelayer}
		\node [style={graph_node},label=left:{$b$}] (1) at (3.5, 2) {};
		\node [style={graph_node},label=below left:{$c$}] (2) at (6.5, 3.5) {};
		\node [style={graph_node},label=below:{$a$}] (3) at (0.5, 3.5) {};
		\node [style={graph_node},label=right:{$x$}] (4) at (3.5, 3.5) {};
		\node [style={graph_node},label=above right:{$y$}] (5) at (5, 4.5) {};
		\node [style={graph_node},label=above left:{$u$}] (6) at (2, 4.5) {};
		\node [style={graph_node},label=above left:{$v$}] (7) at (3.5, 5.5) {};
		\node [style={graph_node},label=above left:{$q$}] (8) at (5, 6.5) {};
		\node [style={graph_node},label=above right:{$w$}] (9) at (6.5, 7.5) {};
		\node [style={graph_node},label=above right:{$\rho$}] (10) at (5, 8.5) {};
		\node [style={graph_node},label=below:{$d$}] (11) at (8, 6.5) {};
		\node [style=none] (12) at (2.75, 2.75) {};
		\node [style=none] (13) at (1.5, 3.25) {};
		\node [style=none] (14) at (0.75, 4.5) {};
		\node [style=none] (15) at (4.25, 2.75) {};
		\node [style=none] (16) at (5.25, 3.5) {};
		\node [style=none] (17) at (6.25, 4.5) {};
		\node [style=none] (18) at (2.25, 3.5) {};
		\node [style=none] (19) at (3.5, 4.5) {};
		\node [style=none] (20) at (3, 4.5) {};
		\node [style=none] (21) at (2, 5.5) {};
		\node [style=none] (22) at (3.75, 4.5) {};
		\node [style=none] (23) at (5, 5.25) {};
		\node [style=none] (24) at (4.75, 5.5) {};
		\node [style=none] (25) at (3.75, 6.5) {};
		\node [style=none] (26) at (5, 7.5) {};
		\node [style=none] (27) at (6, 6.5) {};
		\node [style=none] (28) at (6.75, 6.5) {};
		\node [style=none] (29) at (7.75, 7.5) {};
		\node [style=none] (30) at (5.25, 7.5) {};
		\node [style=none] (31) at (6.25, 8.5) {};
	\end{pgfonlayer}
		\draw [style={extension_edge}] (10.center) to (9.center);
		\draw [style={extension_edge}] (9.center) to (8.center);
		\draw [style={extension_edge}] (8.center) to (7.center);
		\draw [style={extension_edge}] (7.center) to (5.center);
		\draw [style={extension_edge}] (5.center) to (2.center);
		\draw [style={extension_edge}] (4.center) to (1.center);
		\draw [style={extension_edge}] (6.center) to (3.center);
		\draw [style={extension_edge}] (7.center) to (6.center);
		\draw [style={extension_edge}] (6.center) to (4.center);			
		\draw [style={extension_edge}] (9.center) to (11.center);
		\draw [style={cut_line}, draw=red!17, bend left=15] (14.center) to (13.center); 
		\draw [style={cut_line}, draw=red!17, bend left=15] (12.center) to (15.center);
		\draw [style={cut_line}, draw=red!17, bend left=15] (16.center) to (17.center);
		\draw [style={cut_line}, draw=red!34, bend left=15] (18.center) to (19.center);
		\draw [style={cut_line}, draw=red!50, bend right=15] (20.center) to (21.center);
		\draw [style={cut_line}, draw=red!34, in=165, out=60, looseness=1.25] (22.center) to (23.center); 
		\draw [style={cut_line}, draw=red!67, bend right=15] (24.center) to (25.center);
		\draw [style={cut_line}, draw=red!82, bend left=15] (26.center) to (27.center);
		\draw [style={cut_line}, draw=red!17, bend left=15] (28.center) to (29.center); 
		\draw [style={cut_line}, draw=red!100, bend left=15] (30.center) to (31.center);
	\begin{pgfonlayer}{edgelayer}
		\draw [style={graph_edge}, bend left=45, looseness=2.50] (10) to (8);
		\draw [style={graph_edge}, bend right=15] (8) to (6);
		\draw [style={graph_edge}] (6) to (3);
		\draw [style={graph_edge}] (8) to (7);
		\draw [style={graph_edge}] (6) to (4);
		\draw [style={graph_edge}, bend right=45, looseness=2.50] (7) to (4);
		\draw [style={graph_edge}] (4) to (1);
		\draw [style={graph_edge}] (7) to (5);
		\draw [style={graph_edge}] (5) to (2);
		\draw [style={graph_edge}] (10) to (9);
		\draw [style={graph_edge}, in=135, out=-120, looseness=2.25] (9) to (5);
		\draw [style={graph_edge}] (9) to (11);
	\end{pgfonlayer}
\end{tikzpicture}
         \caption{Scanning of the DAG}
     \end{subfigure}
        \caption{A DAG with a single root (a), and a tree extension of the DAG (b), functioning as a route for the scanner lines. The tree extension is indicated by the grey edges, while the arcs of the DAG are drawn back in, following the edges of the tree extension. The scanner lines start at the leaves and move up through the tree extension, at each step becoming brighter red. The tree extension is optimal, thus the scanwidth of this DAG is 3: the maximum number of DAG arcs that are cut by one of the scanner lines.}
        \label{fig:scanning a graph}
\end{figure}
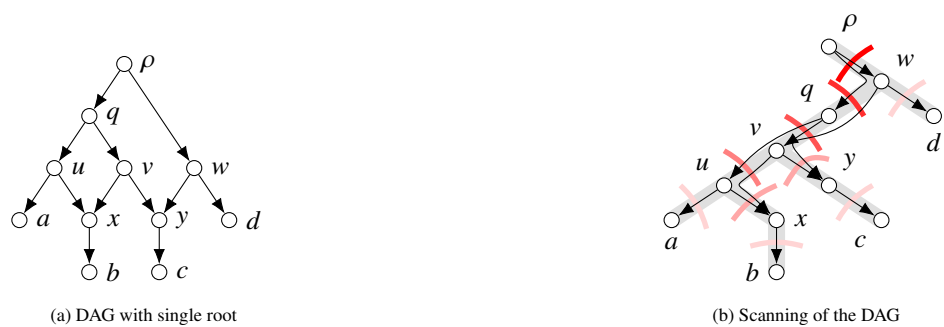

An NP-hardness proof and some structural results were already presented in~\cite{Berry2022b}. Apart from that, scanwidth has only very briefly been mentioned in \cite{rabier2021inference,Scornavacca2022}. Magne et al. \cite{magne2023edge} independently introduced the closely related edge-treewidth, which can be considered the undirected analogue of scanwidth. Similar to scanwidth, edge-treewidth has not seen other research efforts yet. A second closely related parameter is the directed cutwidth (see e.g. \cite{bodlaender2009derivation}). Its solution space is more restrictive than it is for scanwidth since it only allows linear orderings instead of tree extensions. Consequently, most results for directed cutwidth are not immediately transferable to scanwidth.

This paper aims to advance the understanding of scanwidth, with a special focus on exact and heuristic algorithms to compute the NP-hard parameter. The tree extensions returned by these algorithms can then be used by future algorithms relying on scanwidth as a parameter.

The structure of the paper is as follows. The first three subsections of \cref{sec:preliminaries} contain the formal introduction to scanwidth as well as some preliminary graph theoretical notions. In \cref{subsec:bounds} we prove that the scanwidth is bounded from below by the treewidth. Although this bound is already mentioned in \cite{Berry2022b}, it has not previously been formally proved. We also generalize the result by Rabier et al. \cite{rabier2021inference}, that the scanwidth of a binary level-$k$ network is at most $k+1$, to non-binary networks. \cref{sec:reductions} discusses some reduction rules to decrease the size of instances.

\cref{sec:exact_methods} is focuses on exact algorithms that compute the scanwidth. Using a recursive approach, the $O (n! \cdot n \cdot m)$ brute force time complexity is first improved to $O(4^n \cdot \mathrm{polylog}(n))$. This algorithm forms the basis of an algorithm with running time $O((k+r-1) \cdot m \cdot n^{k+r-1})$ for DAGs with $r$ roots, where $k$ is the scanwidth. This proves that with respect to the parameter scanwidth, the scanwidth problem for DAGs with a fixed number of roots falls within the complexity class XP, containing problems solvable in \emph{slicewise polynomial} time: $n^{f(k)}$ time for some computable function~$f$. Using reduction rules, this same algorithm can also be parametrized by the level $\ell$ of a network, running in $O(2^{4 \ell - 1} \cdot \ell \cdot n + n^2)$ time. Therefore, with the level as parameter, the scanwidth problem is in the class of \emph{fixed-parameter tractable} (FPT) problems, which contains problems solvable in $f(\ell) \cdot n^c$ time for some computable function $f$ and constant $c$.

\cref{sec:heuristics} explores different heuristics, with a cut-splitting heuristic showing positive results. Although the heuristic is proved to be non-optimal in general, computational experiments in \cref{sec:experiments} are promising. They show that the heuristic attains an average practical approximation ratio of 1.5 for networks of 30 reticulations and 100 leaves, when enhanced with simulated annealing. The XP algorithm is shown to be the fastest exact computation method in practice, being able to compute the scanwidth of networks up to 30 reticulations and 100 leaves within 500 seconds. Moreover, $88.9 \%$ of those networks are solvable within 60 seconds.

Finally in \cref{sec:conclusion}, we discuss the results and directions for further research.

\section{Preliminaries}\label{sec:preliminaries}
We will use the big-Oh notation to analyze complexities of algorithms. Additionally, we use the less common $\tilde{O}$ to suppress the polynomial and logarithmic factors in the time complexity, resulting in, for example, $O(2^n \cdot n^3 \cdot \log n) = \tilde{O} (2^n)$. Regarding graph theoretical concepts, we mostly follow the standard notation as presented in \cite{diestel2017graph}. Below, we will present some additional conventions and possibly less common notions.

Similarly to the degree $\delta (v)$ of a vertex $v$, we write $\delta (W)$ for the degree of $W \subseteq V$, which counts the number of edges with one endpoint in $W$ and one endpoint in $V \setminus W$ for a graph $G=(V,E)$. To avoid confusion, we sometimes use $\delta_G (v)$ or $\delta_G (W)$ to emphasize the graph $G$. We write $u \connect{G} v$ if two vertices $u$ and $v$ are connected (i.e. there exists a path between them) in a graph $G$. A \emph{block} (or `biconnected component') is a maximal connected induced subgraph without any cut vertices (i.e. vertices whose removal increases the number of connected components).

A \emph{weakly connected} directed graph is a graph whose underlying undirected graph is connected. Analogously, two vertices $u$ and $v$ are weakly connected in $G$ (denoted by $u \connect{G} v$) if they are connected in the underlying undirected graph. A \emph{(weakly connected) component} of $G$ is a maximal weakly connected induced subgraph of $G$. A vertex in a directed graph is a \emph{cut vertex} if its deletion increases the number of weakly connected components. A \emph{block} of a directed graph is a maximal weakly connected induced subgraph without any cut vertices.

\paragraph{Directed acyclic graphs}
If a directed graph $G$ contains no directed cycles, we call it a \emph{directed acyclic graph} (DAG). In a DAG, a vertex with indegree 0 is called a \emph{root} (often labelled as $\rho$), and a vertex with outdegree 0 is referred to as a \emph{leaf}. If $G$ has exactly one root, we call $G$ \emph{rooted}. Otherwise, if $G$ has multiple roots, it is \emph{multi-rooted}. The tails of arcs that enter a vertex $v$ are the \emph{parents} of $v$. Similarly, the heads of arcs coming out of $v$ are \emph{children} of $v$. We call $W \subseteq V$ a \emph{sinkset} if $\deltaout (W) = 0$. We write $W \sqsubseteq U$ for any $U \subseteq V$ if both $W \subseteq U$ and $W$ is a sinkset. Since a DAG contains no directed cycles, it naturally exhibits a top-to-bottom structure. More formally, it defines a partial order on its vertices: we write $v <_G u$ if there exists a directed path from $u$ to $v$.

Unless otherwise specified, this paper will consider each graph $G$ to be a weakly connected, directed acyclic graph without self-loops and parallel arcs. If it is clear from the context, we sometimes drop the adjective `directed' from the notions defined above.

\paragraph{Phylogenetic networks}
\label{subsec:phylogenetic_networks}

A rooted, weakly connected DAG $G=(V,E)$ is a \emph{(rooted) network} if each vertex $v\in V$ is of one of the following types: (i) (unique) \emph{root} with $\deltain(v) = 0$; (ii) \emph{leaf} with $\deltain (v) = 1$ and $\deltaout (v) = 0$; (iii) \emph{tree-vertex} with $\deltain (v) = 1$ and $\deltaout (v) \geq 2$; (iv) \emph{reticulation (vertex)} with $\deltain (v) \geq 2$ and $\deltaout(v) = 1$. Furthermore, if the root has degree 2, the leaves degree 1, and all other vertices degree 3, $G$ is a \emph{binary network}.

The \emph{reticulation number} $r(G)$ of a network $G=(V,E)$ is the sum of indegrees of all reticulation vertices, minus the number of reticulation vertices. A network has \emph{level} $k$ (or is level-$k$) if each block of the network has reticulation number at most $k$. The reticulation number and level of a network are often only defined for binary networks but were generalized to non-binary networks in \cite{van2010phylogenetic}. 

A network~$N$ is a \emph{phylogenetic network} on a set of labels $X$ if its leaves are bijectively labelled by the elements of $X$. As an example, consider the graph from \cref{subfig:scanning_network}, which serves as a phylogenetic network on the set of labels $\{a,b,c,d \}$.

\subsection{Graph layouts}
\label{subsec:graph_layouts}
A \emph{(linear) layout} $\sigma$ (also known as a `linear ordering', `linear arrangement', `numbering', or `labelling') of a graph $G$ is a total ordering of its vertex set. This ordering can be represented by a directed path on $V(G)$. A \emph{tree layout} $\Gamma$ (also known as an `agreeing tree') of a graph $G$ is a partial ordering of its vertex set with a unique largest element, and the constraint that for all $uv \in E(G)$, the vertices $u$ and $v$ must be comparable in $\Gamma$ (i.e. $u <_\Gamma v$ or $v <_\Gamma u$). It is represented by a rooted directed tree on $V(G)$, where the root corresponds to this largest element. Due to the constraint, edges of the graph are not allowed to `cross' different branches of the tree layout.

A linear layout $\sigma$ (resp. tree layout $\Gamma$) of a DAG $G$ is \emph{$G$-respecting} if $u <_G v$ implies $u <_\sigma v$ (resp. $u <_\Gamma v$) for all $u,v \in V(G)$. A $G$-respecting linear layout (resp. tree layout) of $G$ is called an \emph{extension} (resp. \emph{tree extension}) of~$G$. Consequently, all arcs of a DAG point in the same direction when drawn in a (tree) extension. For an extension, the arcs point backwards (i.e. towards the first element of the ordering), while in a tree extension, the arcs point away from the root. \cref{fig:layout_example} serves as a visualization of the above concepts. As a convention, we will always draw (tree) extensions as presented in this figure. To avoid confusion, \cref{subfig:layout:notree} shows a graph $H$ that may mistakenly be interpreted as a tree extension.

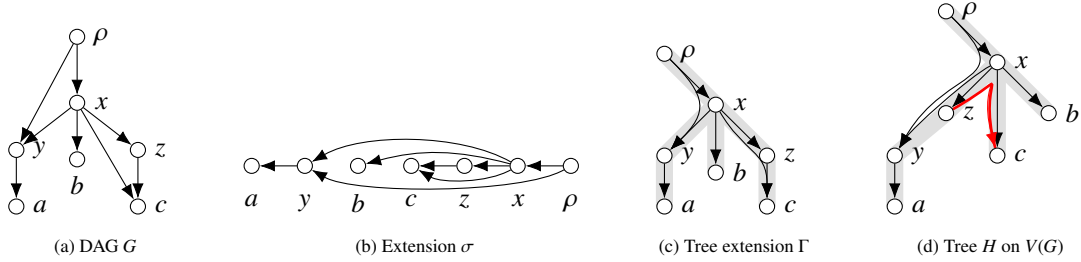
\begin{figure}[htb]
     \centering
     \begin{subfigure}[b]{0.20\textwidth}
         \centering
         \begin{tikzpicture}[xscale=0.4, yscale=0.5]
	\begin{pgfonlayer}{nodelayer}
		\node [style={graph_node}, label=right:{$\rho$}] (51) at (2, 4.5) {};
		\node [style={graph_node}, label=right:{$a$}] (52) at (0, 0) {};
		\node [style={graph_node}, label=below:{$b$}] (53) at (2, 1.25) {};
		\node [style={graph_node}, label=right:{$c$}] (54) at (4, 0) {};
		\node [style={graph_node}, label=right:{$y$}] (55) at (0, 1.5) {};
		\node [style={graph_node}, label=right:{$z$}] (56) at (4, 1.5) {};
		\node [style={graph_node}, label=right:{$x$}] (57) at (2, 2.75) {};
	\end{pgfonlayer}
	\begin{pgfonlayer}{edgelayer}
		\draw [style={graph_edge}] (51) to (57);
		\draw [style={graph_edge}] (51) to (55);
		\draw [style={graph_edge}] (57) to (55);
		\draw [style={graph_edge}] (57) to (56);
		\draw [style={graph_edge}] (57) to (53);
		\draw [style={graph_edge}] (55) to (52);
		\draw [style={graph_edge}] (56) to (54);
		\draw [style={graph_edge}] (57) to (54);
	\end{pgfonlayer}
\end{tikzpicture}
         \caption{DAG $G$}
         \label{subfig:layout:graph}
     \end{subfigure}
     \begin{subfigure}[b]{0.30\textwidth}
         \centering
         \begin{tikzpicture}[scale=0.7]
	\begin{pgfonlayer}{nodelayer}
		\node [style={graph_node}, label={[label distance=1.5mm]below:{$\rho$}}] (65) at (6, 0) {};
		\node [style={graph_node}, label={[label distance=1.5mm]below:{$a$}}] (66) at (0, 0) {};
		\node [style={graph_node}, label={[label distance=1.5mm]below:{$b$}}] (67) at (2, 0) {};
		\node [style={graph_node}, label={[label distance=1.5mm]below:{$c$}}] (68) at (3, 0) {};
		\node [style={graph_node}, label={[label distance=1.5mm]below:{$y$}}] (69) at (1, 0) {};
		\node [style={graph_node}, label={[label distance=1.5mm]below:{$z$}}] (70) at (4, 0) {};
		\node [style={graph_node}, label={[label distance=1.5mm]below:{$x$}}] (71) at (5, 0) {};
	\end{pgfonlayer}
	\begin{pgfonlayer}{edgelayer}
		\draw [style={graph_edge}] (65) to (71);
		\draw [style={graph_edge}, bend left, looseness=0.50] (65) to (69);
		\draw [style={graph_edge}, bend right, looseness=0.75] (71) to (69);
		\draw [style={graph_edge}] (71) to (70);
		\draw [style={graph_edge}, bend right=15] (71) to (67);
		\draw [style={graph_edge}] (69) to (66);
		\draw [style={graph_edge}] (70) to (68);
		\draw [style={graph_edge}, bend left, looseness=0.75] (71) to (68);
	\end{pgfonlayer}
\end{tikzpicture}
         \caption{Extension $\sigma$}
         \label{subfig:layout:ext}
     \end{subfigure}
     \begin{subfigure}[b]{0.20\textwidth}
         \centering
         \begin{tikzpicture}[scale=0.45]
	\begin{pgfonlayer}{nodelayer}
		\node [style={graph_node}, label=right:{$\rho$}] (51) at (0, 4.5) {};
		\node [style={graph_node}, label=right:{$a$}] (52) at (0, 0) {};
		\node [style={graph_node}, label=right:{$b$}] (53) at (1.5, 1) {};
		\node [style={graph_node}, label=right:{$c$}] (54) at (3, 0) {};
		\node [style={graph_node}, label=right:{$y$}] (55) at (0, 1.5) {};
		\node [style={graph_node}, label=right:{$z$}] (56) at (3, 1.5) {};
		\node [style={graph_node}, label=right:{$x$}] (57) at (1.5, 3) {};
	\end{pgfonlayer}
		\draw [style={extension_edge}] (51.center) to (57.center);
		\draw [style={extension_edge}] (57.center) to (53.center);
		\draw [style={extension_edge}] (57.center) to (55.center);
		\draw [style={extension_edge}] (55.center) to (52.center);
		\draw [style={extension_edge}] (57.center) to (56.center);
		\draw [style={extension_edge}] (56.center) to (54.center);
	\begin{pgfonlayer}{edgelayer}
		\draw [style={graph_edge}] (51) to (57);
		\draw [style={graph_edge}, in=45, out=-45, looseness=1.50] (51) to (55);
		\draw [style={graph_edge}] (57) to (55);
		\draw [style={graph_edge}] (57) to (56);
		\draw [style={graph_edge}] (57) to (53);
		\draw [style={graph_edge}] (55) to (52);
		\draw [style={graph_edge}] (56) to (54);
		\draw [style={graph_edge}, in=90, out=-60, looseness=1.25] (57) to (54);
	\end{pgfonlayer}
\end{tikzpicture}
         \caption{Tree extension $\Gamma$}
         \label{subfig:layout:tree}
     \end{subfigure}
     \begin{subfigure}[b]{0.20\textwidth}
         \centering
         \begin{tikzpicture}[scale=0.45]
	\begin{pgfonlayer}{nodelayer}
		\node [style={graph_node}, label=right:{$\rho$}] (58) at (1.25, 6) {};
		\node [style={graph_node}, label=right:{$a$}] (59) at (-0.25, 0.25) {};
		\node [style={graph_node}, label=right:{$b$}] (60) at (4.25, 3) {};
		\node [style={graph_node}, label=right:{$c$}] (61) at (2.75, 1.75) {};
		\node [style={graph_node}, label=right:{$y$}] (62) at (-0.25, 1.75) {};
		\node [style={graph_node}, label=right:{$z$}] (63) at (1.25, 3) {};
		\node [style={graph_node}, label=right:{$x$}] (64) at (2.75, 4.5) {};
	\end{pgfonlayer}
			\draw [style={extension_edge}] (58.center) to (64.center);
		\draw [style={extension_edge}] (64.center) to (63.center);
		\draw [style={extension_edge}] (63.center) to (62.center);
		\draw [style={extension_edge}] (62.center) to (59.center);
		\draw [style={extension_edge}] (64.center) to (61.center);
		\draw [style={extension_edge}] (64.center) to (60.center);
	\begin{pgfonlayer}{edgelayer}
		\draw [style={graph_edge}] (58) to (64);
		\draw [style={graph_edge}, in=60, out=-45, looseness=1.75] (58) to (62);
		\draw [style={graph_edge}, bend right=15, looseness=0.50] (64) to (62);
		\draw [style={graph_edge}] (64) to (63);
		\draw [style={graph_edge}] (64) to (60);
		\draw [style={graph_edge}] (62) to (59);
		\draw [style={graph_edge}] (64) to (61);
		\draw [style={graph_edge},color=red, in=105, out=30, looseness=3.50] (63) to (61);
        \draw [line width=1, color=red, in=105, out=30, looseness=3.50] (63) to (61);
	\end{pgfonlayer}
\end{tikzpicture}
         \caption{Tree $H$ on $V(G)$}
         \label{subfig:layout:notree}
     \end{subfigure}
        \caption{(a): Weakly connected DAG $G$. (b): An extension $\sigma$ of $G$ with the arcs of $G$ also drawn. (c): A tree extension $\Gamma$ of $G$ indicated by the grey arcs, whose direction is downwards. The arcs of $G$ are also drawn in $\Gamma$ and are made to follow the grey arcs. (d): A tree $H$ on $V(G)$ that is not a tree extension, because $z$ and $c$ are not comparable in $H$, while they are adjacent in $G$. Visually this means that the corresponding thick red arc $zc$ of $G$ `crosses' two branches of the tree.}
        \label{fig:layout_example}
\end{figure}

We will now cover some notation, partially adopted from \cite{Berry2022b}, on the above notions. Throughout this paper, we will exclusively reserve the Greek letters $\sigma$ and $\pi$ for extensions, and the Greek capital letters $\Gamma$ and $\Omega$ for tree extensions. For a weakly connected DAG $G$ (resp. a connected undirected graph~$G$) and $U \subseteq V(G)$, we use $\Pi[U]$ to denote the set of extensions (resp. linear layouts) of $G[U]$. Thus, $\Pi[V(G)]$ is the set of all extensions (resp. linear layouts) of $G$. Similarly, we use $\mathcal{T} [U]$ to denote the set of tree extensions (resp. tree layouts) of $G[U]$.

The vertex at position $i$ of a layout $\sigma$ is denoted by $\sigma (i)$. The suborder of $\sigma$ starting at position $i$ till position $j$ is written as $\sigma [i\ldots j]$, while the order starting at $i$ till the last vertex is $\sigma [i \ldots ]$. The restriction of an ordering to a subset $U \subseteq V(G)$ is written as $\sigma[U]$. If $A$ and $B$ are two disjoint subsets of $V(G)$, and $\sigma$ (resp. $\pi$) is a layout of $G[A]$ (resp. $G[B]$), we write $\sigma \circ \pi$ for the concatenation of $\sigma$ and $\pi$ (that is, $\sigma$ followed by $\pi$). Note that the positions and vertices of a linear layout are in bijection. Therefore, we will sometimes treat the vertices as interchangeable with their position in the linear layout. In particular, we will write $\sigma[1 \ldots v]$ to denote $\sigma[1 \ldots \sigma^{-1} (v)]$ and $\sigma [v \ldots]$ to denote $\sigma [\sigma^{-1} (v) \ldots]$. For a vertex set $V$, we also write $[V] = \{ 1, \ldots, |V| \}$. Since subgraphs of the type $G[\sigma[1 \ldots i]]$ will appear throughout this paper, we often denote them as $G[1 \ldots i]$ if the extension $\sigma$ is clear from the context.

\subsection{Cutwidth}
\label{subsec:cutwidth_def}
Following \cite{Berry2022b}, we will first define cutwidth, which will make it easier to explain scanwidth.
Cutwidth is a width parameter for graphs that has seen a lot of attention since the 1970s (see the survey \cite{diaz2002survey} and its addendum \cite{petit2013addenda}). Multiple variants exist, but we focus on the specific version of the parameter for DAGs. There is no consensus on the naming of this DAG-variant of cutwidth. In \cite{bodlaender2012note} it is referred to as `minimum cutwidth for directed acyclic graphs'. Other authors call it `directed cutwidth' \cite{bodlaender2009derivation, Berry2022b}. For the sake of brevity, we will refer to it simply as \emph{cutwidth}.

\begin{definition}[Cutwidth]
\label{def:cutwidth}
Let $G=(V,E)$ be a weakly connected DAG. For an extension $\sigma$ and a position $i$ of $\sigma$, we will denote $\CW_i^\sigma = \{ uv \in E: u \in \sigma[i+1 \ldots], v \in \sigma [1\ldots i] \}$. Then the \emph{cutwidth} of $G$ is 
$\cw (G) = \min_{\sigma \in \Pi[V]} \max_{i \in [V]} |\CW_i^\sigma|.$
Furthermore, we let $\cw (\sigma, G) = \max_{i \in [V]} |\CW_i^\sigma |$ be the cutwidth of $\sigma$, where $|\CW_i^\sigma |$ is the cutwidth of $\sigma$ at position $i$.
\end{definition}

Intuitively, an extension of a DAG is considered optimal in terms of cutwidth if the maximum number of arcs crossing a gap between two vertices is as small as possible. An example of an optimal and a non-optimal extension in terms of cutwidth is shown in \cref{fig:cutwidth_example}. 

\begin{figure}[htb]
     \centering
     \parbox{.31\textwidth}{
     \begin{subfigure}[b]{0.99\linewidth}
         \centering
         \begin{tikzpicture}[scale=0.45]
	\begin{pgfonlayer}{nodelayer}
		\node [style={graph_node}, label=right:{$a$}] (40) at (0, 0) {};
		\node [style={graph_node}, label=right:{$b$}] (41) at (2, 0) {};
		\node [style={graph_node}, label=right:{$c$}] (42) at (4, 0) {};
		\node [style={graph_node}, label=right:{$x$}] (43) at (0, 1.5) {};
		\node [style={graph_node}, label=right:{$y$}] (44) at (2, 1.5) {};
		\node [style={graph_node}, label=right:{$z$}] (45) at (4, 1.5) {};
		\node [style={graph_node}, label=right:{$u$}] (46) at (1, 3) {};
		\node [style={graph_node}, label=right:{$v$}] (47) at (3, 3) {};
		\node [style={graph_node}, label=right:{$q$}] (48) at (2, 4.5) {};
		\node [style={graph_node}, label=right:{$w$}] (49) at (5, 3) {};
		\node [style={graph_node}, label=right:{$\rho$}] (50) at (3, 6) {};
	\end{pgfonlayer}
	\begin{pgfonlayer}{edgelayer}
		\draw [style={graph_edge}] (50) to (48);
		\draw [style={graph_edge}] (48) to (46);
		\draw [style={graph_edge}] (46) to (43);
		\draw [style={graph_edge}] (43) to (40);
		\draw [style={graph_edge}] (48) to (47);
		\draw [style={graph_edge}] (46) to (44);
		\draw [style={graph_edge}] (47) to (44);
		\draw [style={graph_edge}] (44) to (41);
		\draw [style={graph_edge}] (47) to (45);
		\draw [style={graph_edge}] (45) to (42);
		\draw [style={graph_edge}] (50) to (49);
		\draw [style={graph_edge}] (49) to (45);
	\end{pgfonlayer}
\end{tikzpicture}
         \caption{Weakly connected DAG $G$}
         \label{fig:cutwidth_example:graph}
     \end{subfigure}
     }
     \parbox{.67\textwidth}{
     \begin{subfigure}[b]{0.98\linewidth}
         \centering
         \begin{tikzpicture}[xscale=0.75, yscale=0.6]
	\begin{pgfonlayer}{nodelayer}
		\node [style={graph_node}, label={[label distance=1.5mm]below:{$a$}}] (11) at (0, 0) {};
		\node [style={graph_node}, label={[label distance=1.5mm]below:{$b$}}] (12) at (2, 0) {};
		\node [style={graph_node}, label={[label distance=1.5mm]below:{$c$}}] (13) at (5, 0) {};
		\node [style={graph_node}, label={[label distance=1.5mm]below:{$x$}}] (14) at (1, 0) {};
		\node [style={graph_node}, label={[label distance=1.5mm]below:{$y$}}] (15) at (3, 0) {};
		\node [style={graph_node}, label={[label distance=1.5mm]below:{$z$}}] (16) at (6, 0) {};
		\node [style={graph_node}, label={[label distance=1.5mm]below:{$u$}}] (17) at (4, 0) {};
		\node [style={graph_node}, label={[label distance=1.5mm]below:{$v$}}] (18) at (7, 0) {};
		\node [style={graph_node}, label={[label distance=1.5mm]below:{$q$}}] (19) at (8, 0) {};
		\node [style={graph_node}, label={[label distance=1.5mm]below:{$w$}}] (20) at (9, 0) {};
		\node [style={graph_node}, label={[label distance=1.5mm]below:{$\rho$}}] (21) at (10, 0) {};
		\node [style=none] (22) at (6.75, -1) {};
		\node [style=none] (23) at (6.75, 1) {};
	\end{pgfonlayer}
		\draw [style={cut_line}, bend right=15] (23.center) to (22.center);	
	\begin{pgfonlayer}{edgelayer}
		\draw [style={graph_edge}, bend left=330] (21) to (19);
		\draw [style={graph_edge}, in=30, out=150, looseness=0.75] (19) to (17);
		\draw [style={graph_edge}, bend left, looseness=0.50] (17) to (14);
		\draw [style={graph_edge}] (14) to (11);
		\draw [style={graph_edge}] (19) to (18);
		\draw [style={graph_edge}] (17) to (15);
		\draw [style={graph_edge}, bend left=330, looseness=0.75] (18) to (15);
		\draw [style={graph_edge}] (15) to (12);
		\draw [style={graph_edge}] (18) to (16);
		\draw [style={graph_edge}] (16) to (13);
		\draw [style={graph_edge}] (21) to (20);
		\draw [style={graph_edge}, bend left, looseness=0.50] (20) to (16);
	\end{pgfonlayer}
\end{tikzpicture}
         \caption{Optimal extension $\sigma$}
         \label{fig:cutwidth_example:extension_opt}
     \end{subfigure}
     \begin{subfigure}[b]{0.98\linewidth}
         \centering
         \begin{tikzpicture}[xscale=0.75, yscale=0.6]
	\begin{pgfonlayer}{nodelayer}
		\node [style={graph_node}, label={[label distance=1.5mm]below:{$a$}}] (11) at (2, 0) {};
		\node [style={graph_node}, label={[label distance=1.5mm]below:{$b$}}] (12) at (4, 0) {};
		\node [style={graph_node}, label={[label distance=1.5mm]below:{$c$}}] (13) at (0, 0) {};
		\node [style={graph_node}, label={[label distance=1.5mm]below:{$x$}}] (14) at (3, 0) {};
		\node [style={graph_node}, label={[label distance=1.5mm]below:{$y$}}] (15) at (5, 0) {};
		\node [style={graph_node}, label={[label distance=1.5mm]below:{$z$}}] (16) at (1, 0) {};
		\node [style={graph_node}, label={[label distance=1.5mm]below:{$u$}}] (17) at (6, 0) {};
		\node [style={graph_node}, label={[label distance=1.5mm]below:{$v$}}] (18) at (7, 0) {};
		\node [style={graph_node}, label={[label distance=1.5mm]below:{$q$}}] (19) at (8, 0) {};
		\node [style={graph_node}, label={[label distance=1.5mm]below:{$w$}}] (20) at (9, 0) {};
		\node [style={graph_node}, label={[label distance=1.5mm]below:{$\rho$}}] (21) at (10, 0) {};
		\node [style=none] (22) at (5.75, -1) {};
		\node [style=none] (23) at (5.75, 1) {};
	\end{pgfonlayer}
		\draw [style={cut_line}, bend right=15] (23.center) to (22.center);
	\begin{pgfonlayer}{edgelayer}
		\draw [style={graph_edge}, bend left=330] (21) to (19);
		\draw [style={graph_edge}, in=30, out=150, looseness=0.75] (19) to (17);
		\draw [style={graph_edge}, bend left=25, looseness=0.50] (17) to (14);
		\draw [style={graph_edge}] (14) to (11);
		\draw [style={graph_edge}] (19) to (18);
		\draw [style={graph_edge}] (17) to (15);
		\draw [style={graph_edge}, bend left=330, looseness=0.75] (18) to (15);
		\draw [style={graph_edge}] (15) to (12);
		\draw [style={graph_edge}, bend left=15, looseness=0.75] (18) to (16);
		\draw [style={graph_edge}] (16) to (13);
		\draw [style={graph_edge}] (21) to (20);
		\draw [style={graph_edge}, bend right=25, looseness=0.50] (20) to (16);
	\end{pgfonlayer}
\end{tikzpicture}
         \caption{Non-optimal extension $\pi$}
         \label{fig:cutwidth_example:extension_bad}
     \end{subfigure}
     }
        \caption{(a): Weakly connected DAG $G$. (b): An optimal extension $\sigma$ of $G$ with cutwidth 4, attained at the red cut. (c): A non-optimal extension $\pi$ of $G$ with cutwidth 5, attained at the red cut.}
        \label{fig:cutwidth_example}
\end{figure}
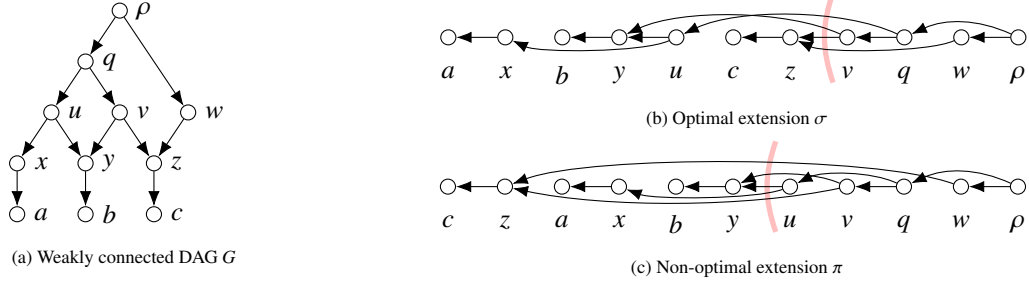

It is implictly shown in \cite{Berry2022b} that computing the cutwidth of a DAG is NP-hard. Regarding exact algorithms, one can compute the cutwidth of a DAG in $\tilde{O} (2^n)$ time \cite{bodlaender2012note}. The theoretically fastest parametrized algorithm for the cutwidth of a DAG was introduced in \cite{bodlaender2009derivation}, and it runs in linear time for fixed cutwidth of $k$, showing that the problem is FPT.

\subsection{Scanwidth}
\label{subsec:scanwidth_def}
In the introductory section of this paper, we provided some intuition on the idea of `scanning' a phylogenetic network. This concept can be extended to all weakly connected DAGs. Throughout this paper we aim to provide results for this broader class of graphs.\footnote{Most results in this paper, if not all, generalize to disconnected graphs by considering a `forest extension', which contains a separate tree extension for each component of the graph.} With cutwidth in mind, we are ready to define scanwidth. The scanwidth of a DAG was formally introduced by Berry, Scornavacca, and Weller \cite{Berry2022b} as follows.
\begin{definition}[Scanwidth]
\label{def:scanwidth-ext}
Let $G=(V,E)$ be a weakly connected DAG. For an extension $\sigma$ and a position $i$ of $\sigma$, we will denote $\SW_i^\sigma (G) = \{ uv \in E: u \in \sigma[i+1 \ldots], v \in \sigma[1 \ldots i], v \connect{G[1\ldots i]} \sigma(i) \}$. Then the \emph{scanwidth of $G$} is 
$$\sw (G) = \min_{\sigma \in \Pi[V]} \max_{i \in [V]} |\SW_i^\sigma(G)|.$$
Furthermore, we let $\sw (\sigma, G) = \max_{i \in [V]} |\SW_i^\sigma (G)|$ be the scanwidth of $\sigma$, where $|\SW_i^\sigma (G)|$ is the scanwidth of $\sigma$ at position $i$. If $G$ is clear from the context, we mostly write $\SW_i^\sigma$ instead of $\SW_i^\sigma (G)$.
\end{definition}

This definition is closely related to the definition of cutwidth. However, instead of counting all arcs in the cut-set $\CW_i^\sigma$, we only count those arcs entering a vertex $v$ that is weakly connected to $\sigma (i)$ in the graph $G[\sigma[1 \ldots i]]$. Recall that we write this as $v \connect{G[1 \ldots i]} \sigma(i)$. Before illustrating this definition with \cref{fig:scanwidth_example}, we introduce an alternative characterization of scanwidth.

The previous definition involves extensions and will turn out the be more convenient in proofs, as induction-based proofs are often easier when iterating over the positions of an extension. In \cite[Prop.\,1]{Berry2022b} it is shown that the next definition equivalently defines scanwidth.\footnote{Berry, Scornavacca, and Weller prove the equivalence only for networks, but the proof does not rely on labelled leaves and a single root. We can therefore extend the equivalence to arbitrary weakly connected DAGs.} This alternative definition relies on tree extensions, thus aligning more closely with the `scanning'-intuition given in the introduction.

\begin{definition}[Scanwidth]
\label{def:scanwidth-tree}
Let $G=(V,E)$ be a weakly connected DAG. For a tree extension $\Gamma$ and a vertex $v$ of $V$, we will denote $\GW_v^\Gamma (G) = \{ xy \in E: x >_\Gamma v \geq_\Gamma y\}$. Then the \emph{scanwidth} of $G$ is 
$$\sw (G) = \min_{\Gamma \in \mathcal{T} (V)} \max_{v \in V} |\GW_v^\Gamma (G)|.$$
Furthermore, we let $\sw (\Gamma, G) = \max_{v \in V} |\GW_v^\Gamma (G)|$ be the scanwidth of $\Gamma$, where $|\GW_v^\Gamma (G)|$ is the scanwidth of $\Gamma$ at vertex $v$. If $G$ is clear from the context, we mostly write $\GW_v^\Gamma$ instead of $\GW_v^\Gamma (G)$.
\end{definition}

To illustrate the two definitions and their relation, we take a look at \cref{fig:scanwidth_example}. \cref{fig:scanwidth_example:graph} depicts the same weakly connected DAG as in \cref{fig:cutwidth_example}, while \cref{fig:scanwidth_example:canonical} shows a tree extension $\Gamma^\sigma$ (the thick grey arcs, whose direction is downwards) with the arcs of the original DAG drawn in it. For each vertex $v$ in the graph, the set $\GW_v^{\Gamma^\sigma}$ contains the arcs that enter the vertex or pass it to reach a vertex lower in the tree extension. Visually, these sets correspond to cuts in the tree extension. As mentioned in the introduction, scanwidth can thus be viewed as a tree analogue of cutwidth. One can quickly check that the scanwidth of this (optimal) tree extension $\Gamma^\sigma$ equals 3, which is attained at the vertex~$v$, where we have that $\GW_v^{\Gamma^\sigma} = \{qv, qu, wz \}$. In a similar fashion, \cref{fig:scanwidth_example:noncanonical} shows a tree extension of the same graph, with a non-optimal scanwidth of 4, attained at the vertex $z$.

\begin{figure}[htb]
     \centering
     \begin{subfigure}[b]{0.31\textwidth}
         \centering
         \begin{tikzpicture}[scale=0.45]
	\begin{pgfonlayer}{nodelayer}
		\node [style={graph_node}, label=right:{$a$}] (40) at (0, 0) {};
		\node [style={graph_node}, label=right:{$b$}] (41) at (2, 0) {};
		\node [style={graph_node}, label=right:{$c$}] (42) at (4, 0) {};
		\node [style={graph_node}, label=right:{$x$}] (43) at (0, 1.5) {};
		\node [style={graph_node}, label=right:{$y$}] (44) at (2, 1.5) {};
		\node [style={graph_node}, label=right:{$z$}] (45) at (4, 1.5) {};
		\node [style={graph_node}, label=right:{$u$}] (46) at (1, 3) {};
		\node [style={graph_node}, label=right:{$v$}] (47) at (3, 3) {};
		\node [style={graph_node}, label=right:{$q$}] (48) at (2, 4.5) {};
		\node [style={graph_node}, label=right:{$w$}] (49) at (5, 3) {};
		\node [style={graph_node}, label=right:{$\rho$}] (50) at (3, 6) {};
	\end{pgfonlayer}
	\begin{pgfonlayer}{edgelayer}
		\draw [style={graph_edge}] (50) to (48);
		\draw [style={graph_edge}] (48) to (46);
		\draw [style={graph_edge}] (46) to (43);
		\draw [style={graph_edge}] (43) to (40);
		\draw [style={graph_edge}] (48) to (47);
		\draw [style={graph_edge}] (46) to (44);
		\draw [style={graph_edge}] (47) to (44);
		\draw [style={graph_edge}] (44) to (41);
		\draw [style={graph_edge}] (47) to (45);
		\draw [style={graph_edge}] (45) to (42);
		\draw [style={graph_edge}] (50) to (49);
		\draw [style={graph_edge}] (49) to (45);
	\end{pgfonlayer}
\end{tikzpicture}
         \caption{Weakly connected DAG $G$}
         \label{fig:scanwidth_example:graph}
     \end{subfigure}
     \begin{subfigure}[b]{0.31\textwidth}
         \centering
         \begin{tikzpicture}[xscale=0.45, yscale=0.4]
	\begin{pgfonlayer}{nodelayer}
		\node [style={graph_node}, label=right:{$a$}] (0) at (0.5, 2) {};
		\node [style={graph_node}, label=right:{$b$}] (1) at (3.5, 2) {};
		\node [style={graph_node}, label=right:{$c$}] (2) at (6.5, 3.5) {};
		\node [style={graph_node}, label=right:{$x$}] (3) at (0.5, 3.5) {};
		\node [style={graph_node}, label=right:{$y$}] (4) at (3.5, 3.5) {};
		\node [style={graph_node}, label=right:{$z$}] (5) at (5, 4.5) {};
		\node [style={graph_node}, label=above left:{$u$}] (6) at (2, 4.5) {};
		\node [style={graph_node}, label=above left:{$v$}] (7) at (3.5, 5.5) {};
		\node [style={graph_node}, label=above left:{$q$}] (8) at (5, 6.5) {};
		\node [style={graph_node}, label=right:{$w$}] (9) at (6.5, 7.5) {};
		\node [style={graph_node}, label=above right:{$\rho$}] (10) at (5, 8.5) {};
	\end{pgfonlayer}
		\draw [style={extension_edge}] (10.center) to (9.center);
		\draw [style={extension_edge}] (9.center) to (8.center);
		\draw [style={extension_edge}] (8.center) to (7.center);
		\draw [style={extension_edge}] (7.center) to (5.center);
		\draw [style={extension_edge}] (5.center) to (2.center);
		\draw [style={extension_edge}] (4.center) to (1.center);
		\draw [style={extension_edge}] (3.center) to (0.center);
		\draw [style={extension_edge}] (6.center) to (3.center);
		\draw [style={extension_edge}] (7.center) to (6.center);
		\draw [style={extension_edge}] (6.center) to (4.center);	
	\begin{pgfonlayer}{edgelayer}
		\draw [style={graph_edge}, bend left=45, looseness=2.50] (10) to (8);
		\draw [style={graph_edge}, bend right=15] (8) to (6);
		\draw [style={graph_edge}] (6) to (3);
		\draw [style={graph_edge}] (3) to (0);
		\draw [style={graph_edge}] (8) to (7);
		\draw [style={graph_edge}] (6) to (4);
		\draw [style={graph_edge}, bend right=45, looseness=2.50] (7) to (4);
		\draw [style={graph_edge}] (4) to (1);
		\draw [style={graph_edge}] (7) to (5);
		\draw [style={graph_edge}] (5) to (2);
		\draw [style={graph_edge}] (10) to (9);
		\draw [style={graph_edge}, in=140, out=-120, looseness=2.00] (9) to (5);
	\end{pgfonlayer}
\end{tikzpicture}
         \caption{Canonical tree extension $\Gamma^\sigma$}
         \label{fig:scanwidth_example:canonical}
     \end{subfigure}
     \begin{subfigure}[b]{0.31\textwidth}
         \centering
         \begin{tikzpicture}[xscale=0.45, yscale=0.4]
	\begin{pgfonlayer}{nodelayer}
		\node [style={graph_node}, label=right:{$a$}] (0) at (-1, 1) {};
		\node [style={graph_node}, label=right:{$b$}] (1) at (2, 1) {};
		\node [style={graph_node}, label=right:{$c$}] (2) at (3.5, 3.5) {};
		\node [style={graph_node}, label=right:{$x$}] (3) at (-1, 2.5) {};
		\node [style={graph_node}, label=right:{$y$}] (4) at (2, 2.5) {};
		\node [style={graph_node}, label=above left:{$z$}] (5) at (2, 4.5) {};
		\node [style={graph_node}, label=above left:{$u$}] (6) at (0.5, 3.5) {};
		\node [style={graph_node}, label=above left:{$v$}] (7) at (3.5, 5.5) {};
		\node [style={graph_node}, label=above left:{$q$}] (8) at (5, 6.5) {};
		\node [style={graph_node}, label=right:{$w$}] (9) at (6.5, 7.5) {};
		\node [style={graph_node}, label=above right:{$\rho$}] (10) at (5, 8.5) {};
		\node [style=none] (11) at (1, 3.75) {};
	\end{pgfonlayer}
		\draw [style={extension_edge}] (10.center) to (9.center);
		\draw [style={extension_edge}] (9.center) to (8.center);
		\draw [style={extension_edge}] (8.center) to (7.center);
		\draw [style={extension_edge}] (7.center) to (5.center);
		\draw [style={extension_edge}] (5.center) to (2.center);
		\draw [style={extension_edge}] (5.center) to (6.center);
		\draw [style={extension_edge}] (6.center) to (3.center);
		\draw [style={extension_edge}] (3.center) to (0.center);
		\draw [style={extension_edge}] (4.center) to (1.center);
		\draw [style={extension_edge}] (6.center) to (4.center);
	\begin{pgfonlayer}{edgelayer}
		\draw [style={graph_edge}, bend left=45, looseness=2.50] (10) to (8);
		\draw [style={graph_edge}, bend right, looseness=0.50] (8) to (6);
		\draw [style={graph_edge}] (6) to (3);
		\draw [style={graph_edge}] (3) to (0);
		\draw [style={graph_edge}] (8) to (7);
		\draw [style={graph_edge}] (6) to (4);
		\draw [style={graph_edge}] (4) to (1);
		\draw [style={graph_edge}] (7) to (5);
		\draw [style={graph_edge}] (5) to (2);
		\draw [style={graph_edge}] (10) to (9);
		\draw [style={graph_edge}, bend left=15, looseness=0.75] (9) to (5);
		\draw [in=105, out=-165, looseness=0.50] (7) to (11.center);
		\draw [style={graph_edge}, in=120, out=-90, looseness=0.75] (11.center) to (4);
	\end{pgfonlayer}
\end{tikzpicture}
         \caption{Tree extension $\Gamma'$}
         \label{fig:scanwidth_example:noncanonical}
     \end{subfigure}
     \vfill
     \begin{subfigure}[b]{0.66\textwidth}
         \centering
         \begin{tikzpicture}[scale=0.75]
	\begin{pgfonlayer}{nodelayer}
	\node [style={none}] (99) at (0, 1) {};
		\node [style={graph_node}, label={[label distance=2.25mm]below:{$a$}}] (11) at (0, 0) {};
		\node [style={graph_node}, label={[label distance=2.25mm]below:{$b$}}] (12) at (2, 0) {};
		\node [style={graph_node}, label={[label distance=2.25mm]below:{$c$}}] (13) at (5, 0) {};
		\node [style={graph_node}, label={[label distance=2.25mm]below:{$x$}}] (14) at (1, 0) {};
		\node [style={graph_node}, label={[label distance=2.25mm]below:{$y$}}] (15) at (3, 0) {};
		\node [style={graph_node}, label={[label distance=2.25mm]below:{$z$}}] (16) at (6, 0) {};
		\node [style={graph_node}, label={[label distance=2.25mm]below:{$u$}}] (17) at (4, 0) {};
		\node [style={graph_node}, label={[label distance=2.25mm]below:{$v$}}] (18) at (7, 0) {};
		\node [style={graph_node}, label={[label distance=1.25mm]below:{$q$}}] (19) at (8, 0) {};
		\node [style={graph_node}, label={[label distance=1.25mm]below:{$w$}}] (20) at (9, 0) {};
		\node [style={graph_node}, label={[label distance=1.25mm]below:{$\rho$}}] (21) at (10, 0) {};
	\end{pgfonlayer}
	\begin{pgfonlayer}{edgelayer}
		\draw [style={graph_edge}, bend left=330] (21) to (19);
		\draw [style={graph_edge}, in=30, out=150, looseness=0.75] (19) to (17);
		\draw [style={graph_edge}, bend left, looseness=0.50] (17) to (14);
		\draw [style={graph_edge}] (14) to (11);
		\draw [style={graph_edge}] (19) to (18);
		\draw [style={graph_edge}] (17) to (15);
		\draw [style={graph_edge}, bend left=330, looseness=0.75] (18) to (15);
		\draw [style={graph_edge}] (15) to (12);
		\draw [style={graph_edge}] (18) to (16);
		\draw [style={graph_edge}] (16) to (13);
		\draw [style={graph_edge}] (21) to (20);
		\draw [style={graph_edge}, bend left, looseness=0.50] (20) to (16);
	\end{pgfonlayer}
\draw[black,fill=gray!50,rounded corners=2mm,opacity=0.35] ($(11.north west)+(-0.2,0.15)$) -- ($(11.north east)+(0.3,0.15)$) -- ($(11.south east)+(0.3,-0.15)$) -- ($(11.south west)+(-0.2,-0.15)$) -- cycle;
\draw[black,fill=gray!70,rounded corners=2mm,opacity=0.35] ($(11.north west)+(-0.3,0.25)$) -- ($(14.north east)+(0.3,0.25)$) -- ($(14.south east)+(0.3,-0.25)$) -- ($(11.south west)+(-0.3,-0.25)$) -- cycle;

\draw[black,fill=gray!50,rounded corners=2mm,opacity=0.35] ($(12.north west)+(-0.2,0.15)$) -- ($(12.north east)+(0.3,0.15)$) -- ($(12.south east)+(0.3,-0.15)$) -- ($(12.south west)+(-0.2,-0.15)$) -- cycle;
\draw[black,fill=gray!70,rounded corners=2mm,opacity=0.35] ($(12.north west)+(-0.3,0.25)$) -- ($(15.north east)+(0.3,0.25)$) -- ($(15.south east)+(0.3,-0.25)$) -- ($(12.south west)+(-0.3,-0.25)$) -- cycle;

\draw[black,fill=gray!70,rounded corners=2mm,opacity=0.35] ($(11.north west)+(-0.4,0.35)$) -- ($(17.north east)+(0.3,0.35)$) -- ($(17.south east)+(0.3,-0.35)$) -- ($(11.south west)+(-0.4,-0.35)$) -- cycle;

\draw[black,fill=gray!50,rounded corners=2mm,opacity=0.35] ($(13.north west)+(-0.2,0.15)$) -- ($(13.north east)+(0.3,0.15)$) -- ($(13.south east)+(0.3,-0.15)$) -- ($(13.south west)+(-0.2,-0.15)$) -- cycle;
\draw[black,fill=gray!70,rounded corners=2mm,opacity=0.35] ($(13.north west)+(-0.3,0.25)$) -- ($(16.north east)+(0.3,0.25)$) -- ($(16.south east)+(0.3,-0.25)$) -- ($(13.south west)+(-0.3,-0.25)$) -- cycle;

\draw[black,fill=gray!30,rounded corners=2mm,opacity=0.25] ($(11.north west)+(-0.5,0.45)$) -- ($(18.north east)+(0.3,0.45)$) -- ($(18.south east)+(0.3,-0.45)$) -- ($(11.south west)+(-0.5,-0.45)$) -- cycle;
\end{tikzpicture}
         \caption{Optimal extension $\sigma$}
         \label{fig:scanwidth_example:extension}
     \end{subfigure}
        \caption{(a): Weakly connected, rooted DAG $G$. (b): Optimal canonical tree extension $\Gamma^\sigma$ with scanwidth 3, attained at the vertex $v$. (c): Non-canonical tree extension $\Gamma'$ with scanwidth 4, attained at the vertex $z$. (d): Optimal extension $\sigma$ with scanwidth 3, attained at the vertex $v$. For each $i\leq 8$, the outermost grey shaded areas containing only vertices belonging to $\sigma [1 \ldots i]$ depict the weakly connected components of $G[1 \ldots i]$. For $i \geq 8$, $G[1 \ldots i]$ is weakly connected and therefore consists of just one component.}
        \label{fig:scanwidth_example}
\end{figure}
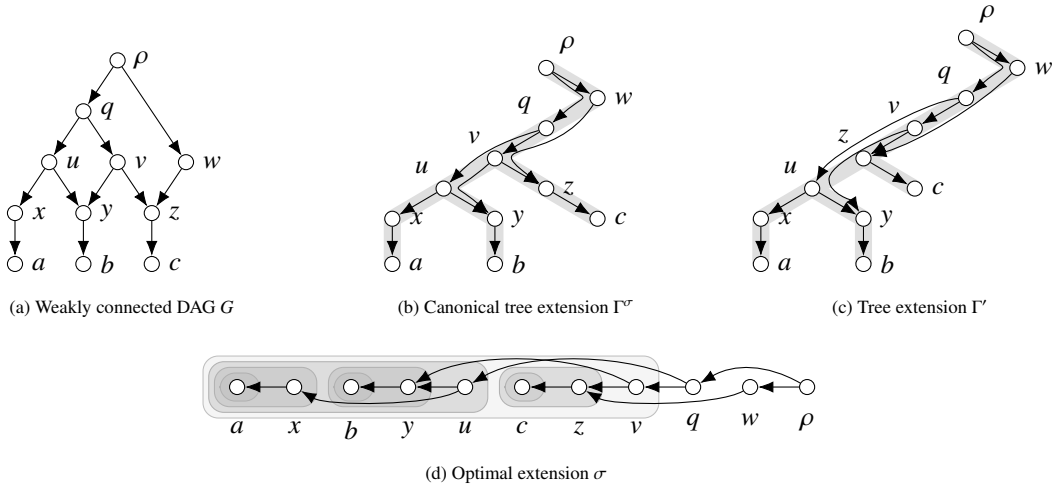

We just stated that the scanwidth of the graph $G$ in \cref{fig:scanwidth_example:graph} equals 3, meaning that there also exists an extension with a scanwidth of 3 when considering \cref{def:scanwidth-ext}. One such extension $\sigma$ is shown in \cref{fig:scanwidth_example:extension}. In comparison to tree extensions, testing membership of an arc in one of the sets~$\SW$ is slightly more involved. It requires knowledge of the weak connectivity relations within subgraphs of the graph. In \cref{fig:scanwidth_example:extension} these relations are depicted by the grey shaded areas. Consider for example the vertex $z =\sigma (7)$. From the figure we can see that $\{a,x,b,y,u \}$ and $\{c, z \}$ form the two components of $G[1 \ldots 7]$. Therefore, $\SW_z^\sigma = \{vz, wz \}$, and the set does not contain the arcs $vy$ and $qu$, as they enter the other component. Compare this to the cutwidth of this extension, where these two arcs would also be counted (see \cref{fig:cutwidth_example:extension_opt}).

Upon further examination of the example in \cref{fig:scanwidth_example}, an interesting observation emerges. Note that the sets $\GW$ associated with the tree extension $\Gamma^\sigma$ and the sets $\SW$ associated with the extension $\sigma$ coincide at each vertex. This is not coincidental, since $\Gamma^\sigma$ is the \emph{canonical tree extension} for $\sigma$, a crucial notion from \cite{Berry2022b} used to prove the equivalence between the two scanwidth definitions. First, recall that the \emph{transitive reduction} of $G$ is another graph $H$ on the same vertex set and with as few arcs as possible, such that for all pairs of vertices $u,v$, there exists a directed path from $u$ to $v$ in $G$ if and only if there exists a directed path from $u$ to $v$ in $H$. It is a well-known fact that the transitive reduction of a DAG is unique and can be obtained by exhaustively deleting arcs $uv$ for which there is a directed path from $u$ to $v$ containing at least one other vertex. Formally, the canonical tree extension is now defined as:

\begin{definition}[Canonical tree extension]
\label{def:canonical_tree}
Let $G=(V,E)$ be a weakly connected DAG and $\sigma$ an extension of $G$. Then, we denote the \emph{canonical tree extension} for $\sigma$ as $\Gamma^\sigma$, and it is defined as the transitive reduction of the DAG
$H = (V, \{uv : u >_\sigma v , u \connect{G[1 \ldots u]} v\} ). $  
\end{definition}

Lemma 5 from \cite{Berry2022b} establishes the relation between canonical tree extensions and extensions. It shows that $\Gamma^\sigma$ is indeed a tree extension of $G$ that has the same scanwidth as $\sigma$. Moreover, any extension of the canonical tree extension has the same scanwidth again. Note that such an extension can easily be found in linear time with a breadth-first-search, given that the tree extension is canonical.

Using Lemma 5 from \cite{Berry2022b} and the fact that the sets $\GW_v^\Gamma$ uniquely determine the tree extension $\Gamma$ (\cref{lem:GW_unique}), it is possible to completely characterize the canonical tree extensions by a more easily checkable condition. The proof of the result is deferred to \cref{sec:omitted_proofs}. Recall that for a tree extension~$\Gamma$ of~$G$ and a vertex~$v$, the graph $G[V(\Gamma_v)]$ is the subgraph of~$G$ induced by all vertices in the subtree of~$\Gamma$ rooted at~$v$.

\begin{restatable}{proposition}{propcanon}
\label{prop:canonical_tree_iff}
Let $G=(V,E)$ be a weakly connected DAG, $\Gamma$ a tree extension of $G$, and $\sigma$ an extension of $\Gamma$. For each $v\in V$, let $\Gamma_v$ be the subtree of $\Gamma$ rooted at $v$. Then, $\Gamma$ is the canonical tree extension for $\sigma$ if and only if $G[V(\Gamma_v)]$ is weakly connected for all $v\in V$.
\end{restatable}

Recalling the examples from \cref{fig:scanwidth_example}, we can now quickly deduce that $\Gamma^\sigma$ in \cref{fig:scanwidth_example:canonical} is indeed canonical for $\sigma$. On the other hand, the tree extension $\Gamma'$ from \cref{fig:scanwidth_example:noncanonical} can now be shown to not be canonical.

One might be interested in the canonical tree extension, while only an extension is at hand. The previous proposition enables us to perform this task in quadratic time. The algorithm (formally described in \cite{holtgrefe2023scanwidth}) directly builds up the tree extension from its leaves to the root. Specifically, the algorithm makes sure that the extension is also an extension of the built-up tree extension, and furthermore, that the subgraphs $G[V(\Gamma_v )]$ are weakly connected.

\subsection[Treewidth is less than or equal to scanwidth, which is less than or equal to the level + 1]{Treewidth $\leq$ scanwidth $\leq$ level + 1}\label{subsec:bounds}

The proofs of the two bounds presented in this subsection are deferred to \cref{sec:omitted_proofs}, since these results are not vital to the other parts of this paper. However, the bounds do paint a nice picture of how the scanwidth fits within the landscape of other graph parameters. Note that it immediately follows from the definitions that scanwidth is bounded from above by cutwidth. It is also trivial to prove that scanwidth is bounded from below by the very similar \emph{edge-treewidth}, as introduced in \cite{magne2023edge}.

\paragraph{Treewidth} In the introduction, the use of scanwidth as a parameter in algorithms was motivated by the successful applicability of another tree measure: \emph{treewidth}. Although normally defined by a so-called \emph{tree-decomposition} (see e.g. \cite{diestel2017graph}), we will use a different - yet equivalent - formulation for treewidth. This allows us to relate treewidth to scanwidth in a more straightforward manner. The formulation we use is adapted from \cite{Scornavacca2022}.

\begin{definition}[Treewidth]
\label{def:treewidth}
Let $G=(V,E)$ be a connected undirected graph. For a tree layout $\Gamma$ and a vertex $v$ of $V$, we will denote $\TW_v^\Gamma (G) = \{ u \in V:  u >_\Gamma v, \exists w \leq_\Gamma v \text{ s.t. } uw \in E \}$. Then the \emph{treewidth} of $G$ is $\tw (G) = \min_{\Gamma \in \mathcal{T} (V)} \max_{v \in V} |\TW_v^\Gamma (G)|$. Furthermore, we denote $\tw (\Gamma, G) = \max_{v \in V} |\TW_v^\Gamma (G)|$.
\end{definition}

Recall that $\mathcal{T} (V)$ is the set of all tree layouts of $G$, since $G$ is an undirected graph. To exemplify the relation between treewidth and scanwidth, consider the canonical tree extension $\Gamma^\sigma$ in \cref{fig:scanwidth_example:canonical}, where $\GW^{\Gamma^\sigma}_v \!(G) = \{wz, qu, qv \}$. For the treewidth, we can disregard the directions, and the set $\TW^{\Gamma^\sigma}_v \!(G)$ now contains only the endpoints in $\GW^{\Gamma^\sigma}_v \!(G) $ that are higher up in the tree than $v$. Therefore, $\TW^{\Gamma^\sigma}_v \!(G) = \{w, q \}$.

In \cite{Berry2022b} it is mentioned without proof that the treewidth of the underlying undirected graph of a DAG lower bounds its scanwidth.\footnote{It was actually stated the other way round (and also mistakenly stated that cutwidth bounds scanwidth from below). However, one of the authors confirmed that it was intended as expressed here (Mathias Weller, personal communication).} This fact is far from obvious when looking at the common definition of the treewidth and heavily relies on the uncommon alternative definition we have given here. As this definition is not referred to in \cite{Berry2022b}, we formally state the result here and provide a proof in the appendix. The proof relies on the fact that a tree extension of a DAG is also a valid tree layout of its underlying undirected graph. As a side note, we remark that another well-known parameter, `tree-depth', is also defined via such tree layouts~\cite{nevsetvril2006tree}; hence a tree extension is also a valid tree-depth decomposition, albeit optimizing a different objective.

\begin{restatable}{lemma}{twbound}
\label{lem:treewidth-bound}
Let $G=(V,E)$ be a weakly connected DAG and $\tilde{G}$ its underlying undirected graph, then $$ \tw (\tilde{G} ) \leq \sw (G)  .$$
\end{restatable}

\paragraph{Level of a network}

Using that indegrees of sinksets of a network are at most the reticulation number + 1 (\cref{lem:sinkset_reticulation_nr}) and that the scanwidth of a network is the maximum scanwidth of its blocks (\cref{cor:split_blocks}), we can obtain the following result:

\begin{restatable}{lemma}{levelbound}
\label{lem:level_bound}
Let $G=(V,E)$ be a level-$k$ network, then $$\sw(G) \leq k+1.$$
\end{restatable}

The fact that the scanwidth of a binary network is at most its level + 1 has already been stated without proof in \cite{Berry2022b} and is proved in the appendix of \cite{rabier2021inference} in a different setting. Our more graph-theoretical proof is self-contained and also holds for non-binary networks.

The above bound is certainly not tight in general. Consider for example the network in \cref{fig:level}, which is a variation of a network from \cite{rabier2021inference}. This network always has a scanwidth of 3 but can be extended to have an arbitrarily large level.

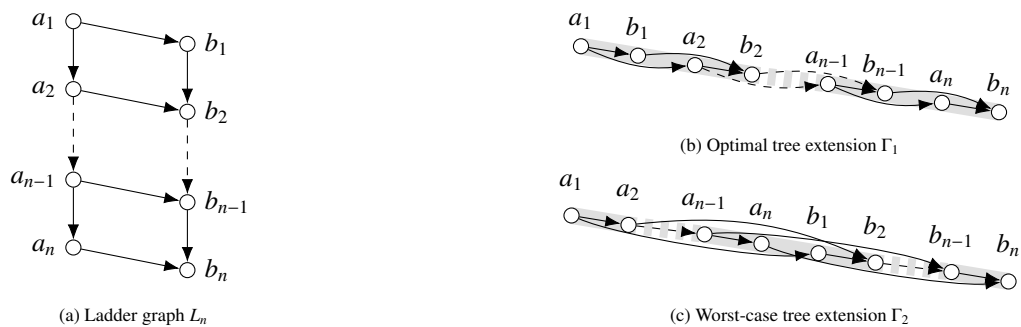
\begin{figure}[htb]
\centering
\valign{#\cr
  \hsize=0.35\columnwidth
  \begin{subfigure}{0.35\columnwidth}
  \centering
  \begin{tikzpicture}[scale=0.6]
	\begin{pgfonlayer}{nodelayer}
		\node [style={graph_node},label=left:{$a_1$}] (1) at (0, 5) {};
		\node [style={graph_node},label=right:{$b_1$}] (2) at (2.5, 4.5) {};
		\node [style={graph_node},label=right:{$b_2$}] (3) at (2.5, 3) {};
		\node [style={graph_node},label=left:{$a_2$}] (4) at (0, 3.5) {};
		\node [style={graph_node},label=left:{$a_{n-1}$}] (5) at (0, 1.5) {};
		\node [style={graph_node},label=right:{$b_{n-1}$}] (6) at (2.5, 1) {};
		\node [style={graph_node},label=right:{$b_n$}] (7) at (2.5, -0.5) {};
		\node [style={graph_node},label=left:{$a_n$}] (8) at (0, 0) {};
	\end{pgfonlayer}
	\begin{pgfonlayer}{edgelayer}
		\draw [style={graph_edge}] (2) to (3);
		\draw [style={graph_edge}] (1) to (4);
		\draw [style={graph_edge}] (1) to (2);
		\draw [style={graph_edge}] (4) to (3);
		\draw [style={graph_edge}] (6) to (7);
		\draw [style={graph_edge}] (5) to (8);
		\draw [style={graph_edge}] (8) to (7);
		\draw [style={graph_edge}] (5) to (6);
		\draw [style={graph_edge_dotted}] (3) to (6);
		\draw [style={graph_edge_dotted}] (4) to (5);
	\end{pgfonlayer}
\end{tikzpicture}
  \caption{Ladder graph $L_n$}
  \end{subfigure}
  
  \cr\noalign{\hfill}
  \hsize=0.6\columnwidth
  
  \begin{subfigure}{0.6\columnwidth}
  \centering
  \begin{tikzpicture}[scale=0.5]
	\begin{pgfonlayer}{nodelayer}
		\node [style={graph_node},label=above:{$a_1$}] (1) at (1.5, -0.25) {};
		\node [style={graph_node},label=above:{$b_1$}] (2) at (3, -0.5) {};
		\node [style={graph_node},label=above:{$b_2$}] (3) at (6, -1) {};
		\node [style={graph_node},label=above:{$a_2$}] (4) at (4.5, -0.75) {};
		\node [style={graph_node},label=above:{$a_{n-1}$}] (5) at (8, -1.25) {};
		\node [style={graph_node},label=above:{$b_{n-1}$}] (6) at (9.5, -1.5) {};
		\node [style={graph_node},label=above:{$b_n$}] (7) at (12.5, -2) {};
		\node [style={graph_node},label=above:{$a_n$}] (8) at (11, -1.75) {};
	\end{pgfonlayer}
	\begin{pgfonlayer}{edgelayer}
		\draw [style={graph_edge}, bend left=15] (2) to (3);
		\draw [style={graph_edge}, bend right=15] (1) to (4);
		\draw [style={graph_edge}] (1) to (2);
		\draw [style={graph_edge}] (4) to (3);
		\draw [style={graph_edge}, bend left=15] (6) to (7);
		\draw [style={graph_edge}, bend right=15] (5) to (8);
		\draw [style={graph_edge}] (8) to (7);
		\draw [style={graph_edge}] (5) to (6);
		\draw [style={graph_edge_dotted}, bend left=15] (3) to (6);
		\draw [style={graph_edge_dotted}, bend right=15] (4) to (5);
	\end{pgfonlayer}		
		\draw [style={extension_edge}] (1.center) to (2.center);
		\draw [style={extension_edge}] (2.center) to (4.center);
		\draw [style={extension_edge}] (4.center) to (3.center);
		\draw [style={extension_edge}] (5.center) to (6.center);
		\draw [style={extension_edge}] (6.center) to (8.center);
		\draw [style={extension_edge_dashed}] (3.center) to (5.center);
		\draw [style={extension_edge}] (8.center) to (7.center);

\end{tikzpicture}
  \caption{Optimal tree extension $\Gamma_1$}
  \end{subfigure}
  
  \vfill
  
  \begin{subfigure}{0.6\columnwidth}
  \centering
  \begin{tikzpicture}[scale=0.5,label distance=3]
	\begin{pgfonlayer}{nodelayer}
		\node [style={graph_node},label=above:{$a_1$}] (1) at (1.5, -0.25) {};
		\node [style={graph_node},label=above:{$b_1$}] (2) at (8, -1.25) {};
		\node [style={graph_node},label=above:{$b_2$}] (3) at (9.5, -1.5) {};
		\node [style={graph_node},label=above:{$a_2$}] (4) at (3, -0.5) {};
		\node [style={graph_node},label=above:{$a_{n-1}$}] (5) at (5, -0.75) {};
		\node [style={graph_node},label=above:{$b_{n-1}$}] (6) at (11.5, -1.75) {};
		\node [style={graph_node},label=above:{$b_n$}] (7) at (13, -2) {};
		\node [style={graph_node},label=above:{$a_n$}] (8) at (6.5, -1) {};
	\end{pgfonlayer}
	\begin{pgfonlayer}{edgelayer}
		\draw [style={graph_edge}] (2) to (3);
		\draw [style={graph_edge}] (1) to (4);
		\draw [style={graph_edge}, bend right=15, looseness=0.50] (1) to (2);
		\draw [style={graph_edge}, bend left=15] (4) to (3);
		\draw [style={graph_edge}] (6) to (7);
		\draw [style={graph_edge}] (5) to (8);
		\draw [style={graph_edge}, bend right=15, looseness=0.50] (8) to (7);
		\draw [style={graph_edge}, bend left=15, looseness=0.50] (5) to (6);
		\draw [style={graph_edge_dotted}] (3) to (6);
		\draw [style={graph_edge_dotted}] (4) to (5);
	\end{pgfonlayer}		
		\draw [style={extension_edge}] (1.center) to (4.center);
		\draw [style={extension_edge}] (5.center) to (8.center);
		\draw [style={extension_edge}] (8.center) to (2.center);
		\draw [style={extension_edge}] (2.center) to (3.center);
		\draw [style={extension_edge}] (6.center) to (7.center);
		\draw [style={extension_edge_dashed}] (4.center) to (5.center);
		\draw [style={extension_edge_dashed}] (3.center) to (6.center);

\end{tikzpicture}
  \caption{Worst-case tree extension $\Gamma_2$}
  \end{subfigure}
  
  \cr
}
\caption{(a): The \emph{ladder-graph} $L_n$ (with $n \geq 3$), which is a rooted binary network with level $n-1$ and $2n$ vertices. (b): An optimal tree extension $\Gamma_1$ of $L_n$ with scanwidth 3. (c): The worst-case tree extension $\Gamma_2$ of $L_n$ with scanwidth $n$.
}
        \label{fig:level}
\end{figure}

\section{Reduction rules}
\label{sec:reductions}
Before discussing algorithms that compute the scanwidth and find its corresponding (tree) extension, we present some reduction rules. In \cref{subsec:blocks} we will explain that one can split the scanwidth problem over the blocks of a network. \cref{subsec:edge_contraction} will cover a rule that suppresses indegree-1 outdegree-1 vertices, while \cref{subsec:total_reduction} summarizes the complete reduction scheme and provides a bound on the size of a reduced instance.

\subsection{Splitting into (s-)blocks}
\label{subsec:blocks}
Magne et al. \cite{magne2023edge} mention that the edge-treewidth can be split into subproblems for each block of a graph. A similar result is true for the scanwidth of a rooted DAG. The intuitive reason is that one may recursively place the vertices of a highest block at the end of an extension. One needs to be careful, however, as this is not the case for DAGs with multiple roots. Then, if multiple blocks contain a root, none of these blocks necessarily needs to be completely above the others. To remedy this problem, we introduce a non-standard generalization of a block for DAGs with multiple roots, which we call \emph{scanwidth-blocks} or \emph{s-blocks}, and whose definition is illustrated in \cref{fig:sblocks}. If a DAG has a single root, its s-blocks coincide with its blocks.

\begin{definition}[S-block]\label{def:sblocks}
Let $G=(V,E)$ be a weakly connected DAG, and let $H$ be the underlying undirected graph of $G$ where edges are added between all roots of $G$. Then for all $W \subseteq V$, the subgraph $G[W]$ is an s-block of $G$ if $H[W]$ is a block of $H$.
\end{definition}

\begin{figure}[htb]
     \centering
     \begin{subfigure}[b]{0.32\textwidth}
         \centering
         \begin{tikzpicture}[xscale=0.45, yscale=0.35]
	\begin{pgfonlayer}{nodelayer}
		\node [style={graph_node}] (0) at (2.75, 12.25) {};
		\node [style={graph_node}] (1) at (3.5, 9) {};
		\node [style={graph_node}] (2) at (4.75, 11) {};
		\node [style={graph_node}] (3) at (4.25, 7.25) {};
		\node [style={graph_node}] (4) at (7, 8.5) {};
		\node [style={graph_node}] (5) at (5.5, 9.75) {};
		\node [style={graph_node}] (6) at (6.75, 11.75) {};
		\node [style={graph_node}] (7) at (8.5, 14.25) {};
		\node [style={graph_node}] (8) at (9, 11) {};
		\node [style={graph_node}] (9) at (6, 6.75) {};
		\node [style={graph_node}] (10) at (8.75, 6.75) {};
		\node [style={graph_node}] (12) at (6, 2.75) {};
		\node [style={graph_node}] (13) at (6, 0) {};
		\node [style={graph_node}] (14) at (7, 1.75) {};
		\node [style={graph_node}] (15) at (9.5, 4) {};
		\node [style={graph_node}] (16) at (9.25, 1.75) {};
		\node [style={graph_node}] (17) at (4.5, 5.25) {};
		\node [style={graph_node}] (18) at (5, 3.5) {};
		\node [style={graph_node}] (19) at (3.25, 3.5) {};
		\node [style={graph_node}] (20) at (2.5, 1.25) {};
		\node [style={graph_node}] (21) at (4.25, 1.5) {};
		\node [style={graph_node}] (22) at (2.5, 7) {};
		\node [style={graph_node}] (23) at (6.5, 5.5) {};
		\node [style={graph_node}] (29) at (7.75, 4.5) {};
		\node [style={graph_node}] (30) at (1.5, 9) {};
	\end{pgfonlayer}
	\begin{pgfonlayer}{edgelayer}
		\draw [style={graph_edge}] (0) to (22);
		\draw [style={graph_edge}] (0) to (1);
		\draw [style={graph_edge}] (2) to (1);
		\draw [style={graph_edge}] (2) to (3);
		\draw [style={graph_edge}] (1) to (3);
		\draw [style={graph_edge}] (1) to (22);
		\draw [style={graph_edge}] (6) to (5);
		\draw [style={graph_edge}] (5) to (4);
		\draw [style={graph_edge}] (6) to (4);
		\draw [style={graph_edge}] (7) to (4);
		\draw [style={graph_edge}] (8) to (4);
		\draw [style={graph_edge}] (7) to (8);
		\draw [style={graph_edge}] (4) to (3);
		\draw [style={graph_edge}] (3) to (17);
		\draw [style={graph_edge}] (17) to (18);
		\draw [style={graph_edge}] (17) to (19);
		\draw [style={graph_edge}] (19) to (20);
		\draw [style={graph_edge}] (19) to (21);
		\draw [style={graph_edge}] (4) to (10);
		\draw [style={graph_edge}] (15) to (16);
		\draw [style={graph_edge}] (12) to (13);
		\draw [style={graph_edge}] (14) to (13);
		\draw [style={graph_edge}] (4) to (9);
		\draw [style={graph_edge}] (9) to (23);
		\draw [style={graph_edge}] (23) to (29);
		\draw [style={graph_edge}] (29) to (12);
		\draw [style={graph_edge}] (29) to (14);
		\draw [style={graph_edge}] (29) to (16);
		\draw [style={graph_edge}] (10) to (29);
		\draw [style={graph_edge}] (29) to (15);
		\draw [style={graph_edge}] (4) to (29);
		\draw [style={graph_edge}] (0) to (30);
	\end{pgfonlayer}
\end{tikzpicture}
         \caption{Weakly connected DAG $G$}
         \label{fig:sblock:cutvertex}
     \end{subfigure}
     \hfill
     \begin{subfigure}[b]{0.32\textwidth}
         \centering
         \begin{tikzpicture}[xscale=0.45, yscale=0.35]
\tikzstyle{infinity_weight}=[draw=black!80, dashed]
\tikzstyle{graph_edge}=[draw=black]
	\begin{pgfonlayer}{nodelayer}
		\node [style={graph_node}] (0) at (2.75, 12.25) {};
		\node [style={graph_node}] (1) at (3.5, 9) {};
		\node [style={graph_node}] (2) at (4.75, 11) {};
		\node [style={graph_node}] (3) at (4.25, 7.25) {};
		\node [style={graph_node}] (4) at (7, 8.5) {};
		\node [style={graph_node}] (5) at (5.5, 9.75) {};
		\node [style={graph_node}] (6) at (6.75, 11.75) {};
		\node [style={graph_node}] (7) at (8.5, 14.25) {};
		\node [style={graph_node}] (8) at (9, 11) {};
		\node [style={graph_node}] (9) at (6, 6.75) {};
		\node [style={graph_node}] (10) at (8.75, 6.75) {};
		\node [style={graph_node}] (12) at (6, 2.75) {};
		\node [style={graph_node}] (13) at (6, 0) {};
		\node [style={graph_node}] (14) at (7, 1.75) {};
		\node [style={graph_node}] (15) at (9.5, 4) {};
		\node [style={graph_node}] (16) at (9.25, 1.75) {};
		\node [style={graph_node}] (17) at (4.5, 5.25) {};
		\node [style={graph_node}] (18) at (5, 3.5) {};
		\node [style={graph_node}] (19) at (3.25, 3.5) {};
		\node [style={graph_node}] (20) at (2.5, 1.25) {};
		\node [style={graph_node}] (21) at (4.25, 1.5) {};
		\node [style={graph_node}] (22) at (2.5, 7) {};
		\node [style={graph_node}] (23) at (6.5, 5.5) {};
		\node [style={graph_node}] (29) at (7.75, 4.5) {};
		\node [style={graph_node}] (30) at (1.5, 9) {};
	\end{pgfonlayer}
	\begin{pgfonlayer}{edgelayer}
		\draw [style={graph_edge}] (0) to (22);
		\draw [style={graph_edge}] (0) to (1);
		\draw [style={graph_edge}] (2) to (1);
		\draw [style={graph_edge}] (2) to (3);
		\draw [style={graph_edge}] (1) to (3);
		\draw [style={graph_edge}] (1) to (22);
		\draw [style={graph_edge}] (6) to (5);
		\draw [style={graph_edge}] (5) to (4);
		\draw [style={graph_edge}] (6) to (4);
		\draw [style={graph_edge}] (7) to (4);
		\draw [style={graph_edge}] (8) to (4);
		\draw [style={graph_edge}] (7) to (8);
		\draw [style={graph_edge}] (4) to (3);
		\draw [style={graph_edge}] (3) to (17);
		\draw [style={graph_edge}] (17) to (18);
		\draw [style={graph_edge}] (17) to (19);
		\draw [style={graph_edge}] (19) to (20);
		\draw [style={graph_edge}] (19) to (21);
		\draw [style={graph_edge}] (4) to (10);
		\draw [style={graph_edge}] (15) to (16);
		\draw [style={graph_edge}] (12) to (13);
		\draw [style={graph_edge}] (14) to (13);
		\draw [style={graph_edge}] (4) to (9);
		\draw [style={graph_edge}] (9) to (23);
		\draw [style={graph_edge}] (23) to (29);
		\draw [style={graph_edge}] (29) to (12);
		\draw [style={graph_edge}] (29) to (14);
		\draw [style={graph_edge}] (29) to (16);
		\draw [style={graph_edge}] (10) to (29);
		\draw [style={graph_edge}] (29) to (15);
		\draw [style={graph_edge}] (4) to (29);
		\draw [style={graph_edge}] (0) to (30);
		\draw [style={infinity_weight}] (0) to (7);
		\draw [style={infinity_weight}] (0) to (2);
		\draw [style={infinity_weight}] (0) to (6);
		\draw [style={infinity_weight}] (7) to (6);
		\draw [style={infinity_weight}] (7) to (2);
		\draw [style={infinity_weight}] (6) to (2);
	\end{pgfonlayer}
\end{tikzpicture}
         \caption{Auxiliary graph $H$}
     \end{subfigure}
          \hfill
     \begin{subfigure}[b]{0.32\textwidth}
         \centering
         \begin{tikzpicture}[xscale=0.45, yscale=0.35]
\tikzstyle{infinity_weight}=[draw=black!30, dashed]
	\begin{pgfonlayer}{nodelayer}
		\node [style={graph_node}] (31) at (15.75, 2.5) {};
		\node [style={graph_node}] (32) at (16, 0) {};
		\node [style={graph_node}] (33) at (17, 1.75) {};
		\node [style={graph_node}] (34) at (17.25, 4.25) {};
		\node [style={graph_node}] (35) at (19.75, 2.75) {};
		\node [style={graph_node}] (36) at (19.5, 0.5) {};
		\node [style={graph_node}] (37) at (18.25, 3.5) {};
		\node [style={graph_node}] (38) at (18, 8.5) {};
		\node [style={graph_node}] (39) at (17, 6.75) {};
		\node [style={graph_node}] (40) at (19.5, 6.75) {};
		\node [style={graph_node}] (41) at (17.25, 5.5) {};
		\node [style={graph_node}] (42) at (18.5, 4.5) {};
		\node [style={graph_node}] (43) at (12.75, 12.25) {};
		\node [style={graph_node}] (44) at (13.5, 9) {};
		\node [style={graph_node}] (45) at (14.75, 11) {};
		\node [style={graph_node}] (46) at (14.25, 7.25) {};
		\node [style={graph_node}] (47) at (17, 8.5) {};
		\node [style={graph_node}] (48) at (15.5, 9.75) {};
		\node [style={graph_node}] (49) at (16.75, 11.75) {};
		\node [style={graph_node}] (50) at (18.5, 14.25) {};
		\node [style={graph_node}] (51) at (19, 11) {};
		\node [style={graph_node}] (52) at (12.5, 7) {};
		\node [style={graph_node}] (53) at (12, 11.25) {};
		\node [style={graph_node}] (54) at (11.25, 8.5) {};
		\node [style={graph_node}] (55) at (15, 6.75) {};
		\node [style={graph_node}] (56) at (14.5, 5.25) {};
		\node [style={graph_node}] (57) at (14, 4.5) {};
		\node [style={graph_node}] (58) at (14.75, 2.75) {};
		\node [style={graph_node}] (59) at (13.5, 5.25) {};
		\node [style={graph_node}] (60) at (12.5, 3.75) {};
		\node [style={graph_node}] (61) at (12.75, 2.75) {};
		\node [style={graph_node}] (62) at (13.5, 1) {};
		\node [style={graph_node}] (63) at (12, 2.75) {};
		\node [style={graph_node}] (64) at (11.25, 1) {};
	\end{pgfonlayer}
	\begin{pgfonlayer}{edgelayer}
		\draw [style={graph_edge}] (31) to (32);
		\draw [style={graph_edge}] (33) to (32);
		\draw [style={graph_edge}] (34) to (31);
		\draw [style={graph_edge}] (34) to (33);
		\draw [style={graph_edge}] (35) to (36);
		\draw [style={graph_edge}] (37) to (36);
		\draw [style={graph_edge}] (37) to (35);
		\draw [style={graph_edge}] (38) to (40);
		\draw [style={graph_edge}] (38) to (39);
		\draw [style={graph_edge}] (39) to (41);
		\draw [style={graph_edge}] (41) to (42);
		\draw [style={graph_edge}] (40) to (42);
		\draw [style={graph_edge}] (38) to (42);
		\draw [style={graph_edge}] (43) to (52);
		\draw [style={graph_edge}] (43) to (44);
		\draw [style={graph_edge}] (45) to (44);
		\draw [style={graph_edge}] (45) to (46);
		\draw [style={graph_edge}] (44) to (46);
		\draw [style={graph_edge}] (44) to (52);
		\draw [style={graph_edge}] (49) to (48);
		\draw [style={graph_edge}] (48) to (47);
		\draw [style={graph_edge}] (49) to (47);
		\draw [style={graph_edge}] (50) to (47);
		\draw [style={graph_edge}] (51) to (47);
		\draw [style={graph_edge}] (50) to (51);
		\draw [style={graph_edge}] (47) to (46);
		\draw [style={graph_edge}] (53) to (54);
		\draw [style={graph_edge}] (55) to (56);
		\draw [style={graph_edge}] (57) to (58);
		\draw [style={graph_edge}] (59) to (60);
		\draw [style={graph_edge}] (61) to (62);
		\draw [style={graph_edge}] (63) to (64);
		\draw [style={infinity_weight}] (43) to (45);
		\draw [style={infinity_weight}] (45) to (49);
		\draw [style={infinity_weight}] (49) to (50);
		\draw [style={infinity_weight}] (50) to (43);
		\draw [style={infinity_weight}] (43) to (49);
		\draw [style={infinity_weight}] (45) to (50);
	\end{pgfonlayer}
\end{tikzpicture}
         \caption{S-blocks of $G$}
         \label{subfig:sblocks}
     \end{subfigure}
        \caption{(a): A multi-rooted weakly connected DAG $G$. (b): The auxiliary undirected graph $H$ (as defined in \cref{def:sblocks}) where the newly added edges are dashed. (c): The s-blocks of the DAG $G$.}
        \label{fig:sblocks}
\end{figure}
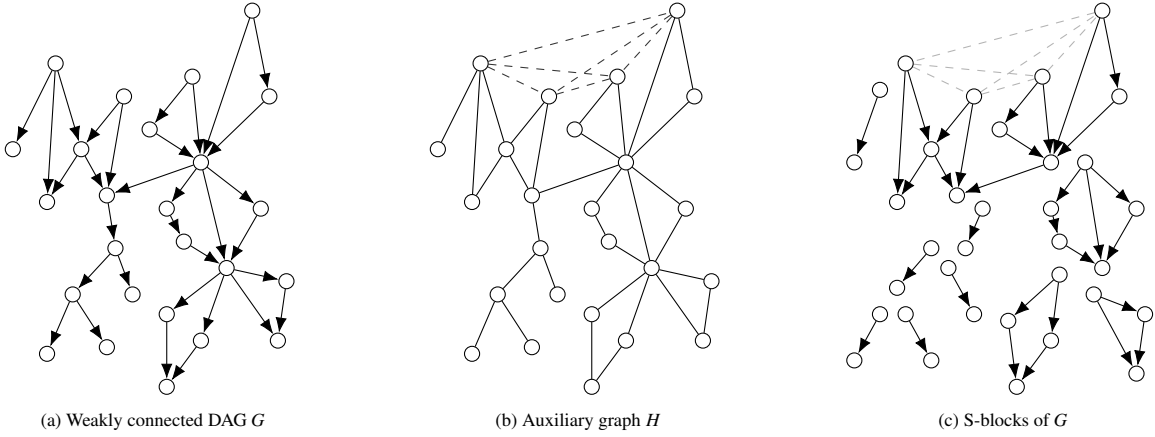

The proof of the following theorem is rather technical, yet the result is very intuitive: we can split the graph into its s-blocks when computing the scanwidth. Note that the proof is constructive and can therefore be used to reconstruct solutions of the whole problem from its subproblems.
\begin{theorem}
\label{thm:split_sblocks}
Let $G$ be a weakly connected DAG and $\mathcal{S}(G)$ the set of s-blocks of $G$. Then, we have that $\sw(G) = \max_{H \in \mathcal{S}(G)} \sw(H). $
\end{theorem}
\begin{proof}
We will prove this by induction on the number of s-blocks $k$. The base case where $k=1$ follows trivially.

Let $k\geq 1$, and assume that the theorem holds for any graph with $k$ s-blocks. Let $G$ be a weakly connected DAG with $k+1$ s-blocks. By definition these s-blocks are exactly the blocks of the graph $H$ from \cref{def:sblocks}. But then, there must be an s-block that contains all the roots of $G$ (the `root-block'), as those roots form a clique in the auxiliary graph $H$ (see also \cref{subfig:sblocks}). Let $G_1 = G[V_1]$ (with $V_1 \subset V$) be an s-block that is not the root-block and such that $G_1$ only shares one vertex $v$ with the rest of $G$. As all vertices in $V_1$ have a path from a root in $G$ and all such paths must pass through $v$, we must have that $v$ is the unique root of $G_1$.

Now define $G_2 = G[V \backslash V_1 \cup \{ v \}] = (V_2, E_2)$. Then, $G_1$ and $G_2$ are subgraphs that only have $v$ in common. Let $\sigma_1 \in \Pi[V_1]$ and $\sigma_2 \in \Pi[V_2]$. Since $v$ is the root of $G_1$, we know that it always holds that $\sigma_1 = \sigma_1' \circ (v)$, where $\sigma_1' \in \Pi[V_1 \setminus \{ v\}]$. Because $G_1$ is a pendent s-block of $G$, we can define $\sigma = \sigma_1' \circ \sigma_2 \in \Pi[V]$. Since no arcs are directed from $V_2 \setminus \{ v\}$ to $V_1 \setminus \{ v\}$ (because both parts are only connected through $v$), we then have:
\begin{align}
\sw(\sigma, G) = \max_{w \in V} |\SW_w^\sigma (G)| &= \max \left\{ \max_{w \in V_1\setminus \{v\}} |\SW_w^{\sigma} (G)|, \max_{ w \in V_2} |\SW_w^{\sigma} (G)| \right\} \nonumber\\
&= \max \left\{ \max_{w \in V_1} |\SW_w^{\sigma_1} (G_1)|, \max_{ w \in V_2} |\SW_w^{\sigma_2} (G_2)| \right\} \nonumber\\
&= \max \left\{ \sw(\sigma_1, G_1) , \sw(\sigma_2 , G_2) \right\}.\label{eq:sblocks}
\end{align}
Here, we used that $|\SW_v^{\sigma_1} (G_1)| = 0$ as $v$ is the root of $G_1$.

We will now prove that $\sw( G) = \max \{\sw(G_1) , \sw(G_2)\}$. Since the $k+1$ s-blocks of $G$ are exactly $G_1$ and the $k$ s-blocks of $G_2$, the induction hypothesis will then prove the theorem.

($\leq$)
Let $\sigma_1 = \sigma_1' \circ (v) \in \Pi[V_1]$ and $\sigma_2 \in \Pi[V_2]$, such that $\sw(G_1) = \sw(\sigma_1, G_1)$ and $\sw(G_2) = \sw(\sigma_2, G_2)$. Let $\sigma = \sigma_1' \circ \sigma_2 \in \Pi[V]$. Then, using \cref{eq:sblocks}, $\sw(G) \leq  \max \{\sw(G_1) , \sw(G_2)\}$.

($\geq$) $G_1$ and $G_2$ are both weakly connected subgraphs of $G$. Thus their respective scanwidth values both form a lower bound on the scanwidth of $G$, providing us with $\sw(G)  \geq \max \{\sw(G_1) , \sw(G_2)\}$.
\end{proof}

An immediate corollary is of this result is that one can split over the blocks to compute the scanwidth if $G$ has a single root, since the s-blocks and blocks then coincide.

\begin{corollary}
\label{cor:split_blocks}
Let $G$ be a weakly connected rooted DAG and $\mathcal{B}(G)$ the set of blocks of $G$. Then, we have that $\sw(G) = \max_{H \in \mathcal{B}(G)} \sw(H). $
\end{corollary}

\subsection{Suppressing indegree-1 outdegree-1 vertices} \label{subsec:edge_contraction}
We now introduce another safe reduction rule that allows us to simplify the graph by removing certain arcs and vertices. It formalizes the idea that if a vertex has only one incoming arc and one outgoing arc, the two arcs can be replaced by a single arc and the vertex can be deleted. In other words, we suppress vertices of indegree-1 and outdegree-1.

\begin{lemma}
\label{lem:edge_contraction}
Let $G=(V,E)$ be a weakly connected DAG, and let $v\in V$ be a vertex with a single parent $u$ and a single child $w$ such that $uw \notin E$. Define a new DAG $H = (V \setminus \{ v\}, E \setminus \{ uv, vw \} \cup \{uw \})$. Then, $\sw(G) = \sw (H)$.
\end{lemma}
\begin{proof}
We first present the following obvious claim without proof (see \cite{holtgrefe2023scanwidth} for its straightforward proof).

\emph{Claim:} Let $\sigma$ be an extension of $G$, and let $\pi = \sigma[V \setminus \{ v\} ]$ be an extension of $H$. Then for all $a,b \in V \setminus \{ v\}$, we have that $a \connect{G[\sigma[1 \ldots b]]} b$ if and only if $a \connect{H[\pi[1 \ldots b]]} b$. Furthermore, $a >_G b$ if and only if $a >_H b$.

(\emph{$\sw(G)\geq \sw(H)$}):
Let $\sigma$ be an optimal extension for $G$, and define the extension $\pi = \sigma[V \setminus \{ v\}]$ for $H$. Now let $z \in V\setminus \{v \}$ be arbitrary, and let $xy \in \SW_z^\pi (H)$. If $xy \neq uw$, then $xy \in \SW_z^\sigma (G)$ by both parts of the claim. If $xy = uw$, then clearly either $uv$ or $vw$ is in $\SW_z^\sigma (G)$. Thus, we get that $|\SW_z^\pi (H)| \leq |\SW_z^\sigma (G)| $. Therefore, 
$$\sw(G) = \sw(\sigma, G)  = \max_{z \in V} |\SW_z^\sigma(G)| \geq  \max_{z \in V\setminus \{ v \}} |\SW_z^\pi(H)| =\sw (\pi, H) \geq \sw (H).$$

(\emph{$\sw(G)\leq \sw(H)$}):
Let $\pi$ be an optimal extension for $H$ with $w = \pi (i)$. We also have that $\sigma = \pi[1 \ldots i] \circ (v) \circ \pi[i+1 \ldots ]$ is an extension of $G$. Note that then $\sigma[V \setminus \{ v \} ] = \pi$, as in the claim. Now let $z \in V\setminus \{v \}$ be arbitrary, and let $xy \in \SW_z^\sigma (G)$. If $xy\neq uv$ and $xy \neq vw$, then $xy \in \SW_z^\pi (H)$ by both parts of the claim and the definition of the sets $\SW$. In the other case, $xy$ is either $uv$ or $vw$ (note that they can not both be in the set $\SW_z^\sigma (G)$). But then, $uw \in \SW_z^\pi (H)$. Overall, we get that $|\SW_z^\sigma (H)| \leq |\SW_z^\pi (H)|$. By similar arguments, one finds that $|\SW_v^\sigma (G)| =|\SW_w^\pi (H)|$. Therefore, 
\begin{align*}
\sw(G) &\leq  \sw(\sigma, G)  = \max_{z \in V} |\SW_z^\sigma(G)| = \max \left\{ |\SW_v^\sigma(G)|,  \max_{z \in V \setminus \{ v\} } |\SW_z^\sigma(G)| \right\} \\
&\leq \max \left\{ |\SW_w^\pi(H)|,  \max_{z \in V \setminus \{ v\} } |\SW_z^\pi(H)| \right\} = \max_{z \in V \setminus \{ v\} } |\SW_z^\pi(H)| = \sw ( \pi, H) = \sw (H).
\qedhere
\end{align*}
\end{proof}

The reason we impose the restriction that $uw \notin E$ in the above lemma is that this would create a multigraph. Of course, one can generalize scanwidth to multigraphs and remove the restriction. Similarly, if we were to generalize scanwidth to weighted graphs, the restriction can be removed by increasing the weight of the already existing arc.

\subsection{Complete decomposition scheme}
\label{subsec:total_reduction}
In combination with any exact algorithm, the reduction rules from the previous two subsections can be used to compute the scanwidth. These rules are gathered in \cref{alg:reduction}.

\begin{algorithm}[H]
\caption{Decomposition algorithm}
\label{alg:reduction}
\Input{Weakly connected DAG $G=(V,E)$, an exact algorithm \texttt{ExactSW} that computes the scanwidth of a weakly connected DAG.}
\Output{Scanwidth $\sw$ of $G$.}
\KwInit{$\sw \gets 0$}\\
$G' \gets $ underlying undirected graph of $G$, with edges added between all roots of $G$\\
$\mathcal{S} \gets \{G[W]: W \subseteq V, G'[W] \text{ is a block of } G' \}$\tcp*{Set of s-blocks of $G$}
\For{$H \in \mathcal{S}$}{
\If{$H$ is a single arc}{$\sw \gets \max\{\sw, 1\}$}
\ElseIf{$H$ has a unique root and its underlying undirected graph is a cycle}{$\sw \gets \max\{\sw, 2\}$}
\Else{
$H' \gets H$ after exhaustively suppressing vertices using \cref{lem:edge_contraction}\\
$\sw \gets \max \{ \sw, \texttt{ExactSW} (H') \}$
}
}
\Return{$\sw$}
\end{algorithm}

To summarize, the algorithm first decomposes a DAG into its s-blocks. It then checks for each s-block whether it is a single arc or a rooted cycle. If this is the case, the scanwidth of that s-block is already known. For the remaining s-blocks, we use the vertex suppression rule to decrease their size. Then, we only need to exactly calculate the scanwidth of those s-blocks with some exact algorithm. \cref{fig:reductions} provides an illustration of the decomposition on an example graph.

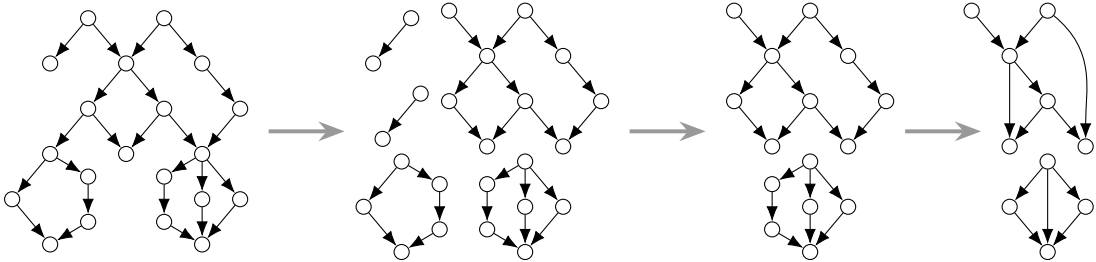
\begin{figure}[htb]
     \centering
		\begin{tikzpicture}[yscale=0.40, xscale=0.5]
\tikzstyle{hybrid_edge}=[-Stealth, draw=gray!80, line width=1.5pt]
	\begin{pgfonlayer}{nodelayer}
		\node [style={graph_node}] (42) at (6, -0.25) {};
		\node [style={graph_node}] (43) at (1, 1.25) {};
		\node [style={graph_node}] (44) at (3, 1.25) {};
		\node [style={graph_node}] (45) at (5, 1.25) {};
		\node [style={graph_node}] (46) at (2, 2.75) {};
		\node [style={graph_node}] (47) at (4, 2.75) {};
		\node [style={graph_node}] (48) at (3, 4.25) {};
		\node [style={graph_node}] (49) at (6, 2.75) {};
		\node [style={graph_node}] (50) at (4, 5.75) {};
		\node [style={graph_node}] (51) at (2, 0.5) {};
		\node [style={graph_node}] (52) at (0, -0.25) {};
		\node [style={graph_node}] (53) at (1, -1.75) {};
		\node [style={graph_node}] (54) at (2, 5.75) {};
		\node [style={graph_node}] (55) at (4, 0.5) {};
		\node [style={graph_node}] (56) at (5, -1.75) {};
		\node [style={graph_node}] (58) at (5, 4.25) {};
		\node [style={graph_node}] (59) at (1, 4.25) {};
		\node [style={graph_node}] (67) at (4, -1) {};
		\node [style={graph_node}] (68) at (5, -0.25) {};
		\node [style={graph_node}] (73) at (9.75, 1.75) {};
		\node [style={graph_node}] (74) at (10.75, 3.25) {};
		\node [style={graph_node}] (75) at (12.5, 1.5) {};
		\node [style={graph_node}] (76) at (14.5, 1.5) {};
		\node [style={graph_node}] (77) at (11.5, 3) {};
		\node [style={graph_node}] (78) at (13.5, 3) {};
		\node [style={graph_node}] (79) at (12.5, 4.5) {};
		\node [style={graph_node}] (80) at (15.5, 3) {};
		\node [style={graph_node}] (81) at (13.5, 6) {};
		\node [style={graph_node}] (82) at (11.5, 6) {};
		\node [style={graph_node}] (83) at (14.5, 4.5) {};
		\node [style={graph_node}] (84) at (14.5, -0.5) {};
		\node [style={graph_node}] (85) at (13.5, 1) {};
		\node [style={graph_node}] (86) at (12.5, 0.25) {};
		\node [style={graph_node}] (87) at (13.5, -2) {};
		\node [style={graph_node}] (88) at (12.5, -1.25) {};
		\node [style={graph_node}] (89) at (13.5, -0.5) {};
		\node [style={graph_node}] (90) at (10.5, 5.75) {};
		\node [style={graph_node}] (91) at (9.5, 4.25) {};
		\node [style={graph_node}] (117) at (20, 1.5) {};
		\node [style={graph_node}] (118) at (22, 1.5) {};
		\node [style={graph_node}] (119) at (19, 3) {};
		\node [style={graph_node}] (120) at (21, 3) {};
		\node [style={graph_node}] (121) at (20, 4.5) {};
		\node [style={graph_node}] (122) at (23, 3) {};
		\node [style={graph_node}] (123) at (21, 6) {};
		\node [style={graph_node}] (124) at (19, 6) {};
		\node [style={graph_node}] (125) at (22, 4.5) {};
		\node [style={graph_node}] (126) at (22, -0.5) {};
		\node [style={graph_node}] (127) at (21, 1) {};
		\node [style={graph_node}] (128) at (20, 0.25) {};
		\node [style={graph_node}] (129) at (21, -2) {};
		\node [style={graph_node}] (130) at (20, -1.25) {};
		\node [style={graph_node}] (131) at (21, -0.5) {};
		\node [style={graph_node}] (132) at (26.25, 1.5) {};
		\node [style={graph_node}] (133) at (28.25, 1.5) {};
		\node [style={graph_node}] (135) at (27.25, 3) {};
		\node [style={graph_node}] (136) at (26.25, 4.5) {};
		\node [style={graph_node}] (138) at (27.25, 6) {};
		\node [style={graph_node}] (139) at (25.25, 6) {};
		\node [style={graph_node}] (141) at (28.25, -0.5) {};
		\node [style={graph_node}] (142) at (27.25, 1) {};
		\node [style={graph_node}] (144) at (27.25, -2) {};
		\node [style={graph_node}] (147) at (26.25, -0.5) {};
		\node [style=none] (148) at (6.75, 2) {};
		\node [style=none] (149) at (8.75, 2) {};
		\node [style=none] (150) at (16.25, 2) {};
		\node [style=none] (151) at (18.25, 2) {};
		\node [style=none] (156) at (23.5, 2) {};
		\node [style=none] (157) at (25.5, 2) {};
		\node [style={graph_node}] (158) at (2, -1) {};
		\node [style={graph_node}] (159) at (10.25, 1) {};
		\node [style={graph_node}] (160) at (11.25, 0.25) {};
		\node [style={graph_node}] (161) at (9.25, -0.5) {};
		\node [style={graph_node}] (162) at (10.25, -2) {};
		\node [style={graph_node}] (163) at (11.25, -1.25) {};
	\end{pgfonlayer}
	\begin{pgfonlayer}{edgelayer}
		\draw [style={graph_edge}] (50) to (48);
		\draw [style={graph_edge}] (48) to (46);
		\draw [style={graph_edge}] (46) to (43);
		\draw [style={graph_edge}] (48) to (47);
		\draw [style={graph_edge}] (46) to (44);
		\draw [style={graph_edge}] (47) to (44);
		\draw [style={graph_edge}] (47) to (45);
		\draw [style={graph_edge}] (45) to (42);
		\draw [style={graph_edge}] (49) to (45);
		\draw [style={graph_edge}] (43) to (52);
		\draw [style={graph_edge}] (52) to (53);
		\draw [style={graph_edge}] (43) to (51);
		\draw [style={graph_edge}] (54) to (48);
		\draw [style={graph_edge}] (50) to (58);
		\draw [style={graph_edge}] (58) to (49);
		\draw [style={graph_edge}] (54) to (59);
		\draw [style={graph_edge}] (45) to (55);
		\draw [style={graph_edge}] (55) to (67);
		\draw [style={graph_edge}] (67) to (56);
		\draw [style={graph_edge}] (42) to (56);
		\draw [style={graph_edge}] (45) to (68);
		\draw [style={graph_edge}] (68) to (56);
		\draw [style={graph_edge}] (74) to (73);
		\draw [style={graph_edge}] (81) to (79);
		\draw [style={graph_edge}] (79) to (77);
		\draw [style={graph_edge}] (79) to (78);
		\draw [style={graph_edge}] (77) to (75);
		\draw [style={graph_edge}] (78) to (75);
		\draw [style={graph_edge}] (78) to (76);
		\draw [style={graph_edge}] (80) to (76);
		\draw [style={graph_edge}] (82) to (79);
		\draw [style={graph_edge}] (81) to (83);
		\draw [style={graph_edge}] (83) to (80);
		\draw [style={graph_edge}] (85) to (84);
		\draw [style={graph_edge}] (85) to (86);
		\draw [style={graph_edge}] (86) to (88);
		\draw [style={graph_edge}] (88) to (87);
		\draw [style={graph_edge}] (84) to (87);
		\draw [style={graph_edge}] (85) to (89);
		\draw [style={graph_edge}] (89) to (87);
		\draw [style={graph_edge}] (90) to (91);
		\draw [style={graph_edge}] (123) to (121);
		\draw [style={graph_edge}] (121) to (119);
		\draw [style={graph_edge}] (121) to (120);
		\draw [style={graph_edge}] (119) to (117);
		\draw [style={graph_edge}] (120) to (117);
		\draw [style={graph_edge}] (120) to (118);
		\draw [style={graph_edge}] (122) to (118);
		\draw [style={graph_edge}] (124) to (121);
		\draw [style={graph_edge}] (123) to (125);
		\draw [style={graph_edge}] (125) to (122);
		\draw [style={graph_edge}] (127) to (126);
		\draw [style={graph_edge}] (127) to (128);
		\draw [style={graph_edge}] (128) to (130);
		\draw [style={graph_edge}] (130) to (129);
		\draw [style={graph_edge}] (126) to (129);
		\draw [style={graph_edge}] (127) to (131);
		\draw [style={graph_edge}] (131) to (129);
		\draw [style={graph_edge}] (138) to (136);
		\draw [style={graph_edge}] (136) to (135);
		\draw [style={graph_edge}] (135) to (132);
		\draw [style={graph_edge}] (135) to (133);
		\draw [style={graph_edge}] (139) to (136);
		\draw [style={graph_edge}] (142) to (141);
		\draw [style={graph_edge}] (141) to (144);
		\draw [style={graph_edge}] (142) to (147);
		\draw [style={graph_edge}] (147) to (144);
		\draw [style={graph_edge}, in=90, out=-45] (138) to (133);
		\draw [style={graph_edge}] (136) to (132);
		\draw [style={hybrid_edge}] (148.center) to (149.center);
		\draw [style={hybrid_edge}] (150.center) to (151.center);
		\draw [style={hybrid_edge}] (156.center) to (157.center);
		\draw [style={graph_edge}] (51) to (158);
		\draw [style={graph_edge}] (158) to (53);
		\draw [style={graph_edge}] (159) to (161);
		\draw [style={graph_edge}] (161) to (162);
		\draw [style={graph_edge}] (159) to (160);
		\draw [style={graph_edge}] (160) to (163);
		\draw [style={graph_edge}] (163) to (162);
		\draw [style={graph_edge}] (142) to (144);
	\end{pgfonlayer}
\end{tikzpicture}
        \caption{Illustration of the decomposition algorithm. The first graph shows the original graph. Then, the graph is split into its s-blocks. Thereafter, the single arcs and the cycle with a unique root are `deleted', since their scanwidths are known. Note that we then also know that the scanwidth of the complete graph is at least 2. Lastly, we exhaustively suppress indegree-1 outdegree-1 vertices in the remaining s-blocks, until it will create a multigraph.}
        \label{fig:reductions}
\end{figure}

In the next lemma, we formally prove the algorithm's correctness and its time complexity. We note that the proofs of \cref{thm:split_sblocks,lem:edge_contraction} can be used to make the algorithm constructive if \texttt{ExactSW} also returns an optimal extension.
\begin{lemma}
\label{lem:alg_reduction}
Let $G=(V,E)$ be a weakly connected DAG of $n$ vertices and $m$ arcs, and let \texttt{ExactSW} be an algorithm that finds the scanwidth of any weakly connected DAG $H' =(V', E')$ in $O( f(|V'|, |E'|))$ time for some function $f$. Then, \cref{alg:reduction} returns the scanwidth of $G$ in $O(n^2 + n \cdot f(n', m'))$ time. Here, $n'\leq n$ (resp. $m'\leq m$) is the maximum number of vertices (resp. arcs) of any of the graphs $H'$ that appear in \cref{alg:reduction}.
\end{lemma}
\begin{proof}
\emph{Correctness:} By \cref{def:sblocks}, $\mathcal{S}$ will be the set of s-blocks of $G$. From \cref{thm:split_sblocks} it follows that we can indeed maximize over the scanwidth of the s-blocks of $G$. Since a single arc has just one extension of scanwidth 1, the first if statement is correct. Any extension of a cycle with a unique root has a scanwidth of 2, showing correctness of the second if statement. Correctness of the else statement now follows from \cref{lem:edge_contraction}, and the assumption that \texttt{ExactSW} correctly returns the scanwidth of $H'$.

\emph{Time complexity:} The time complexity of the creation of $G'$ is bounded by $O(n^2)$. We can find the blocks of $G'$ in $O(|V(G')| + |E(G')|)=O(n^2)$ time \cite{tarjan}. We then loop over at most $n$ s-blocks. Each of those s-blocks $H$ contains at most $n$ vertices. Thus, it takes $O(n)$ time to check whether $H$ is a single arc, or if a cycle with a unique root. Exhaustively suppressing vertices also takes $O(n)$ time. Using the algorithm \texttt{ExactSW} on $H'$ then takes $f(n', m')$ time, where $n'$ and $m'$ are bounds on the size of $H'$, as defined in the lemma. Summarizing, the complexity becomes $O(n^2 + n \cdot (n + f(n', m')) )=  O(n^2 + n \cdot f(n', m'))$.
\end{proof}

Due to this lemma, it is valuable to search for an upper bound on the size of the graphs $H'$ in the algorithm. This could help to bound the running time of an exact algorithm.

It turns out that if $G$ is a network, we can bound the size of the graphs $H'$ by a linear function of its level. A crucial observation here is that the graphs $H'$ are very similar to the so-called \emph{simple level-$k$ generators}: building blocks of binary networks introduced in \cite{van2009algorithms,van2008constructing}. In \cite[Lem\,4.2]{van2009algorithms} it is shown that the number of vertices of a level-$k$ generator is at most $3k-1$, while the number of arcs is at most $4k-2$. Using techniques from this proof, we create similar bounds for our graphs $H'$, although we need some adjustments of the proof to allow for non-binary networks.

\begin{lemma}
\label{lem:reduction_level_bound}
Let $G=(V,E)$ be a level-$k$ network with $k\geq 1$, and let $H$ be a block of $G$. Let $H'$ be the graph obtained from~$H$ after exhaustively applying the vertex suppression reduction rule from \cref{lem:edge_contraction}. Then, $|V(H')| \leq 4k-1$ and $|E(H')| \leq 5k-2$.
\end{lemma}
\begin{proof}
We can assume that $H'$ is not a single arc, as the lemma then follows trivially. With a careful analysis one can prove the following claim on the types of vertices that can occur in $H'$ (see \cite{holtgrefe2023scanwidth} for the formal proof). Note that $H'$ does not include the leaves of the network or their incident edges, since these form single-arc blocks and they have already been disregarded. 

\emph{Claim:} Every vertex $v \in V(H')$ is of one of the following types: (i) unique \emph{root} with $\deltain(v) = 0$ and $\deltaout(v) \geq 2$; (ii) \emph{flow vertex} with $\deltain(v) = \deltaout(v) = 1$; (iii) \emph{tree-vertex} with $\deltain(v) = 1$ and $\deltaout (v) \geq 2$; (iv) \emph{reticulation vertex} with $\deltain (v) \geq 2$ and $\deltaout(v) \leq 1$.

Now suppose that the unique root of $H'$ has outdegree $\alpha\geq 2$. Furthermore, suppose that $H'$ contains $f$ flow vertices, $t$ tree-vertices with average outdegree $\beta \geq 2$, and $r$ reticulations with average indegree $\gamma_1\geq 2$ and average outdegree $\gamma_2 \leq 1$. By the claim, this covers all possible vertices in $H'$.

The sum of indegrees of all vertices in $H'$ is now $f + t + \gamma_1 r$, while the sum of outdegrees is $\alpha + f + t \beta + \gamma_2 r$. Since these values must be equal, we get that $(\beta -1 ) t = (\gamma_1 - \gamma_2) r - \alpha$, and thus $t \leq \gamma_1 r - \alpha$. Every flow vertex in $H'$ must be the parent of a reticulation, otherwise we would be able to suppress the flow vertex. Furthermore, each reticulation~$v$ can have at most $\deltain (v) - 1$ flow vertices as its parents. Otherwise, if all its parents were flow vertices, we would be able to suppress at least one of them. Thus, we also have that $f \leq (\gamma_1 - 1 ) \cdot r$. Using these two inequalities, we now get for the total number of vertices:
\begin{align*}
|V(H')| &= 1 + f + t + r 
\leq 1 + (\gamma_1 - 1 ) \cdot r + \gamma_1 r - \alpha + r .
\end{align*}
As the total number of arcs equals the sum of indegrees, we also have:
\begin{align*}
|E(H')| = f + t + \gamma_1 r \leq (\gamma_1 - 1 ) \cdot r + \gamma_1 r - \alpha + \gamma_1 r.
\end{align*}

Combining these inequalities with $\gamma_1 \geq 2$ and $\alpha \geq 2$, we obtain $|V(H')|
\leq 4 (\gamma_1 - 1) r  -1$ and $|E(H')| \leq 5 (\gamma_1 - 1) r - 2$.
Now note that
$(\gamma_1 - 1) \cdot r = \sum_{v \in V(H'): \deltain(v) \geq 2} (\deltain(v) - 1)$.
Because $G$ is level-$k$, the graph $H$ must have a reticulation number of at most $k$. Suppressing some indegree-1 outdegree-1 vertices in $H$ will not increase this number. By the definition of the reticulation number from \cref{subsec:phylogenetic_networks}, we thus get that $\sum_{v \in V(H'): \deltain(v) \geq 2} (\deltain(v) - 1) \leq k$. This gives us $(\gamma_1 - 1) r \leq k$. Filling this in into the two upper bounds then proves the lemma.
\end{proof}

\section{Exact methods}\label{sec:exact_methods}
In this section, we focus on methods that compute the scanwidth exactly. From an optimization point of view, we aim to solve the following problem:

\begin{problem}
  \problemtitle{\textsc{Scanwidth}}
  \probleminput{Weakly connected DAG $G$.}
  \problemobjective{Find an extension $\sigma$ of $G$ with minimum scanwidth.}
\end{problem}

Note that we could also define the problem to search for an optimal tree extension (see \cref{def:scanwidth-tree}), but in this section we only consider solution methods that construct optimal extensions. As discussed in \cref{subsec:scanwidth_def} it is possible to create an optimal tree extension from an optimal extension in quadratic time.

As a benchmark, we will first shortly mention an exhaustive search method that solves \textsc{Scanwidth}. The intuitive idea behind the approach is to generate all possible permutations of the vertices in a DAG and check each permutation to determine if it is a valid extension. If it is, we calculate the scanwidth of the extension, and finally, we select the extension with the smallest scanwidth as the optimal solution. The combinatorial explosion of the search space results in such a brute force solution taking $O(n!\cdot n \cdot m)$ time on a weakly connected DAG of $n$ vertices and $m$ arcs.

\subsection{Recursive algorithm}
In this subsection, we develop a recursive algorithm that solves \textsc{Scanwidth}. The main idea is to recursively keep splitting the graph into two (almost) equal-sized parts. The method we will develop is based on an algorithm from Bodlaender et al. \cite{bodlaender2012note}. They present a solution method for a range of ordering problems on undirected graphs, using a technique by Gurevich and Shelah \cite{gurevich1987expected} for the \textsc{Travelling Salesman Problem}. We extend this approach to DAGs and specifically tailor it towards scanwidth. To adapt the approach to DAGs, we introduce the following concept.
\begin{definition}[Ordered partition]
\label{def:ordered_partition}
Let $G=(V,E)$ be a weakly connected DAG and $\PP = (A_1, \ldots, A_r)$ an $r$-partition of $V$, with $r \geq 2$ and the sets $A_i$ allowed to be empty. We call $\PP$ an \emph{ordered $r$-partition} of $V$, if for all $1\leq i < r$, no arcs in $E$ are directed from $\bigcup_{j\leq i} A_j$ towards $\bigcup_{j> i} A_j$.
\end{definition}
Intuitively, an ordered partition respects the DAG's direction: for any vertex $v \in A_j$, all outgoing edges from $v$ go only to sets $A_i$ with $i \le j$. As a consequence, it is possible to concatenate extensions of the induced subgraphs of the ordered partition, resulting in an extension of the entire graph. See \cref{subfig:ordered_partition_part,subfig:ordered_partition_ext} for an illustration of this concept, with $A_1 = L$, $A_2 = W$, and $A_3 = R$, where we use this relabeling for later reference. Note that arcs are allowed to `skip' over parts of the partition and that the set $L$ is forced to be a sinkset.

\begin{figure}[htb]
     \centering
     \parbox{.31\textwidth}{
     \begin{subfigure}[b]{0.32\textwidth}
         \centering
         \begin{tikzpicture}[scale=0.45]
	\begin{pgfonlayer}{nodelayer}
		\node [style={graph_node}, label=right:{$a$}, fill=green!50] (40) at (0, 0) {};
		\node [style={graph_node}, label=right:{$b$}, fill=green!50] (41) at (2, 0) {};
		\node [style={graph_node}, label=right:{$c$}, fill=red!50] (42) at (4, 0) {};
		\node [style={graph_node}, label=right:{$x$}, fill=green!50] (43) at (0, 1.5) {};
		\node [style={graph_node}, label=right:{$y$}, fill=red!50] (44) at (2, 1.5) {};
		\node [style={graph_node}, label=right:{$z$}, fill=red!50] (45) at (4, 1.5) {};
		\node [style={graph_node}, label=right:{$u$}, fill=red!50] (46) at (1, 3) {};
		\node [style={graph_node}, label=right:{$v$}, fill=red!50] (47) at (3, 3) {};
		\node [style={graph_node}, label=right:{$q$}, fill=blue!50] (48) at (2, 4.5) {};
		\node [style={graph_node}, label=right:{$w$}, fill=blue!50] (49) at (5, 3) {};
		\node [style={graph_node}, label=right:{$\rho$}, fill=blue!50] (50) at (3, 6) {};

            \node [style=none] (11) at (-0.5, -0.5) {};
		\node [style=none] (12) at (-0.5, 2) {};
		\node [style=none] (13) at (0.5, 2) {};
		\node [style=none] (14) at (0.5, 0.5) {};
		\node [style=none] (15) at (2.5, 0.5) {};
		\node [style=none] (16) at (2.5, -0.5) {};

		\node [style=none] (17) at (4.5, -0.5) {};
		\node [style=none] (18) at (4.5, 2) {};
		\node [style=none] (19) at (3.5, 3.5) {};
		\node [style=none] (20) at (0.5, 3.5) {};
		\node [style=none] (21) at (0.5, 2.5) {};
		\node [style=none] (22) at (1.5, 1) {};
		\node [style=none] (23) at (3.5, 1) {};
		\node [style=none] (24) at (3.5, -0.5) {};

		\node [style=none] (25) at (2.25, 3.75) {};
		\node [style=none] (26) at (3, 5.25) {};
		\node [style=none] (27) at (4.5, 2.75) {};
		\node [style=none] (28) at (5.35, 2.4) {};
		\node [style=none] (29) at (5.5, 3.25) {};
		\node [style=none] (30) at (3, 6.75) {};
		\node [style=none] (31) at (1.75, 5) {};
		\node [style=none] (32) at (1.25, 4.25) {};

		\node [style=none, label=right:{\large{$L$}}] (33) at (-2, 0.75) {};
		\node [style=none, label=right:{\large{$W$}}] (34) at (4.5, 0.75) {};
		\node [style=none, label=right:{\large{$R$}}] (35) at (4, 4.75) {};

	\end{pgfonlayer}
	\begin{pgfonlayer}{edgelayer}
		\draw [style={graph_edge}] (50) to (48);
		\draw [style={graph_edge}] (48) to (46);
		\draw [style={graph_edge}] (46) to (43);
		\draw [style={graph_edge}] (43) to (40);
		\draw [style={graph_edge}] (48) to (47);
		\draw [style={graph_edge}] (46) to (44);
		\draw [style={graph_edge}] (47) to (44);
		\draw [style={graph_edge}] (44) to (41);
		\draw [style={graph_edge}] (47) to (45);
		\draw [style={graph_edge}] (45) to (42);
		\draw [style={graph_edge}] (50) to (49);
		\draw [style={graph_edge}] (49) to (45);
        \draw [style={graph_edge}, bend right=40] (48) to (43);
	\end{pgfonlayer}

\draw[black!40,fill=gray!10,rounded corners=2mm,opacity=0.35] ($(11.center)$) -- ($(12.center)$) -- ($(13.center)$) -- ($(14.center)$) -- ($(15.center)$) -- ($(16.center)$) -- cycle;

\draw[black!40,fill=gray!10,rounded corners=2mm,opacity=0.35] ($(17.center)$) -- ($(18.center)$) -- ($(19.center)$) -- ($(20.center)$) -- ($(21.center)$) -- ($(22.center)$) -- ($(23.center)$) -- ($(24.center)$) -- cycle;

\draw[black!40,fill=gray!10,rounded corners=2mm,opacity=0.35] ($(25.center)$) -- ($(26.center)$) -- ($(27.center)$) -- ($(28.center)$) -- ($(29.center)$) -- ($(30.center)$) -- ($(31.center)$) -- ($(32.center)$) -- cycle;

\end{tikzpicture}
         \caption{Ordered 3-partition of a DAG}
         \label{subfig:ordered_partition_part}
     \end{subfigure}
     }
     \parbox{.67\textwidth}{
     \begin{subfigure}[b]{0.98\linewidth}
         \centering
         \begin{tikzpicture}[yscale=0.75,xscale=0.85]
	\begin{pgfonlayer}{nodelayer}
		\node [style={graph_node}, fill=green!50, label={below:{$a$}}] (11) at (0, 0) {};
		\node [style={graph_node}, fill=green!50, label={below:{$b$}}] (12) at (2, 0) {};
		\node [style={graph_node}, fill=red!50, label={below:{$c$}}] (13) at (5, 0) {};
		\node [style={graph_node}, fill=green!50, label={below:{$x$}}] (14) at (1, 0) {};
		\node [style={graph_node}, fill=red!50, label={below:{$y$}}] (15) at (3, 0) {};
		\node [style={graph_node}, fill=red!50, label={below:{$z$}}] (16) at (6, 0) {};
		\node [style={graph_node}, fill=red!50, label={below:{$u$}}] (17) at (4, 0) {};
		\node [style={graph_node}, fill=red!50, label={below:{$v$}}] (18) at (7, 0) {};
		\node [style={graph_node}, fill=blue!50, label={below:{$q$}}] (19) at (8, 0) {};
		\node [style={graph_node}, fill=blue!50, label={below:{$w$}}] (20) at (9, 0) {};
		\node [style={graph_node}, fill=blue!50, label={below:{$\rho$}}] (21) at (10, 0) {};
	\end{pgfonlayer}
	\begin{pgfonlayer}{edgelayer}
		\draw [style={graph_edge}, bend left=330] (21) to (19);
		\draw [style={graph_edge}, in=30, out=150, looseness=0.75] (19) to (17);
		\draw [style={graph_edge}, bend left, looseness=0.50] (17) to (14);
		\draw [style={graph_edge}] (14) to (11);
		\draw [style={graph_edge}] (19) to (18);
		\draw [style={graph_edge}] (17) to (15);
		\draw [style={graph_edge}, bend left=330, looseness=0.75] (18) to (15);
		\draw [style={graph_edge}] (15) to (12);
		\draw [style={graph_edge}] (18) to (16);
		\draw [style={graph_edge}] (16) to (13);
		\draw [style={graph_edge}] (21) to (20);
		\draw [style={graph_edge}, bend left, looseness=0.50] (20) to (16);
        \draw [style={graph_edge}, bend right=30, looseness=0.75] (19) to (14);
	\end{pgfonlayer}
\end{tikzpicture}
         \caption{Extension $\sigma$ adhering to the ordered 3-partition}
     \label{subfig:ordered_partition_ext}
     \end{subfigure}
     \begin{subfigure}[b]{0.98\linewidth}
         \centering
         \begin{tikzpicture}[yscale=0.75,xscale=0.85]
	\begin{pgfonlayer}{nodelayer}
		\node [style={graph_node}, fill=green!25, label={below:{$a$}}] (11) at (0, 0) {};
		\node [style={graph_node}, fill=green!25, label={below:{$b$}}] (12) at (2, 0) {};
		\node [style={graph_node}, fill=red!60, label={below:{$c$}}] (13) at (5, 0) {};
		\node [style={graph_node}, fill=green!25, label={below:{$x$}}] (14) at (1, 0) {};
		\node [style={graph_node}, fill=red!60, label={below:{$y$}}] (15) at (3, 0) {};
		\node [style={graph_node}, fill=red!60, label={below:{$z$}}] (16) at (6, 0) {};
		\node [style={graph_node}, fill=red!60, label={below:{$u$}}] (17) at (4, 0) {};
		\node [style={graph_node}, fill=red!60, label={below:{$v$}}] (18) at (7, 0) {};
		\node [style={graph_node}, fill=blue!25, label={below:{$q$}}] (19) at (8, 0) {};
		\node [style={graph_node}, fill=blue!25, label={below:{$w$}}] (20) at (9, 0) {};
		\node [style={graph_node}, fill=blue!25, label={below:{$\rho$}}] (21) at (10, 0) {};
	\end{pgfonlayer}
	\begin{pgfonlayer}{edgelayer}
		\draw [style={graph_edge},dashed,color=black!50, bend left=330] (21) to (19);
		\draw [style={graph_edge}, in=30, out=150, looseness=0.75] (19) to (17);
		\draw [style={graph_edge}, bend left, looseness=0.50] (17) to (14);
		\draw [style={graph_edge},dashed,color=black!50] (14) to (11);
		\draw [style={graph_edge}] (19) to (18);
		\draw [style={graph_edge}] (17) to (15);
		\draw [style={graph_edge}, bend left=330, looseness=0.75] (18) to (15);
		\draw [style={graph_edge}] (15) to (12);
		\draw [style={graph_edge}] (18) to (16);
		\draw [style={graph_edge}] (16) to (13);
		\draw [style={graph_edge},dashed,color=black!50] (21) to (20);
		\draw [style={graph_edge}, bend left, looseness=0.50] (20) to (16);
        \draw [style={graph_edge}, bend right=30, looseness=0.75] (19) to (14);
	\end{pgfonlayer}
	
\draw[black,fill=red!40,rounded corners=2mm,opacity=0.35] ($(15.north west)+(-0.2,0.6)$) -- ($(18.north east)+(0.5,0.6)$) -- ($(18.south east)+(0.5,-0.67)$) -- ($(15.south west)+(-0.2,-0.67)$) -- cycle;
\node [style=none] (99) at (0, 0.8) {};

\end{tikzpicture}
         \caption{Visualization of partial scanwidth}
     \label{subfig:ordered_partition_window}
     \end{subfigure}
     }
        \caption{(a): A weakly connected DAG $G = (V,E)$, with the colors indicating an ordered 3-partition $\PP =(L,W,R)$ of $V$. (b): An extension $\sigma$ of $G$ which is a concatenation of extensions of the three subgraphs induced by the partition $\PP$. (c): The same extension $\sigma$ but with a `window' drawn around the vertices in~$W$, aiding the interpretation of partial scanwidth. The arcs between a pair of vertices within~$L$ or a pair vertices within~$R$ are grey and dashed because they never count towards the partial scanwidth for this ordered 3-partition.}
        \label{fig:ordered_partition}
\end{figure}

Our recursive approach uses a natural generalization of scanwidth: \emph{partial scanwidth}. This concept allows us to analyze the scanwidth of only a subset of the vertices of a graph. This will be useful to break down the problem into smaller subproblems. By solving these subproblems recursively, we can build up the scanwidth of the entire graph. Before we state the formal definition, recall that $\Pi [W]$ is the set of all extensions of $G[W]$. 

\begin{definition}[Partial scanwidth]\label{def:psw}
Let $G=(V,E)$ be a weakly connected DAG and $\PP = (L, W, R)$ an ordered 3-partition of $V$ such that $W \neq \emptyset$. For $\sigma \in \Pi[W]$ and a position $i$ of $\sigma$, we will denote $$\Pdash\PSW_i^\sigma(G) = \{ uv \in E: u \in \sigma[i+1 \ldots ]\cup R, v \connect{G[\sigma[1 \ldots i] \cup L]} \sigma (i) \}.$$ Then the \emph{partial scanwidth} of $G$ for $\PP$ is 
$$\Pdash\psw (G) = \min_{\sigma \in \Pi[W]} \max_{i \in W} |\Pdash\PSW_i^\sigma(G)|.$$
Furthermore, we let $\Pdash\psw (\sigma, G) = \max_{i \in W} |\Pdash\PSW_i^\sigma( G)|$ be the partial scanwidth of $\sigma$ for $\PP$. If $G$ is clear from the context, we mostly write  $\Pdash \PSW_i^\sigma$ instead of $\Pdash\PSW_i^\sigma(G)$.
\end{definition}

Essentially, partial scanwidth only considers the scanwidth within a `window' $W$ of the vertices of $G$, while assuming the vertices in the set $L$ to be positioned on the left of $W$ in an extension, and those in the set $R$ to be right of $W$. We emphasize that arcs from $R$ to $L$ may belong to a set $\Pdash\PSW_i^\sigma(G)$, whereas arcs between pairs of vertices within $R$ or within $L$ are never part of these sets. \cref{subfig:ordered_partition_window} clarifies this concept.

We can now proceed with the main recursive idea of the algorithm. It builds upon Lemma~4 from \cite{bodlaender2012note}, which presents a similar concept for more general functions on undirected graphs. The following lemma is in some sense an adapted version for DAGs, specifically tailored to (partial) scanwidth. See \cref{fig:recursive_alg} for an accompanying illustration with a specific choice of~$W'$.

\begin{figure}[htb]
     \centering
     \begin{subfigure}[b]{0.32\textwidth}
         \centering
         \begin{tikzpicture}[scale=0.45]
  \begin{pgfonlayer}{nodelayer}
    \node [style={graph_node}, label=right:{$a$}, fill=green!50] (40b) at (0, 0) {};
    \node [style={graph_node}, label=right:{$b$}, fill=green!50] (41b) at (2, 0) {};
    \node [style={graph_node}, label=right:{$c$}, fill=red!50] (42b) at (4, 0) {};
    \node [style={graph_node}, label=right:{$x$}, fill=green!50] (43b) at (0, 1.5) {};
    \node [style={graph_node}, label=right:{$y$}, fill=red!50] (44b) at (2, 1.5) {};
    \node [style={graph_node}, label=right:{$z$}, fill=red!50] (45b) at (4, 1.5) {};
    \node [style={graph_node}, label=right:{$u$}, fill=red!50] (46b) at (1, 3) {};
    \node [style={graph_node}, label=right:{$v$}, fill=red!50] (47b) at (3, 3) {};
    \node [style={graph_node}, label=right:{$q$}, fill=blue!50] (48b) at (2, 4.5) {};
    \node [style={graph_node}, label=right:{$w$}, fill=blue!50] (49b) at (5, 3) {};
    \node [style={graph_node}, label=right:{$\rho$}, fill=blue!50] (50b) at (3, 6) {};

    \node [style=none] (11b) at (-0.5, -0.5) {};
    \node [style=none] (12b) at (-0.5, 2) {};
    \node [style=none] (13b) at (0.5, 2) {};
    \node [style=none] (14b) at (0.5, 0.5) {};
    \node [style=none] (15b) at (2.5, 0.5) {};
    \node [style=none] (16b) at (2.5, -0.5) {};

    \node [style={special_node2}] (31) at (3.75, 1.75) {};
    \node [style={special_node2}] (32) at (4.25, 1.75) {};
    \node [style={special_node2}] (33) at (4.25, -0.25) {};
    \node [style={special_node2}] (34) at (3.75, -0.25) {};

    \node [style={special_node}] (65) at (0.75, 3.25) {};
    \node [style={special_node}] (66) at (3.25, 3.25) {};
    \node [style={special_node}] (67) at (3.25, 2.75) {};
    \node [style={special_node}] (68) at (2.25, 1.25) {};
    \node [style={special_node}] (69) at (1.75, 1.25) {};
    \node [style={special_node}] (70) at (0.75, 2.75) {};

    \node [style=none] (56) at (0.5, 3.5) {};
    \node [style=none] (57) at (3.5, 3.5) {};
    \node [style=none] (58) at (4.5, 2) {};
    \node [style=none] (59) at (4.5, -0.5) {};
    \node [style=none] (60) at (3.5, -0.5) {};
    \node [style=none] (61) at (3.5, 1) {};
    \node [style=none] (62) at (2.5, 1) {};
    \node [style=none] (63) at (1.5, 1) {};
    \node [style=none] (64) at (0.5, 2.5) {};

    \node [style=none] (25b) at (2.25, 3.75) {};
    \node [style=none] (26b) at (3, 5.25) {};
    \node [style=none] (27b) at (4.5, 2.75) {};
    \node [style=none] (28b) at (5.35, 2.4) {};
    \node [style=none] (29b) at (5.5, 3.25) {};
    \node [style=none] (30b) at (3, 6.75) {};
    \node [style=none] (31b) at (1.75, 5) {};
    \node [style=none] (32b) at (1.25, 4.25) {};

    \node [style=none, label=right:{\large{$L$}}] (33b) at (-2, 0.75) {};
    \node [style=none, label=left:{\normalsize{$W\setminus W'$}}] (34b) at (0.55, 3.0) {};
    \node [style=none, label=right:{\normalsize{$W'$}}] (35b) at (4.25, 0.75) {};
    \node [style=none, label=right:{\large{$R$}}] (36b) at (4, 4.75) {};
  \end{pgfonlayer}

  \begin{pgfonlayer}{edgelayer}
    \draw [style={graph_edge}] (50b) to (48b);
    \draw [style={graph_edge}] (48b) to (46b);
    \draw [style={graph_edge}] (46b) to (43b);
    \draw [style={graph_edge}] (43b) to (40b);
    \draw [style={graph_edge}] (48b) to (47b);
    \draw [style={graph_edge}] (46b) to (44b);
    \draw [style={graph_edge}] (47b) to (44b);
    \draw [style={graph_edge}] (44b) to (41b);
    \draw [style={graph_edge}] (47b) to (45b);
    \draw [style={graph_edge}] (45b) to (42b);
    \draw [style={graph_edge}] (50b) to (49b);
    \draw [style={graph_edge}] (49b) to (45b);
    \draw [style={graph_edge}, bend right=40] (48b) to (43b);
  \end{pgfonlayer}

  \draw[black!40,fill=gray!10,rounded corners=2mm,opacity=0.35]
    ($(11b.center)$) -- ($(12b.center)$) -- ($(13b.center)$) -- ($(14b.center)$) -- ($(15b.center)$) -- ($(16b.center)$) -- cycle;

  \draw[black!40,fill=gray!65,rounded corners=2mm,opacity=0.35]
    ($(65.center)$) -- ($(66.center)$) -- ($(67.center)$) -- ($(68.center)$) -- 
    ($(69.center)$) -- ($(70.center)$) -- cycle;

  \draw[black!40,fill=gray!10,rounded corners=2mm,opacity=0.35]
    ($(56.center)$) -- ($(57.center)$) -- ($(58.center)$) -- ($(59.center)$) -- ($(60.center)$) -- ($(61.center)$) -- ($(62.center)$) -- ($(63.center)$) -- ($(64.center)$) -- cycle;

  \draw[black!40,fill=gray!65,rounded corners=2mm,opacity=0.35]
    ($(31.center)$) -- ($(32.center)$) -- ($(33.center)$) -- ($(34.center)$) -- cycle;

  \draw[black!40,fill=gray!10,rounded corners=2mm,opacity=0.35]
    ($(25b.center)$) -- ($(26b.center)$) -- ($(27b.center)$) -- ($(28b.center)$) -- ($(29b.center)$) -- ($(30b.center)$) -- ($(31b.center)$) -- ($(32b.center)$) -- cycle;
\end{tikzpicture}
         \caption{Ordered 3-partition of a DAG}
         \label{subfig:recursive_alg1}
     \end{subfigure}
     \begin{subfigure}[b]{0.67\textwidth}
         \centering
         \begin{tikzpicture}[scale=0.45]
  \begin{pgfonlayer}{nodelayer}
    \node [style={graph_node}, label=right:{$a$}, fill=green!50] (40r1) at (0, 0) {};
    \node [style={graph_node}, label=right:{$b$}, fill=green!50] (41r1) at (2, 0) {};
    \node [style={graph_node}, label=right:{$c$}, fill=red!50] (42r1) at (4, 0) {};
    \node [style={graph_node}, label=right:{$x$}, fill=green!50] (43r1) at (0, 1.5) {};
    \node [style={graph_node}, label=right:{$y$}, fill=blue!50] (44r1) at (2, 1.5) {};
    \node [style={graph_node}, label=right:{$z$}, fill=red!50] (45r1) at (4, 1.5) {};
    \node [style={graph_node}, label=right:{$u$}, fill=blue!50] (46r1) at (1, 3) {};
    \node [style={graph_node}, label=right:{$v$}, fill=blue!50] (47r1) at (3, 3) {};
    \node [style={graph_node}, label=right:{$q$}, fill=blue!50] (48r1) at (2, 4.5) {};
    \node [style={graph_node}, label=right:{$w$}, fill=blue!50] (49r1) at (5, 3) {};
    \node [style={graph_node}, label=right:{$\rho$}, fill=blue!50] (50r1) at (3, 6) {};

    \node [style=none] (r1n36) at (2.5, 1) {};
    \node [style=none] (r1n37) at (3.5, 2.5) {};
    \node [style=none] (r1n38) at (5.25, 2.5) {};
    \node [style=none] (r1n39) at (5.5, 3.25) {};
    \node [style=none] (r1n40) at (3, 6.75) {};
    \node [style=none] (r1n41) at (1.75, 5) {};
    \node [style=none] (r1n42) at (1.25, 4.25) {};
    \node [style=none] (r1n43) at (0.5, 3.5) {};
    \node [style=none] (r1n44) at (0.5, 2.5) {};
    \node [style=none] (r1n45) at (1.5, 1) {};

    \node [style=none] (r1L11) at (-0.5, -0.5) {};
    \node [style=none] (r1L12) at (-0.5, 2) {};
    \node [style=none] (r1L13) at (0.5, 2) {};
    \node [style=none] (r1L14) at (0.5, 0.5) {};
    \node [style=none] (r1L15) at (2.5, 0.5) {};
    \node [style=none] (r1L16) at (2.5, -0.5) {};

    \node [style=none] (r1Wp17) at (3.5, -0.5) {};
    \node [style=none] (r1Wp18) at (3.5, 2.0) {};
    \node [style=none] (r1Wp19) at (4.5, 2.0) {};
    \node [style=none] (r1Wp20) at (4.5, -0.5) {};

    \node [style=none] (r1L11) at (-0.5, -0.5) {};
    \node [style=none] (r1L12) at (-0.5, 2) {};
    \node [style=none] (r1L13) at (0.5, 2) {};
    \node [style=none] (r1L14) at (0.5, 0.5) {};
    \node [style=none] (r1L15) at (2.5, 0.5) {};
    \node [style=none] (r1L16) at (2.5, -0.5) {};

    \node [style=none] (r1Wp17) at (3.5, -0.5) {};
    \node [style=none] (r1Wp18) at (3.5, 2.0) {};
    \node [style=none] (r1Wp19) at (4.5, 2.0) {};
    \node [style=none] (r1Wp20) at (4.5, -0.5) {};

    \node [style=none, label=right:{\large{$L$}}] (33r1) at (-2, 0.75) {};
    \node [style=none, label=left:{\large{$R\cup (W\setminus W')$}}] (34r1) at (1.5, 4.9) {};
    \node [style=none, label=right:{\large{$W'$}}] (35r1) at (4.75, 0.75) {};
  \end{pgfonlayer}

  \begin{pgfonlayer}{edgelayer}
    \draw [style={graph_edge}] (50r1) to (48r1);
    \draw [style={graph_edge}] (48r1) to (46r1);
    \draw [style={graph_edge}] (46r1) to (43r1);
    \draw [style={graph_edge}] (43r1) to (40r1);
    \draw [style={graph_edge}] (48r1) to (47r1);
    \draw [style={graph_edge}] (46r1) to (44r1);
    \draw [style={graph_edge}] (47r1) to (44r1);
    \draw [style={graph_edge}] (44r1) to (41r1);
    \draw [style={graph_edge}] (47r1) to (45r1);
    \draw [style={graph_edge}] (45r1) to (42r1);
    \draw [style={graph_edge}] (50r1) to (49r1);
    \draw [style={graph_edge}] (49r1) to (45r1);
    \draw [style={graph_edge}, bend right=40] (48r1) to (43r1);
  \end{pgfonlayer}

  \draw[black!40,fill=gray!10,rounded corners=2mm,opacity=0.35]
    ($(r1n43.center)$) -- ($(r1n44.center)$) -- ($(r1n45.center)$) -- ($(r1n36.center)$) -- ($(r1n37.center)$) -- ($(r1n38.center)$) -- ($(r1n39.center)$) -- ($(r1n40.center)$) -- ($(r1n41.center)$) -- ($(r1n42.center)$) -- cycle;

  \draw[black!40,fill=gray!10,rounded corners=2mm,opacity=0.35]
    ($(r1L11.center)$) -- ($(r1L12.center)$) -- ($(r1L13.center)$) -- ($(r1L14.center)$) -- ($(r1L15.center)$) -- ($(r1L16.center)$) -- cycle;

  \draw[black!40,fill=gray!10,rounded corners=2mm,opacity=0.35]
    ($(r1Wp17.center)$) -- ($(r1Wp18.center)$) -- ($(r1Wp19.center)$) -- ($(r1Wp20.center)$) -- cycle;
\end{tikzpicture}
         \quad
         \begin{tikzpicture}[scale=0.45]
  \begin{pgfonlayer}{nodelayer}
    \node [style={graph_node}, label=right:{$a$}, fill=green!50] (40r2) at (0, 0) {};
    \node [style={graph_node}, label=right:{$b$}, fill=green!50] (41r2) at (2, 0) {};
    \node [style={graph_node}, label=right:{$c$}, fill=green!50] (42r2) at (4, 0) {};
    \node [style={graph_node}, label=right:{$x$}, fill=green!50] (43r2) at (0, 1.5) {};
    \node [style={graph_node}, label=right:{$y$}, fill=red!50] (44r2) at (2, 1.5) {};
    \node [style={graph_node}, label=right:{$z$}, fill=green!50] (45r2) at (4, 1.5) {};
    \node [style={graph_node}, label=right:{$u$}, fill=red!50] (46r2) at (1, 3) {};
    \node [style={graph_node}, label=right:{$v$}, fill=red!50] (47r2) at (3, 3) {};
    \node [style={graph_node}, label=right:{$q$}, fill=blue!50] (48r2) at (2, 4.5) {};
    \node [style={graph_node}, label=right:{$w$}, fill=blue!50] (49r2) at (5, 3) {};
    \node [style={graph_node}, label=right:{$\rho$}, fill=blue!50] (50r2) at (3, 6) {};

    \node [style=none] (r2n15) at (2.5, 0.5) {};
    \node [style=none] (r2n25) at (3.5, 2) {};
    \node [style=none] (r2n26) at (4.5, 2) {};
    \node [style=none] (r2n27) at (4.5, -0.5) {};
    \node [style=none] (r2n28) at (3.5, -0.5) {};
    \node [style=none] (r2n29) at (2.5, -0.5) {};
    \node [style=none] (r2n30) at (-0.5, -0.5) {};
    \node [style=none] (r2n31) at (-0.5, 2) {};
    \node [style=none] (r2n32) at (0.5, 2) {};
    \node [style=none] (r2n33) at (0.5, 0.5) {};
    \node [style=none] (r2n34) at (3.5, 0.5) {};

    \node [style=none] (21r2) at (1.8, 1.0) {};
    \node [style=none] (22r2) at (0.3, 3.4) {};
    \node [style=none] (23r2) at (3.7, 3.4) {};
    \node [style=none] (24r2) at (2.2, 1.0) {};

    \node [style=none] (25r2) at (2.25, 3.75) {};
    \node [style=none] (26r2) at (3, 5.25) {};
    \node [style=none] (27r2) at (4.5, 2.75) {};
    \node [style=none] (28r2) at (5.35, 2.4) {};
    \node [style=none] (29r2) at (5.5, 3.25) {};
    \node [style=none] (30r2) at (3, 6.75) {};
    \node [style=none] (31r2) at (1.75, 5) {};
    \node [style=none] (32r2) at (1.25, 4.25) {};

    \node [style=none, label=right:{\large{$R$}}] (33r2) at (4, 4.75) {};
    \node [style=none, label=left:{\large{$W\setminus W'$}}] (34r2) at (0.6, 3.05) {};
    \node [style=none, label=right:{\large{$L\cup W'$}}] (35r2) at (4.75, 0.75) {};
  \end{pgfonlayer}

  \begin{pgfonlayer}{edgelayer}
    \draw [style={graph_edge}] (50r2) to (48r2);
    \draw [style={graph_edge}] (48r2) to (46r2);
    \draw [style={graph_edge}] (46r2) to (43r2);
    \draw [style={graph_edge}] (43r2) to (40r2);
    \draw [style={graph_edge}] (48r2) to (47r2);
    \draw [style={graph_edge}] (46r2) to (44r2);
    \draw [style={graph_edge}] (47r2) to (44r2);
    \draw [style={graph_edge}] (44r2) to (41r2);
    \draw [style={graph_edge}] (47r2) to (45r2);
    \draw [style={graph_edge}] (45r2) to (42r2);
    \draw [style={graph_edge}] (50r2) to (49r2);
    \draw [style={graph_edge}] (49r2) to (45r2);
    \draw [style={graph_edge}, bend right=40] (48r2) to (43r2);
  \end{pgfonlayer}

  \draw[black!40,fill=gray!10,rounded corners=2mm,opacity=0.35]
    ($(r2n34.center)$) -- ($(r2n25.center)$) -- ($(r2n26.center)$) -- ($(r2n27.center)$) -- ($(r2n28.center)$) -- ($(r2n29.center)$) -- ($(r2n30.center)$) -- ($(r2n31.center)$) -- ($(r2n32.center)$) -- ($(r2n33.center)$) -- ($(r2n15.center)$) -- cycle;

  \draw[black!40,fill=gray!10,rounded corners=2mm,opacity=0.35]
    ($(21r2.center)$) -- ($(22r2.center)$) -- ($(23r2.center)$) -- ($(24r2.center)$) -- cycle;

  \draw[black!40,fill=gray!10,rounded corners=2mm,opacity=0.35]
    ($(25r2.center)$) -- ($(26r2.center)$) -- ($(27r2.center)$) -- ($(28r2.center)$) -- ($(29r2.center)$) -- ($(30r2.center)$) -- ($(31r2.center)$) -- ($(32r2.center)$) -- cycle;
\end{tikzpicture}
         \caption{Two ordered 3-partitions of the same DAG}
     \label{subfig:recursive_alg2}
     \end{subfigure}
        \caption{Example illustrating an iteration of \cref{alg:recursion}. (a): The weakly connected DAG $G = (V,E)$ and the ordered 3-partition $\PP =(L,W,R)$ of $V$ from \cref{subfig:ordered_partition_part}. The set $W$ is split further into $W'$ and $W \setminus W'$, with  $|W'| = \left \lfloor{|W| / 2} \right \rfloor$ and no arc directed from  $W'$ to $W \setminus W'$. (b): The two ordered 3-partitions $(L,W',R\cup (W\backslash W'))$ and $(L \cup W' ,W \backslash W',R)$ used in both \cref{lem:psw_recursion}(b) and \cref{alg:recursion}.}
        \label{fig:recursive_alg}
\end{figure}

\begin{lemma}
\label{lem:psw_recursion}
Let $G=(V,E)$ be a weakly connected DAG and $\PP = (L,W,R)$ an ordered 3-partition of $V$, with $W \neq \emptyset$.

\begin{enumerate}[label={(\alph*)}, noitemsep,topsep=3pt]
\item If $|W| = 1$ (with $W = \{w \}$), then $\Pdash \psw (G) = |\{uv \in E: u \in R, v\connect{G[L\cup W]} w\}|.$

\item If $|W| \geq 2$, then for any $1 \leq k < |W|$,
$$\Pdash\psw (G) = \min_{W' \in \mathcal{W}_k}  \max \Big\{ \PP_1^{W'}\sdash\psw (G) , \PP_2^{W'}\sdash\psw (G)  \Big\},$$
where $\mathcal{W}_k = \{W' \subseteq W: |W'| =k \text{ and no arc in } E \text{ is directed from } W' \text{ to } W \setminus W'\}$ and for all $W' \in \mathcal{W}_k$ we let $\PP_1^{W'} = (L,W',R\cup (W\backslash W'))$ and $\PP_2^{W'} = (L \cup W' ,W \backslash W',R)$.
\end{enumerate}
\end{lemma}
\begin{proof}
\emph{(a):} By \cref{def:psw} and the fact that $G[W]$ only has a single extension $(w)$, we get \begin{align*}
\Pdash\psw(G) &= \min_{\sigma \in \Pi[W]} \max_{i \in W} |\Pdash\PSW_i^\sigma| = |\Pdash\PSW_w^{(w)}| 
= | \{ uv \in E(G): u \in R,    v\connect{G[L\cup W]} w \}|.
\end{align*}

\emph{(b):} First note that both $\PP_1^{W'}$ and $\PP_2^{W'}$ are also ordered 3-partitions of $V$ with the middle set non-empty. Now, we prove an equality that is at the core of the result. Let $W' \in \mathcal{W}_k$ be arbitrary, and let $\sigma_1 \in \Pi[W']$, $\sigma_2 \in \Pi[W \setminus W']$ and $\sigma = \sigma_1 \circ \sigma_2  \in \Pi[W]$. We then have that 
\begin{align}
\Pdash\psw(\sigma, G) &= \max \left\{\max_{w \in W'} \Pdash\PSW_w^\sigma , \max_{w \in W\setminus W'} \Pdash\PSW_w^\sigma \right\} \nonumber \\
&= \max \left\{\max_{w \in W'} \PP_1^{W'}\sdash\PSW_w^{\sigma_1} , \max_{w \in W\setminus W'} \PP_2^{W'}\sdash\PSW_w^{\sigma_2} \right\} \nonumber \\
&= \max \left\{ \PP_1^{W'}\sdash\psw( \sigma_1, G),  \PP_2^{W'}\sdash\psw( \sigma_2, G) \right\}. \label{eq:psw}
\end{align}
Here, the second equality is essential. It uses the important observation that if we only maximize over a consecutive subsequence of vertices in $\sigma$, we can just as well put the vertices to the left of this subsequence in the set $L$ and the ones to the right in the set $R$. We are now ready to prove the lemma.

($\leq$) Let $W' \in \mathcal{W}_k$ be arbitrary. Furthermore, let $\sigma_1 \in \Pi [W']$ be such that $\PP_1^{W'}\sdash\psw(\sigma_1, G) = \PP_1^{W'}\sdash\psw(G)$ and $\sigma_2 \in \Pi [W \setminus W']$ be such that $\PP_2^{W'}\sdash\psw(\sigma_2, G) = \PP_2^{W'}\sdash\psw (G)$. Both exist by definition of the partial scanwidth. We now define $\sigma = \sigma_1 \circ \sigma_2 \in \Pi[W]$. Using \cref{eq:psw}, we then have that
$\Pdash\psw (G) \leq \max \left\{ \PP_1^{W'}\sdash\psw(G),  \PP_2^{W'}\sdash\psw (G) \right\}$.
Because $W' \in \mathcal{W}_k$ was arbitrary, we obtain
$$\Pdash\psw (G) \leq \min_{W' \in \mathcal{W}_k}  \max \Big\{ \PP_1^{W'}\sdash\psw (G) , \PP_2^{W'}\sdash\psw (G)  \Big\}.$$

($\geq$) Let $\sigma \in \Pi[W]$ be such that $\Pdash\psw (G) = \Pdash\psw (\sigma, G)$, which exists by \cref{def:psw}. Now choose $W''$ to be the set of the first $k$ vertices of $\sigma$ (clearly $W'' \in \mathcal{W}_k$). We denote by $\sigma_1$ the ordering consisting of the first $k$ vertices of $\sigma$ (in the same order). Similarly, $\sigma_2$ denotes the $|W| - k$ other vertices (again keeping the order). Thus, $\sigma = \sigma_1 \circ \sigma_2$, with $\sigma_1 \in \Pi[W'']$ and $\sigma_2 \in \Pi[W \setminus W'']$. Then, again using \cref{eq:psw}, $\Pdash\psw (G) \geq\max \{ \PP_1^{W''}\sdash\psw(G),  \PP_2^{W''}\sdash\psw( G) \}$.
As $W''$ was an element of $\mathcal{W}_k$, we can minimize over all $W' \in\mathcal{W}_k$ to obtain
$$\Pdash\psw (G) \geq \min_{W' \in \mathcal{W}_k}  \max \Big\{ \PP_1^{W'}\sdash\psw (G) , \PP_2^{W'}\sdash\psw (G)  \Big\},$$ which proves the lemma.
\end{proof}

With this at first glance quite complicated recursive relation, we can formulate the relatively concise and elegant \cref{alg:recursion} that solves \textsc{Scanwidth}, by setting $k$ equal to half the size of the set $W$. In this way, we will be able to bound the number of 3-partitions that are considered. See again \cref{fig:recursive_alg} for an example of the recursion of the algorithm for a specific choice of~$W'$ in the for-loop.  Note that in a practical implementation of the algorithm (and also in the subsequent algorithm), we can replace the value $\infty$ with $|E| + 1$: the trivial upper bound of the scanwidth.

\begin{algorithm}[htb]
\caption{Recursive algorithm to solve \textsc{Scanwidth}.}
\label{alg:recursion}
\Input{Weakly connected DAG $G=(V,E)$.}
\Output{Scanwidth $\sw$ of $G$, optimal extension $\sigma_{\mathrm{opt}}$.}
\SetKwFunction{partsw}{PartialScanwidth}
$\sw, \sigma_{\mathrm{opt}} \gets \partsw( \emptyset, V, \emptyset)$ \\
\Return{$\sw, \sigma_{\mathrm{opt}}$}\\
\setcounter{AlgoLine}{0}
\nonl \Proc{$\partsw(L, W, R)$}{
\KwInit{$\psw \gets \infty$; $\sigma \gets ()$}\\
\If{$|W|=1$ with $W = \{w\}$}{
$\psw \gets |\{uv \in E: u \in R, v\connect{G[L\cup W]} w \}|$\\
$\sigma \gets (v)$
}
\ElseIf{$|W| > 1$}{
\For{$W' \subseteq W : |W'| = \left \lfloor{\frac{|W|}{2}}\right \rfloor \text{ and no arc in } E \text{ is directed from } W' \text{ to } W \setminus W'$}
{
$\psw_1', \sigma_1' \gets \partsw (L,W',R\cup (W\setminus W')) $ \label{line:psw1}\\
$\psw_2', \sigma_2' \gets \partsw (L\cup W',W\setminus W',R) $\label{line:psw2}\\
$\psw' \gets \max \{ \psw_1', \psw_2' \}$\\
\If{$\psw' < \psw$}{$\psw \gets \psw'$\\
$\sigma \gets \sigma_1' \circ \sigma_2'$
}
}
}
\Return{$\psw, \sigma$}}
\end{algorithm}

It turns out that \cref{alg:recursion} runs in $\tilde{O}(4^n)$ time, which is a major improvement over the earlier discussed brute force solution running in $\tilde{O} (n!)$ time. In the following theorem, we prove correctness of the algorithm and its time complexity.

\begin{theorem}
\label{thm:recursive_algorithm}
Let $G=(V,E)$ be a weakly connected DAG of $n$ vertices and $m$ arcs, then \cref{alg:recursion} solves \textsc{Scanwidth} in $\tilde{O} (4^n)$ time and polynomial space.
\end{theorem}
\begin{proof}
\emph{Correctness:}
The correctness of the partial scanwidth that is returned by the subroutine (and consequently the correctness of the scanwidth of $G$) follows directly from \cref{lem:psw_recursion}, while the correctness of the corresponding extension is a direct consequence of the constructive nature of the proof of that lemma. The algorithm terminates, as each recursive call is made for a strictly smaller set $W$.

\emph{Time complexity:} The following analysis mimics the analysis of the recursive algorithm in \cite{bodlaender2012note}. Let $T(m, k)$ denote the time it takes to run the subroutine \texttt{PartialScanwidth} for a set $W$ with $|W|=k$, and with $m$ the number of arcs of $G$. If $k\geq 2$, then we loop over at most all subsets of size $\left \lfloor{k / 2}\right \rfloor $ of the set $W$. There are $\binom{k}{\left \lfloor{k / 2}\right \rfloor }$ such subsets. For each of these subsets, we have two recursive calls: one for a set of size $\left \lfloor{k / 2}\right \rfloor$ and one for a set of size $k - \left \lfloor{k / 2}\right \rfloor = \left \lceil{k / 2}\right \rceil$. Outside of the for-loop, we do some work in $O(m)$ time. Furthermore, each iteration of the for-loop can also be performed in $O(m)$ time per recursive call. Overall, there exists some constant $c\geq 0$, such that all these computations are bounded by $c \cdot m$. Thus, we obtain the following recurrence relation:
$$\begin{dcases}
  T(m,1) \leq c \cdot m,  & \text{ if } k = 1; \\
  T(m,k) \leq \binom{k}{\left \lfloor{k / 2}\right \rfloor } \big( T(m,\left \lfloor{k / 2}\right \rfloor ) + T(m,\left \lceil{k / 2}\right \rceil ) \big) + c \cdot m  , & \text{ if } k \geq 2.
\end{dcases}$$

As in \cite{bodlaender2012note}, it then follows that $T(m, k) = \tilde{O} (4^k)$. Since the algorithm runs in $T(m,n)$ time, we obtain the time complexity of $\tilde{O} (4^n)$.

\emph{Space complexity:} The recursion depth of the algorithm is $O ( \log n)$, due to the sets $W$ being split in half. Furthermore, within each recursive step, only polynomial space is used. Therefore, the complete algorithm uses polynomial space. (See also \cite{bodlaender2012exact}, where the same explanation is given for a specific case of the algorithm from \cite{bodlaender2012note}, applied to treewidth.)
\end{proof}

\subsection{Dynamic programming}
In this subsection, we employ a dynamic programming approach to efficiently solve \textsc{Scanwidth} in polynomial time when the value of the scanwidth is bounded by a constant. More formally, we try to solve the following fixed-parameter problem:

\begin{problem}
  \problemtitle{$k$-\textsc{Scanwidth}}
  \probleminput{Weakly connected DAG $G$.}
  \problemobjective{Find an extension $\sigma$ of $G$ with a scanwidth of at most $k$, if it exists. Otherwise, certify that the scanwidth of $G$ is larger than $k$.}
\end{problem}

Part of the algorithm we will present shows some similarities with a dynamic programming algorithm from \cite{bodlaender2012note}, whose authors adapt a classical technique by Held and Karp \cite{held1962dynamic} for the \textsc{Travelling Salesman Problem} to address general ordering problems on undirected graphs. Particularly, both our and their algorithm consider subsets (or equivalently, 2-partitions) of the vertex set of a graph at most once, and recursively decrease the size of the considered sets by one. We introduce some new machinery specific to the scanwidth to further bound the time complexity, and provide an interpretation of the time complexity in terms of roots of subgraphs of the graph.

Our new algorithm neatly fits within the framework of ordered partitions and the partial scanwidth from before. Specifically, we consider ordered 2-partitions instead of ordered 3-partitions. This leads to a restricted case of the partial scanwidth where $L$ (the set of vertices to the left of the set $W$) is empty.
\begin{definition}[Restricted partial scanwidth]
\label{def:rpsw}
Let $G$ be a weakly connected DAG and $\QQ = (W, R)$ an ordered 2-partition of $V$ such that $W \neq \emptyset$. The \emph{restricted partial scanwidth} of $G$ for $\QQ$ is 
$$\Qdash\rpsw (G) = \Pdash \psw (G),$$ where $\PP = (\emptyset, W, R)$. Similarly, we define $\Qdash\rpsw(\sigma, G) = \Pdash\psw(\sigma, G)$, and $\Qdash\RPSW_i^\sigma (G) = \Pdash\PSW_i^\sigma( G)$.
\end{definition}
Note that if $(W, R)$ is an ordered 2-partition, then $(\emptyset, W, R)$ is indeed an ordered 3-partition of a graph. Hence, in light of \cref{def:psw}, the notions from \cref{def:rpsw} are well-defined. Furthermore, we emphasize that a partition $(W, R)$ being an ordered 2-partition means nothing more than that $W$ is a sinkset of the graph, since arcs are not allowed to be oriented from $W$ to $R$ (see \cref{subfig:dp_alg1}).

\begin{figure}[htb]
     \centering
     \begin{subfigure}[b]{0.32\textwidth}
         \centering
         \begin{tikzpicture}[scale=0.45]
  \begin{pgfonlayer}{nodelayer}
    \node [style={graph_node}, label=right:{$a$}, fill=red!50] (dp1a) at (0, 0) {};
    \node [style={graph_node}, label=right:{$b$}, fill=red!50] (dp1b) at (2, 0) {};
    \node [style={graph_node}, label=right:{$c$}, fill=red!50] (dp1c) at (4, 0) {};
    \node [style={graph_node}, label=right:{$x$}, fill=red!50] (dp1x) at (0, 1.5) {};
    \node [style={graph_node}, label=right:{$y$}, fill=red!50] (dp1y) at (2, 1.5) {};
    \node [style={graph_node}, label=right:{$z$}, fill=red!50] (dp1z) at (4, 1.5) {};
    \node [style={graph_node}, label=right:{$u$}, fill=red!50] (dp1u) at (1, 3) {};
    \node [style={graph_node}, label=right:{$v$}, fill=blue!50] (dp1v) at (3, 3) {};
    \node [style={graph_node}, label=right:{$q$}, fill=blue!50] (dp1q) at (2, 4.5) {};
    \node [style={graph_node}, label=right:{$w$}, fill=red!50] (dp1w) at (5, 3) {};
    \node [style={graph_node}, label=right:{$\rho$}, fill=blue!50] (dp1rho) at (3, 6) {};

    \node [style=none] (38) at (4.5, -0.5) {};
    \node [style=none] (39) at (-0.5, -0.5) {};
    \node [style=none] (40) at (-0.5, 2) {};
    \node [style=none] (41) at (0.5, 3.5) {};
    \node [style=none] (42) at (1.5, 3.5) {};
    \node [style=none] (43) at (2.5, 2) {};
    \node [style=none] (44) at (3.5, 2) {};
    \node [style=none] (45) at (4.5, 3.5) {};
    \node [style=none] (46) at (5.5, 3.5) {};
    \node [style=none] (47) at (5.5, 2.5) {};
    \node [style=none] (48) at (4.5, 1.5) {};

		\node [style={special_node3}] (18) at (1.5, 5) {};
		\node [style={special_node3}] (19) at (2.5, 6.5) {};
		\node [style={special_node3}] (20) at (3.5, 6.5) {};
		\node [style={special_node3}] (22) at (3.5, 2.5) {};
		\node [style={special_node3}] (23) at (2.5, 2.5) {};
		\node [style={special_node3}] (24) at (1.5, 4) {};

    \node [style={special_node}] (dp1u125) at (-0.25, 1.75) {};
    \node [style={special_node}] (dp1u126) at (0.75, 3.25) {};
    \node [style={special_node}] (dp1u127) at (1.25, 3.25) {};
    \node [style={special_node}] (dp1u128) at (2.25, 1.75) {};
    \node [style={special_node}] (dp1u129) at (2.25, -0.25) {};
    \node [style={special_node}] (dp1u130) at (-0.25, -0.25) {};

    \node [style={special_node2}] (dp1u231) at (3.75, 1.75) {};
    \node [style={special_node2}] (dp1u232) at (4.25, 1.75) {};
    \node [style={special_node2}] (dp1u233) at (4.25, -0.25) {};
    \node [style={special_node2}] (dp1u234) at (3.75, -0.25) {};

    \node [style={special_node2}] (31) at (3.75, -0.25) {};
    \node [style={special_node2}] (32) at (4.25, -0.25) {};
    \node [style={special_node2}] (33) at (4.25, 1.5) {};
    \node [style={special_node2}] (34) at (5.25, 2.75) {};
    \node [style={special_node2}] (35) at (5.25, 3.25) {};
    \node [style={special_node2}] (36) at (4.75, 3.25) {};
    \node [style={special_node2}] (37) at (3.75, 1.5) {};

    \node [style=none, label=right:{\large{$W$}}] (dp1Wlabel) at (-1.5, 3.25) {};
    \node [style=none, label=right:{\large{$R$}}] (dp1Rlabel) at (0.75, 5.3) {};
    \node [style=none, label=below:{\normalsize{$U_1$}}] (dp1U1label) at (-1, 1.45) {};
    \node [style=none, label=below:{\normalsize{$U_2$}}] (dp1U2label) at (5.2, 1.45) {};
  \end{pgfonlayer}

  \begin{pgfonlayer}{edgelayer}
    \draw [style={graph_edge}] (dp1rho) to (dp1q);
    \draw [style={graph_edge}] (dp1q) to (dp1u);
    \draw [style={graph_edge}] (dp1u) to (dp1x);
    \draw [style={graph_edge}] (dp1x) to (dp1a);
    \draw [style={graph_edge}] (dp1q) to (dp1v);
    \draw [style={graph_edge}] (dp1u) to (dp1y);
    \draw [style={graph_edge}] (dp1v) to (dp1y);
    \draw [style={graph_edge}] (dp1y) to (dp1b);
    \draw [style={graph_edge}] (dp1v) to (dp1z);
    \draw [style={graph_edge}] (dp1z) to (dp1c);
    \draw [style={graph_edge}] (dp1rho) to (dp1w);
    \draw [style={graph_edge}] (dp1w) to (dp1z);
    \draw [style={graph_edge}, bend right=40] (dp1q) to (dp1x);
  \end{pgfonlayer}

  \draw[style=set_box]
    ($(38.center)$) -- ($(39.center)$) -- ($(40.center)$) -- ($(41.center)$) -- ($(42.center)$) -- ($(43.center)$) -- ($(44.center)$) -- ($(45.center)$) -- ($(46.center)$) -- ($(47.center)$) -- ($(48.center)$) -- cycle;

  \draw[style=set_box]
    ($(18.center)$) -- ($(19.center)$) -- ($(20.center)$) -- ($(22.center)$) -- ($(23.center)$) -- ($(24.center)$) -- cycle;

  \draw[style=set_box_inner]
    ($(dp1u125.center)$) -- ($(dp1u126.center)$) -- ($(dp1u127.center)$) -- ($(dp1u128.center)$) -- ($(dp1u129.center)$) -- ($(dp1u130.center)$) -- cycle;

  \draw[style=set_box_inner]
    ($(31.center)$) -- ($(32.center)$) -- ($(33.center)$) -- ($(34.center)$) -- ($(35.center)$) -- ($(36.center)$) -- ($(37.center)$) -- cycle;

\end{tikzpicture}
         \caption{Ordered 2-partition of a DAG}
         \label{subfig:dp_alg1}
     \end{subfigure}
     \begin{subfigure}[b]{0.67\textwidth}
         \centering
         \begin{tikzpicture}[scale=0.45]
  \begin{pgfonlayer}{nodelayer}
    \node [style={graph_node}, label=right:{$a$}, fill=red!50] (dp1a) at (0, 0) {};
    \node [style={graph_node}, label=right:{$b$}, fill=red!50] (dp1b) at (2, 0) {};
    \node [style={graph_node}, label=right:{$c$}, fill=blue!50] (dp1c) at (4, 0) {};
    \node [style={graph_node}, label=right:{$x$}, fill=red!50] (dp1x) at (0, 1.5) {};
    \node [style={graph_node}, label=right:{$y$}, fill=red!50] (dp1y) at (2, 1.5) {};
    \node [style={graph_node}, label=right:{$z$}, fill=blue!50] (dp1z) at (4, 1.5) {};
    \node [style={graph_node}, label=right:{$u$}, fill=blue!50] (dp1u) at (1, 3) {};
    \node [style={graph_node}, label=right:{$v$}, fill=blue!50] (dp1v) at (3, 3) {};
    \node [style={graph_node}, label=right:{$q$}, fill=blue!50] (dp1q) at (2, 4.5) {};
    \node [style={graph_node}, label=right:{$w$}, fill=blue!50] (dp1w) at (5, 3) {};
    \node [style={graph_node}, label=right:{$\rho$}, fill=blue!50] (dp1rho) at (3, 6) {};

    \node [style=none, label=right:{\large{$V \setminus (U_1 \setminus \{u\}) $}}] (dp1Rlabel) at (-3.75, 5.3) {};
    \node [style=none, label=below:{\large{$U_1 \setminus \{u\}$}}] (dp1U1label) at (-2.2, 1.45) {};

    \node [style=none] (44) at (2.5, -0.5) {};
    \node [style=none] (45) at (2.5, 2) {};
    \node [style=none] (46) at (-0.5, 2) {};
    \node [style=none] (47) at (-0.5, -0.5) {};

    \node [style={special_node3}] (55) at (4.5, -0.5) {};
    \node [style={special_node3}] (56) at (4.5, 1.5) {};
    \node [style={special_node3}] (57) at (5.5, 2.5) {};
    \node [style={special_node3}] (58) at (5.5, 3.5) {};
    \node [style={special_node3}] (59) at (3.5, 6.5) {};
    \node [style={special_node3}] (60) at (2.5, 6.5) {};
    \node [style={special_node3}] (61) at (1.5, 5) {};
    \node [style={special_node3}] (62) at (0.5, 3.5) {};
    \node [style={special_node3}] (63) at (0.5, 2.5) {};
    \node [style={special_node3}] (64) at (2.5, 2.5) {};
    \node [style={special_node3}] (65) at (3.5, 2.5) {};
    \node [style={special_node3}] (66) at (3.5, -0.5) {};
  \end{pgfonlayer}

  \begin{pgfonlayer}{edgelayer}
    \draw [style={graph_edge}] (dp1rho) to (dp1q);
    \draw [style={graph_edge}] (dp1q) to (dp1u);
    \draw [style={graph_edge}] (dp1u) to (dp1x);
    \draw [style={graph_edge}] (dp1x) to (dp1a);
    \draw [style={graph_edge}] (dp1q) to (dp1v);
    \draw [style={graph_edge}] (dp1u) to (dp1y);
    \draw [style={graph_edge}] (dp1v) to (dp1y);
    \draw [style={graph_edge}] (dp1y) to (dp1b);
    \draw [style={graph_edge}] (dp1v) to (dp1z);
    \draw [style={graph_edge}] (dp1z) to (dp1c);
    \draw [style={graph_edge}] (dp1rho) to (dp1w);
    \draw [style={graph_edge}] (dp1w) to (dp1z);
    \draw [style={graph_edge}, bend right=40] (dp1q) to (dp1x);
  \end{pgfonlayer}

  \draw[style=set_box]
($(44.center)$) -- ($(45.center)$) -- ($(46.center)$) -- ($(47.center)$) -- cycle;
 
  \draw[style=set_box]
($(55.center)$) -- ($(56.center)$) -- ($(57.center)$) -- ($(58.center)$) -- ($(59.center)$) -- ($(60.center)$) -- ($(61.center)$) -- ($(62.center)$) -- ($(63.center)$) -- ($(64.center)$) -- ($(65.center)$) -- ($(66.center)$) -- cycle;

\end{tikzpicture}
         \quad
         \begin{tikzpicture}[scale=0.45]
  \begin{pgfonlayer}{nodelayer}
    \node [style={graph_node}, label=right:{$a$}, fill=blue!50] (dp1a) at (0, 0) {};
    \node [style={graph_node}, label=right:{$b$}, fill=blue!50] (dp1b) at (2, 0) {};
    \node [style={graph_node}, label=right:{$c$}, fill=red!50] (dp1c) at (4, 0) {};
    \node [style={graph_node}, label=right:{$x$}, fill=blue!50] (dp1x) at (0, 1.5) {};
    \node [style={graph_node}, label=right:{$y$}, fill=blue!50] (dp1y) at (2, 1.5) {};
    \node [style={graph_node}, label=right:{$z$}, fill=red!50] (dp1z) at (4, 1.5) {};
    \node [style={graph_node}, label=right:{$u$}, fill=blue!50] (dp1u) at (1, 3) {};
    \node [style={graph_node}, label=right:{$v$}, fill=blue!50] (dp1v) at (3, 3) {};
    \node [style={graph_node}, label=right:{$q$}, fill=blue!50] (dp1q) at (2, 4.5) {};
    \node [style={graph_node}, label=right:{$w$}, fill=blue!50] (dp1w) at (5, 3) {};
    \node [style={graph_node}, label=right:{$\rho$}, fill=blue!50] (dp1rho) at (3, 6) {};

    \node [style=none, label=right:{\large{$V \setminus (U_2 \setminus \{w\}) $}}] (dp1Rlabel) at (-3.75, 5.3) {};
    \node [style=none, label=below:{\large{$U_2 \setminus \{w\}$}}] (dp1U1label) at (6.5, 1.45) {};

    \node [style=none] (76) at (4.5, -0.5) {};
    \node [style=none] (77) at (4.5, 2) {};
    \node [style=none] (78) at (3.5, 2) {};
    \node [style=none] (79) at (3.5, -0.5) {};

    \node [style={special_node3}] (67) at (2.5, -0.5) {};
    \node [style={special_node3}] (68) at (-0.5, -0.5) {};
    \node [style={special_node3}] (69) at (-0.5, 2) {};
    \node [style={special_node3}] (70) at (2.5, 6.5) {};
    \node [style={special_node3}] (71) at (3.5, 6.5) {};
    \node [style={special_node3}] (72) at (5.5, 3.5) {};
    \node [style={special_node3}] (73) at (5.5, 2.5) {};
    \node [style={special_node3}] (74) at (3.5, 2.5) {};
    \node [style={special_node3}] (75) at (2.5, 1.5) {};
        
  \end{pgfonlayer}

  \begin{pgfonlayer}{edgelayer}
    \draw [style={graph_edge}] (dp1rho) to (dp1q);
    \draw [style={graph_edge}] (dp1q) to (dp1u);
    \draw [style={graph_edge}] (dp1u) to (dp1x);
    \draw [style={graph_edge}] (dp1x) to (dp1a);
    \draw [style={graph_edge}] (dp1q) to (dp1v);
    \draw [style={graph_edge}] (dp1u) to (dp1y);
    \draw [style={graph_edge}] (dp1v) to (dp1y);
    \draw [style={graph_edge}] (dp1y) to (dp1b);
    \draw [style={graph_edge}] (dp1v) to (dp1z);
    \draw [style={graph_edge}] (dp1z) to (dp1c);
    \draw [style={graph_edge}] (dp1rho) to (dp1w);
    \draw [style={graph_edge}] (dp1w) to (dp1z);
    \draw [style={graph_edge}, bend right=40] (dp1q) to (dp1x);
  \end{pgfonlayer}

  \draw[style=set_box]
($(76.center)$) -- ($(77.center)$) -- ($(78.center)$) -- ($(79.center)$) -- cycle;

  \draw[style=set_box]
($(67.center)$) -- ($(68.center)$) -- ($(69.center)$) -- ($(70.center)$) -- ($(71.center)$) -- ($(72.center)$) -- ($(73.center)$) -- ($(74.center)$) -- ($(75.center)$) -- cycle;

\end{tikzpicture}
         \caption{Two ordered 2-partition of the same DAG}
     \label{subfig:dp_alg2}
     \end{subfigure}
        \caption{Example illustrating an iteration of \cref{alg:dynamic_program_XP}. (a): The weakly connected DAG $G = (V,E)$ from \cref{subfig:ordered_partition_part} and an ordered 2-partition $\QQ = (W,R)$ of $V$. The set $W$ is split further into $U_1$ and $U_2$, with $G[U_1]$ and $G[U_2]$ forming the two connected components of $G[W]$. Both subgraphs have a single root: $u$ is the root of $G[U_1]$ and $w$ is the root of $G[U_2]$. (b): The two ordered 2-partitions $(U_1 \setminus \{u\}, V \setminus (U_1 \setminus \{u\}))$ and $(U_2 \setminus \{w\}, V \setminus (U_2 \setminus \{w\}))$ that appear when combining \cref{lem:rpsw_recursion}(b) and \cref{lem:rpsw_component_splitting}, which form the basis of \cref{alg:dynamic_program_XP}.}
        \label{fig:dp_alg}
\end{figure}

Similar to the partial scanwidth, the restricted partial scanwidth focuses only at a `window' of an extension. One only considers the vertices in a sinkset $W$ while considering the other vertices in the set $R$ to be right of $W$, disregarding the exact position these other vertices may have. 

A useful by-product of the connection to partial scanwidth is that the main recursion of the dynamic programming algorithm can be seen as a special case of \cref{lem:psw_recursion}. Another key recursion applies when the subgraph $G[W]$ is not weakly connected, allowing the problem to be split across the weakly connected components of $G[W]$. These ideas are formalized in \cref{lem:rpsw_recursion} and \cref{lem:rpsw_component_splitting}, respectively. See \cref{fig:dp_alg} for an illustration that combines the two recursive steps.

\begin{lemma}
\label{lem:rpsw_recursion}
Let $G=(V,E)$ be a weakly connected DAG and $\QQ = (W,R)$ an ordered 2-partition of $V$ with $W \neq \emptyset$.
\begin{enumerate}[label={(\alph*)},noitemsep,topsep=3pt]
\item If $|W| = 1$ (with $W = \{w \}$), then $\Qdash\rpsw (G) = \delta^{\mathrm{in}} (w).$

\item If $|W| \geq 2$ and $G[W]$ is weakly connected, then
$$\Qdash\rpsw (G) = \min_{\rho \in P (G[W])}  \max \Big\{ (W\setminus \{\rho \} ,R\cup \{\rho \} ) \ssdash \rpsw (G), \delta^{\mathrm{in}} (W)  \Big\},$$
where $P(G[W])$ is the set of roots of $G[W]$.
\end{enumerate}
\end{lemma}
\begin{proof}
\emph{(a):} First note that $R = V \setminus \{w\}$ in this case. Then from \cref{lem:psw_recursion}a and \cref{def:rpsw} we obtain $\Qdash \rpsw(G) = |\{uv \in E: u \in V \setminus \{w\}, v\connect{G[\emptyset \cup \{w\}]} w\}| = \delta^{\mathrm{in}} (w)$.

(b): Recall that in \cref{lem:psw_recursion}b, we defined the collection of sets $\mathcal{W}_k = \{W' \subseteq W: |W'| =k \text{ and no arc in } E \\ \text{ is directed from } W' \text{ to } W \setminus W'\}$ for ordered 3-partitions $(L, W, R)$ and some $k \in \{ 1, \ldots, |W|-1 \}$. We now use this lemma with $k = |W|-1$, to get
\begin{align*}
\Qdash\rpsw (G) 
&=\min_{W' \in \mathcal{W}_{|W|-1}}  \max \Big\{ (\emptyset,W',R\cup (W\setminus W'))\ssdash\psw(G), (\emptyset \cup W' ,W \setminus W',R) \ssdash\psw (G)   \Big\} \\
&= \min_{W'= W \setminus \{\rho \}: \rho \in P(G[W])}  \max \Big\{ (W',R\cup \{\rho \}) \ssdash\rpsw (G) , ( W' ,\{ \rho \},R) \ssdash \psw (G)  \Big\} \\
&= \min_{ \rho \in P(G[W])}  \max \Big\{ (W \setminus \{\rho \},R\cup \{\rho \})  \ssdash \rpsw (G), ( W \setminus \{\rho \} ,\{ \rho \},R) 
\ssdash \psw (G) \Big\}.
\end{align*}
The crucial observation in the above equality is that $\mathcal{W}_{|W| - 1}$ contains precisely the subsets of $W$ obtained by removing one vertex that is a root of $G[W]$.

Now using \cref{lem:psw_recursion}a, we immediately have that
$( W \setminus \{\rho \} ,\{ \rho \},R)  \ssdash \psw (G)  = |\{uv \in E(G): u \in R, v\connect{G[W]} \rho\}|$. This last expression equals $\deltain(W)$, since $G[W]$ is weakly connected.
\end{proof}

\begin{lemma}
\label{lem:rpsw_component_splitting}
Let $G=(V,E)$ be a weakly connected DAG and $\QQ = (W,R)$ an ordered 2-partition of $V$ with $W \neq \emptyset$. Then
$$\Qdash\rpsw (G) =  \max_{U_i \vartriangleleft W}  \Big\{ (U_i, V \setminus U_i )\ssdash\rpsw (G)  \Big\},$$ 
where $U_i \vartriangleleft W$ indicates that $G[U_i]$ is a weakly connected component of $G[W]$.
\end{lemma}
\begin{proof}
Because $W$ must be a sinkset, each $U_i \vartriangleleft W$ must also be a sinkset. Therefore, $(U_i, V \setminus U_i)$ is indeed an ordered 2-partition of $V$, for each $U_i \vartriangleleft W$.

First, we prove a critical equality. Let $r$ be the number of weakly connected components of $G[W]$. For each $i \in \{1, \ldots , r  \}$, let $\sigma_i \in \Pi[U_i]$ be arbitrary, and define $\sigma_1 \circ \dots \circ \sigma_r = \sigma \in \Pi[W]$. This is indeed an extension because the different $\sigma_i$ are not weakly connected in $G[W]$. Using \cref{def:rpsw} and denoting $|\Qdash \RPSW_w^\sigma (G)|$ by $\rpsw_w^\sigma (\QQ)$ for ease of notation (and similar for $\PSW$), we have:
\begin{align}
\Qdash\rpsw(\sigma, G) = \max_{w \in W} \rpsw_w^\sigma (\QQ) 
&= \max_{U_i \vartriangleleft W}  \left\{ \max_{w \in U_i} \psw_w^\sigma ((\emptyset, W,R)) \right\} \nonumber \\
&= \max_{U_i \vartriangleleft W}  \left\{ \max_{w \in U_i} \psw_w^{\sigma_i} \left(\left(\bigcup\nolimits_{1\leq j < i} U_j, U_i,R \cup \bigcup\nolimits_{j>i} U_j \right)\right) \right\} \nonumber  \\
&= \max_{U_i \vartriangleleft W}  \left\{ \max_{w \in U_i} \psw_w^{\sigma_i} \left(\left(\emptyset, U_i,R \cup \bigcup\nolimits_{j\neq i} U_j \right)\right) \right\} \nonumber  \\
&= \max_{U_i \vartriangleleft W}  \Big\{ (U_i, V \setminus U_i )\ssdash\rpsw (\sigma_i, G)  \Big\}.
\label{eq:component}
\end{align}
The third equality is similar to the third equality of \cref{eq:psw} that appears in the proof of \cref{lem:psw_recursion}b. It uses the observation that if we only maximize over vertices that form a consecutive subsequence of the extension $\sigma$, we can just as well put the vertices to the left of this subsequence in the set $L$, and the ones to the right in the set $R$. For the fourth equality we then use that in $G[W]$, the vertices of $\bigcup_{1\leq j < i} U_j$ are not weakly connected to those in $U_i$. Thus, we can put them in the $R$-set, without changing the restricted partial scanwidth. We are now ready to prove the lemma.

($\leq$) For all $U_i \vartriangleleft W$, we let $\sigma_i \in \Pi[U_i]$ be an optimal extension. In other words, $(U_i,V \setminus U_i )\ssdash\rpsw (G) = (U_i,V \setminus U_i )\ssdash\rpsw (\sigma_i, G)$. Now let $\sigma = \sigma_1 \circ \dots \circ \sigma_r \in \Pi[W]$. Using \cref{eq:component}, we get that
$\Qdash \rpsw (G) \leq \max_{U_i \vartriangleleft W} \left\{  (U_i,V \setminus U_i) \ssdash\rpsw (G) \right\}$.

($\geq$) We first present a claim. As the claim is quite intuitive, we refer to \cite{holtgrefe2023scanwidth} for the rather technical proof.

\emph{Claim:} Let $\sigma \in \Pi[W]$ be such that for some $k \in \{1, \ldots, |W|-1 \}$, $\sigma(k)$ and $\sigma(k+1)$ are not weakly connected in $G[W]$. Let $\pi$ be obtained from $\sigma$ by swapping $\sigma(k)$ and $\sigma (k+1)$. Then, $\pi \in \Pi[W]$ and $\Qdash \rpsw(\sigma, G) = \Qdash \rpsw(\pi, G)$.

Let $\sigma \in \Pi[W]$ be such that $\Pdash\rpsw (G) = \Pdash\rpsw(\sigma, G)$, which exists by definition. We can also assume that $\sigma = \sigma_1 \circ \dots \circ \sigma_r$, where for each $i$ we have $\sigma_i \in \Pi[U_i]$. Such an extension exists since we can keep swapping consecutive vertices from different $U_i$ until this condition holds, if $\sigma$ does not have this property. By the claim, this is also an extension, and it will give the same restricted partial scanwidth. From \cref{eq:component} it then quickly follows that
$\Qdash \rpsw (G) \geq \max_{U_i \vartriangleleft W} \left\{  (U_i,V \setminus U_i) \ssdash\rpsw (G) \right\}$.
\end{proof}

The previous two lemmas now result in the following corollary, which forms the core of our dynamic programming algorithm. In particular, part (a) is an immediate consequence of \cref{lem:rpsw_component_splitting}, while parts (b) and (c) follow from \cref{lem:rpsw_recursion}.

\begin{corollary}
\label{cor:rpsw_larger_k}
Let $G=(V,E)$ be a weakly connected DAG, $k\geq 1$ an integer and $\QQ = (W,R)$ an ordered 2-partition of $V$ such that $W \neq \emptyset$. Then, 
\begin{enumerate}[label={(\alph*)},noitemsep,topsep=3pt]
\item $\Qdash \rpsw (G) >k$ if and only if $(U_i, V \setminus U_i )\ssdash\rpsw( G) >k$ for some component $G[U_i]$ of $G[W]$.

\item For $|W|=1$, $\Qdash\rpsw (G) > k$ if and only if $\deltain (W) > k$.

\item For $|W|\geq 2$ such that $G[W]$ is weakly connected, $\Qdash\rpsw (G) > k$ if and only if $\deltain (W) > k$ or $(W\setminus \{\rho \} ,R\cup \{\rho \} ) \ssdash\rpsw (G) >k$ for all roots $\rho$ of $G[W]$.
\end{enumerate}
\end{corollary}

With this corollary in mind, we can set up a dynamic programming procedure, described in \cref{alg:dynamic_program_XP}. The algorithm involves the use of a `table', denoted by $T$, to store previously calculated results. At the cost of more space, the recursive procedure can first check if a result is already known, thereby saving time. See \cref{fig:dp_alg} for an example where $G[W]$ consists of two weakly connected components, each with a single root, so the algorithm recurses on two ordered 2-partitions (assuming $k$ is not too small, which would prevent the recursive calls).

\begin{algorithm}[htb]
\caption{Dynamic programming algorithm to solve $k$-\textsc{Scanwidth}.}
\label{alg:dynamic_program_XP}
\Input{Weakly connected DAG $G=(V,E)$, integer $k\geq 1$. }
\Output{If $\sw(G) \leq k$: scanwidth $\sw$ of $G$ and an optimal extension $\sigma_{\mathrm{opt}}$. If $\sw(G) > k$: $\infty$ and an incomplete extension.}
\SetKwFunction{dresswcs}{R-PartialScanwidth}
$T \gets$ empty table to tabulate results, indexed by all 2-partitions of $V$ \\
$\sw, \sigma_{\mathrm{opt}} \gets \dresswcs(  V, \emptyset, k)$ \label{algline:initial} \\
\Return{$\sw, \sigma_{\mathrm{opt}}$}\\
\setcounter{AlgoLine}{0}
\nonl \Proc{$\dresswcs(W, R, k)$}{
\If(\tcp*[f]{Look up result in global table, if available}){$T(W, R) $ exists}{\Return $T(W,R)$}
\KwInit{$\rpsw \gets \infty$ ; $\sigma \gets ()$}\\
\For{weakly connected component $G[U_i]$ of $G[W]$, $(i = 1, \dots, r)$ \label{algline:for}}{
\KwInit{$\rpsw_i \gets \infty$ ; $\sigma_i \gets ()$}\\
\If{$|U_i|=1$ with $U_i = \{v\} \mathrm{\normalfont{\textbf{ and }}}  \delta^{\mathrm{in}} (U_i) \leq k$}{
$\rpsw_i \gets \delta^{\mathrm{in}} (v)$\\
$\sigma_i \gets (v)$
}
\ElseIf{$|U_i|>1 \mathrm{\normalfont{\textbf{ and }}}  \delta^{\mathrm{in}} (U_i) \leq k$}{
\For{root $\rho$ of $G[U_i]$}
{
$\rpsw_1', \sigma_1' \gets \dresswcs (U_i \setminus \{\rho \},V \setminus (U_i \setminus \{ \rho \} ), k) $\\
$\rpsw' \gets \max \{ \rpsw_1', \delta^{\mathrm{in}} (U_i) \}$\\
\If{$\rpsw' < \rpsw_i$}{$\rpsw_i \gets \rpsw'$\\
$\sigma_i \gets \sigma_1' \circ (\rho)$
}
}
}
}
$\rpsw \gets \max \{ \rpsw_i: i = 1, \dots, r \}$ \label{algline:max}\\
$\sigma \gets \sigma_1 \circ \dots \circ \sigma_r$ \\
$T(W,R) \gets \rpsw, \sigma$ \tcp*{Store result in global table}
\Return{$\rpsw, \sigma$}}
\end{algorithm}

Before \cref{thm:dynamic_programXP} formally states correctness of the algorithm, we need two lemmas that help to further bound the number of considered sets in the algorithm and consequently its time complexity. 
  
An \emph{antichain} is a subset of a (partially) ordered set, in which each pair of elements is incomparable to each other. In our context, the roots of a sinkset form an antichain when considering the natural order of a DAG. This is proved in the next lemma. Recall that we write $W \sqsubseteq U$ for any $U \subseteq V$ if both $W \subseteq U$ and $W$ is a sinkset.

\begin{lemma}
\label{lem:sinksets-antichains}
Let $G=(V,E)$ be a weakly connected DAG. Then for all $W \sqsubseteq V$, the roots of $G[W]$ form an antichain with respect to the partial order $<_G$. Moreover, there is a one-to-one correspondence between the sets $W\sqsubseteq V$ and the antichains of the partial order $<_G$, defined by the roots of $G[W]$.
\end{lemma}
\begin{proof}
We will start with the first statement. Assume towards a contradiction that the roots of $G[W]$ are not an antichain. If we let $P(G[W])$ be the set of roots of $G[W]$, we must then have that for some $\rho_1, \rho_2 \in P(G[W])$ it holds that $\rho_1 <_G \rho_2$. Now let $v \in V$ be such that $\rho_1 <_G v \leq_G \rho_2$ and $(v, \rho_1) \in E$. Such a vertex exists, as otherwise $\rho_1$ and $\rho_2$ would not be comparable. But as $v \leq_G \rho_2$ and $W$ is a sinkset, $v$ must be in $W$. Thus, $\rho_1$ can not be a root of $G[W]$: a contradiction. This proves the first statement.

Let us denote the set of all $W \sqsubseteq V$ by $\mathcal{V}$ and the set of all vertex-antichains with respect to $<_G$ by $\mathcal{A}$. We can now define a function $f$ $: \mathcal{V} \rightarrow \mathcal{A}$ by $f(W) = P(G[W])$ for all $W \in \mathcal{V}$. By the first statement, $f$ indeed maps all sets $W$ into $\mathcal{A}$. It can easily be shown that $f$ is both surjective and injective (see \cite{holtgrefe2023scanwidth} for the full proof), thus proving that $f$ is a bijection from the sinksets to the vertex-antichains of $G$, as desired.
\end{proof}

With the characterization of the sinksets of a DAG by means of their roots, we can bound the number of sinksets that have bounded indegree. 
\begin{lemma}
\label{lem:nr_sinksets_bound}
Let $G=(V,E)$ be a weakly connected DAG of $n$ vertices and $r$ roots, and let $k\geq 1$ be an integer. Then, the number of sets $W \sqsubseteq V$ such that $\deltain (W) \leq k$, is bounded from above by $n^{k+r-1}$.
\end{lemma}
\begin{proof}
Let $k \geq 1$ be arbitrary and $\mathcal{V}_k = \{W \sqsubseteq V: \deltain(W) \leq k \}$. We first prove the claim that for all $W \in \mathcal{V}_k$, $G[W]$ has at most $k+r-1$ roots.

\begin{claimproof}
Let $W  \in \mathcal{V}_k$ be arbitrary. We now consider two cases.

\emph{Case 1: each root of $G$ is a root of $W$.} As $W$ is a sinkset, this means that $W=V$, and so $G[W] = G$. Therefore, $G[W]$ has exactly $r$ roots. Because $k\geq 1$, the bound then follows.

\emph{Case 2: there exists a root of $G$ that is not a root of $G[W]$.} Then, at most $r-1$ of the roots of $G[W]$, are also a root of $G$. Therefore, $G[W]$ has at most $r-1$ roots with indegree 0 in $G$. Furthermore, $G[W]$ has at most $k$ roots with an indegree of at least 1 in $G$ (otherwise, the indegree of $W$ would be larger than $k$, and then $W \notin \mathcal{V}_k$). Together, this gives that $G[W]$ has at most $k + r -1 $ roots.
\end{claimproof}
Together with \cref{lem:sinksets-antichains}, this claim shows that for all $W \in \mathcal{V}_k$, the roots of $G[W]$ form a vertex-antichain of $G$, with a size of at most $k+r-1$. Now let $\mathcal{A}_\ell$ denote the set of vertex-antichains of $G$ of size at most $\ell$. Then we can define a function $h: \mathcal{V}_k \rightarrow \mathcal{A}_{k+r-1}$ by $h(W) = P(G[W])$ for all $W \in \mathcal{V}_k$ (here, $P(G[W])$ indicates the set of roots of $G[W]$ again). But then, $h$ is actually a restriction of the function $f$ from the proof of \cref{lem:sinksets-antichains}, to the domain $\mathcal{V}_k$ (and with a smaller co-domain). Using that this function $f$ was bijective, we must have that $h: \mathcal{V}_k \rightarrow \mathcal{A}_{k+r-1}$ is injective, otherwise this would contradict the injectivity of $f$. Therefore, we naturally have that $|\mathcal{V}_k| \leq |\mathcal{A}_{k+r-1}|$. So, it suffices to count the number of vertex-antichains in $G$ of size at most $k+r-1$, to get an upper bound for our lemma. We will now show by induction on $\ell$ that $|\mathcal{A}_\ell| \leq n^\ell$ for all integers $\ell \geq 1$, which will then prove the lemma.

The base case where $\ell = 1$ is trivial. Now assume that the induction hypothesis holds for $\ell$, so $|\mathcal{A}_\ell| \leq n^\ell$. Note that all vertex-antichains of size exactly $\ell+1$ can be created from the antichains of size $\ell$, by the addition of one incomparable vertex. As there are only $n - \ell$ vertices left to be added for each antichain of size $\ell$, we obtain that $|\mathcal{A}_{\ell+1}| \leq  n^\ell + n^\ell \cdot (n - \ell ) \leq n^{\ell+1}$.
\end{proof}

The previous lemma allows us to bound the number of sets that are considered during the execution of the algorithm. This sets up the stage for proving the time complexity of \cref{alg:dynamic_program_XP} in the following theorem. 

\begin{theorem}
\label{thm:dynamic_programXP}
Let $G=(V,E)$ be a weakly connected DAG of $n$ vertices, $m$ arcs and $r$ roots, and let $k\geq 1$ be an integer. Then, \cref{alg:dynamic_program_XP} solves $k$-\textsc{Scanwidth} in $O ( (k+r-1) \cdot n^{k+r-1}\cdot m )$ time and $O ( (k+r-1) \cdot n^{k+r})$ space.
\end{theorem}
\begin{proof}
\emph{Correctness:} In the algorithm we use $\infty$ as a placeholder value whenever the restricted partial scanwidth is larger than $k$, and in that case, we do not care about the corresponding extension. The correctness of the subroutine \texttt{R-PartialScanwidth} now follows immediately from \cref{cor:rpsw_larger_k}, which precisely defines when $\infty$ should be returned. The correct extensions are returned due to the constructive nature of \cref{lem:rpsw_recursion,lem:rpsw_component_splitting}, on which \cref{cor:rpsw_larger_k} is based.

\emph{Time complexity:} Within each subroutine call, it takes $O(m)$ time to create the components $G[U_i]$ and to compute all indegrees $\deltain (U_i)$. Given these indegrees, only $O(|U_i|)$ extra time is spent in the outer for-loop (without the recursive call). But since the $U_i$ are disjoint, the complete time spend per subroutine call must then be dominated by $O(m)$.

Due to the tabulation, we execute the subroutine at most once for each of the (at most $2^n$) sinksets of $G$. The subroutine is only executed for sinksets $U_i$ that are created by the deletion of a root of a weakly connected sinkset $W$ with indegree at most $k$. According to \cref{lem:nr_sinksets_bound}, there are at most $n^{k+r-1}$ such sinksets $W$. As stated in the claim of the proof of \cref{lem:nr_sinksets_bound} these sinksets $W$ have at most $k+r-1$ roots, yielding $O((k+r-1) \cdot n^{k+r-1})$ recursive subroutine calls. This results in a total time complexity of $O((k+r-1) \cdot n^{k+r-1} \cdot m)$.

\emph{Space complexity:}
We store an extension of size $O(n)$ and a value $\rpsw$ for each considered sinkset. This requires $O ((k+r-1) \cdot n^{k+r})$ space. The space needed for the graph, and to run the subroutine, is also surely bounded by this function.
\end{proof}

From this theorem, we can deduce a nice complexity result for DAGs with a fixed number of roots. Our algorithm then functions as an XP algorithm when considering scanwidth as the parameter.

\begin{corollary}
\label{cor:sw_is_xp}
Let $G=(V,E)$ be a weakly connected rooted DAG with $n$ vertices, $m$ arcs, $r$ roots, and a scanwidth of $k$. Then, there exists an algorithm that solves \textsc{Scanwidth} in $O((k+r-1) \cdot m \cdot n^{k+r-1})$ time and $O((k+r-1) \cdot n^{k + r})$ space. Thus, for DAGs with a fixed number of roots \textsc{Scanwidth} is in $\mathrm{XP}$ when considering the scanwidth as the parameter.
\end{corollary}
\begin{proof}
By repeatedly running \cref{alg:dynamic_program_XP} we can solve $i$-\textsc{Scanwidth} for an increasing value of $i$. When we eventually reach the value $k$ (which equals the scanwidth), we will find an optimal extension. If we keep the intermediate results of the previous algorithm runs in the table $T$, we consider the same amount of sinksets as we would have considered by directly solving $k$-\textsc{Scanwidth}. Thus, the time and space complexity is the same as in \cref{thm:dynamic_programXP}. Since we consider the number of roots $r$ to be fixed, this directly proves the XP result stated in the corollary.
\end{proof}

In \cref{subsec:total_reduction} we introduced a decomposition algorithm aimed at reducing the size of an instance. We proved that, in the case of networks, these reduced instances have their size bounded by a linear function of the level. If we apply this to the algorithm described in the previous corollary, we can formulate another complexity result, proving that for networks \textsc{Scanwidth} is FPT when considering the level as a parameter.
\begin{corollary}
\label{cor:level_FPT}
Let $G=(V,E)$ be a level-$\ell$ network of $n$ vertices and $m$ arcs. Then, there exists an algorithm that solves \textsc{Scanwidth} in $O (2^{4 \ell-1} \cdot \ell \cdot n + n^2)$ time. Thus, for networks \textsc{Scanwidth} is in $\mathrm{FPT}$ when considering the level as the parameter.
\end{corollary}
\begin{proof}
The algorithm in the previous corollary visits each sinkset at most once and spents at most $O(m)$ time on each such set. Thus it also has its time complexity bounded by $O(2^n \cdot m)$. When combined with decomposition algorithm~\ref{alg:reduction}, we can then solve \textsc{Scanwidth} in $O(n^2 + n \cdot m' \cdot 2^{n'})$ time, according to \cref{lem:alg_reduction}. Here, $n'$ (resp. $m'$) is the maximum number of vertices (resp. arcs) of any of the subproblems created by the decomposition algorithm. $G$ is assumed to be a network, and thus the graphs of the different subproblems created by the decomposition algorithm have at most $4\ell -1$ vertices, and at most $5\ell-2$ arcs, according to \cref{lem:reduction_level_bound}. Substituting these numbers for $n'$ and $m'$ gives the desired result.
\end{proof}

\section{Heuristics}\label{sec:heuristics}
Motivated by the successful use of sub-optimal solution methods for other width parameters \cite{diaz2002survey}, we now divert our attention to heuristics. We observe that the sets $\SW_i^\sigma$ of an extension $\sigma$ of a DAG correspond to arc-cuts of the DAG. In the next two subsections, we leverage this observation to develop a heuristic. Instinctively, it makes sense to search for a small arc-cut in the graph and then use that cut to split into two subproblems. \cref{subsec:dag-cut} takes a closer look at these \emph{DAG-cuts} that appear in a (tree) extension and thereafter develops the corresponding heuristic idea. We reserve \cref{sec:greedy_heuristic} for a brief discussion on how to apply two other already existing algorithmic frameworks to scanwidth.

\subsection{Repeated DAG-cut-splitting heuristic}
\label{subsec:dag-cut}
First recall some standard terminology for cuts. A \emph{(directed) cut} in a directed graph $G=(V,E)$ is a partition $C=(S, T)$ of $V$ (with $|S|, |T| > 0$). The corresponding \emph{cut-set} is the set $\{uv \in E: u \in S, v \in T\}$. A directed cut $C$ is \emph{minimal} if no other cut exists that has a cut-set that is contained in the cut-set of $C$. For two distinct vertices $s, t \in V$, an \emph{$s$-$t$ cut} is a directed cut $(S, T)$ such that $s \in S$ and $t \in T$. The size (resp. weight) of the cut refers to the size (resp. sum of weights) of the cut-set and we denote it by $|C|$ (resp. $w(C)$, where $w$ is the weight function of the graph).

We can now introduce a certain type of cut, for which we will show shortly that they are the exact cuts that appear in extensions. To illustrate the definition, see \cref{subfig:DAG-cuts}.
\begin{definition}[DAG-cut]\label{def:dagcut}
Let $G=(V, E)$ be a weakly connected DAG, then we call a directed cut $C=(S, T)$ a \emph{DAG-cut} if $C$ is minimal and $T$ is a sinkset. If $|S| = 1$ or $|T| = 1$, $C$ is a \emph{trivial} DAG-cut.
\end{definition}

\begin{figure}[htb]
     \begin{subfigure}[b]{0.45\textwidth}
         \centering
	\begin{tikzpicture}[scale=0.40]
	\begin{pgfonlayer}{nodelayer}
		\node [style={graph_node}] (0) at (3, 0) {};
		\node [style={graph_node}] (1) at (0, 3) {};
		\node [style={graph_node}] (2) at (1.25, 3) {};
		\node [style={graph_node}] (3) at (2.5, 3) {};
		\node [style={graph_node}] (4) at (3, 6) {};
		\node [style={graph_node}] (5) at (5, 3) {};
		\node [style={graph_node}] (6) at (5, 9) {};
		\node [style={graph_node}] (7) at (5.5, 6) {};
		\node [style={graph_node}] (8) at (6.75, 6) {};
		\node [style={graph_node}] (9) at (8, 6) {};
		\node [style={graph_node}] (10) at (7.75, 0.5) {};
		\node [style=none] (11) at (6, 0.25) {};
		\node [style=none, label=below:{$C_3$}] (12) at (8.75, 1.5) {};
		\node [style=none, label=above:{$C_5$}] (13) at (8.75, 3) {};
		\node [style=none] (14) at (1, 0.5) {};
		\node [style=none, label=above:{$C_1$}] (15) at (4, 8.5) {};
		\node [style=none] (16) at (8.75, 5.25) {};
		\node [style=none, label=left:{$C_4$}] (17) at (2.75, 7.5) {};
		\node [style=none] (18) at (5, 1.75) {};
		\node [style=none] (19) at (5.25, 1) {};
		\node [style=none, label=above:{$C_2$}] (20) at (1, 5.25) {};
		\node [style=none] (21) at (4.75, 1.25) {};
	\end{pgfonlayer}
			\draw [style={cut_line}, bend left] (11.center) to (12.center);
		\draw [style={cut_line}, in=-165, out=-90, looseness=1.50] (15.center) to (16.center);
		\draw [style={cut_line}, in=165, out=0, looseness=0.75] (17.center) to (18.center);
		\draw [style={cut_line}, in=0, out=-165, looseness=0.75] (13.center) to (19.center);
		\draw [style={cut_line}, in=45, out=-180, looseness=0.50] (19.center) to (14.center);
		\draw [style={cut_line}, in=135, out=-15, looseness=1.50] (20.center) to (21.center);
	\begin{pgfonlayer}{edgelayer}
		\draw [style={graph_edge}] (6) to (4);
		\draw [style={graph_edge}] (4) to (1);
		\draw [style={graph_edge}] (4) to (2);
		\draw [style={graph_edge}] (4) to (3);
		\draw [style={graph_edge}] (6) to (7);
		\draw [style={graph_edge}] (6) to (8);
		\draw [style={graph_edge}] (6) to (9);
		\draw [style={graph_edge}] (3) to (0);
		\draw [style={graph_edge}] (2) to (0);
		\draw [style={graph_edge}] (1) to (0);
		\draw [style={graph_edge}] (7) to (5);
		\draw [style={graph_edge}] (8) to (5);
		\draw [style={graph_edge}] (9) to (5);
		\draw [style={graph_edge}] (4) to (5);
		\draw [style={graph_edge}] (5) to (0);
		\draw [style={graph_edge}] (5) to (10);
		\draw [style={graph_edge}] (9) to (10);
	\end{pgfonlayer}
\end{tikzpicture}
         \caption{Weakly connected DAG $G$}
         \label{subfig:DAG-cuts}
     \end{subfigure}
\hfill
     \begin{subfigure}[b]{0.45\textwidth}
         \centering
	\begin{tikzpicture}[scale=0.40,
semi_fill/.style 2 args={fill=#2, path picture={
\fill[#1] (path picture bounding box.south east) -|
                         (path picture bounding box.north west) -- cycle;}}
]

\tikzstyle{infinity_weight}=[-{Latex[scale=1.2]}, draw=gray!70, dashed]
	\begin{pgfonlayer}{nodelayer}
		\node [style={graph_node}] (0) at (3, 0) {};
		\node [style={graph_node}, fill=black] (1) at (0, 3) {};
		\node [style={graph_node}, fill=black] (2) at (1.25, 3) {};
		\node [style={graph_node}, fill=black, label=right:{$t$}] (3) at (2.5, 3) {};
		\node [style={graph_node}, fill=gray!50, label=right:{$s$}] (4) at (3, 6) {};
		\node [style={graph_node}, fill=black] (5) at (5, 3) {};
		\node [style={graph_node}] (6) at (5, 9) {};
		\node [style={graph_node}, fill=gray!50] (7) at (5.5, 6) {};
		\node [style={graph_node}, fill=gray!50] (8) at (6.75, 6) {};
		\node [style={graph_node}, fill=gray!50] (9) at (8, 6) {};
		\node [style={graph_node}, semi_fill={black}{gray!50}] (200) at (8, 6) {};		
		\node [style={graph_node}] (10) at (7.75, 0.5) {};
		\node [style=none] (11) at (6, 0.25) {};
		\node [style=none] (12) at (8.75, 1.5) {};
		\node [style=none] (13) at (8.75, 3) {};
		\node [style=none] (14) at (1, 0.5) {};
		\node [style=none] (15) at (4, 8.5) {};
		\node [style=none] (16) at (8.75, 5.25) {};
		\node [style=none] (17) at (2.75, 7.5) {};
		\node [style=none] (18) at (5, 1.75) {};
		\node [style=none] (19) at (5.25, 1) {};
		\node [style=none, label=above:{$C_2$}] (20) at (1, 5.25) {};
		\node [style=none] (21) at (4.75, 1.25) {};
	\end{pgfonlayer}
\draw [style={cut_line}, in=135, out=-15, looseness=1.50] (20.center) to (21.center);
	\begin{pgfonlayer}{edgelayer}
		\draw [style={graph_edge}] (6) to (4);
		\draw [style={graph_edge}] (4) to (1);
		\draw [style={graph_edge}] (4) to (2);
		\draw [style={graph_edge}] (4) to (3);
		\draw [style={graph_edge}] (6) to (7);
		\draw [style={graph_edge}] (6) to (8);
		\draw [style={graph_edge}] (6) to (9);
		\draw [style={graph_edge}] (3) to (0);
		\draw [style={graph_edge}] (2) to (0);
		\draw [style={graph_edge}] (1) to (0);
		\draw [style={graph_edge}] (7) to (5);
		\draw [style={graph_edge}] (8) to (5);
		\draw [style={graph_edge}] (9) to (5);
		\draw [style={graph_edge}] (4) to (5);
		\draw [style={graph_edge}] (5) to (0);
		\draw [style={graph_edge}] (5) to (10);
		\draw [style={graph_edge}] (9) to (10);
		\draw [style={infinity_weight}, bend left=15] (4) to (6);
		\draw [style={infinity_weight}, bend left=15] (7) to (6);
		\draw [style={infinity_weight}, bend left=15] (8) to (6);
		\draw [style={infinity_weight}, bend right=15, looseness=0.75] (9) to (6);
		\draw [style={infinity_weight}, bend left=15] (5) to (7);
		\draw [style={infinity_weight}, bend left=15] (5) to (8);
		\draw [style={infinity_weight}, bend right=15] (5) to (9);
		\draw [style={infinity_weight}, bend right=15, looseness=0.75] (10) to (9);
		\draw [style={infinity_weight}, bend right=15] (10) to (5);
		\draw [style={infinity_weight}, bend right=15] (5) to (4);
		\draw [style={infinity_weight}, bend right, looseness=0.75] (3) to (4);
		\draw [style={infinity_weight}, bend right=15, looseness=0.75] (2) to (4);
		\draw [style={infinity_weight}, bend left=15, looseness=1.25] (1) to (4);
		\draw [style={infinity_weight}, bend left=15, looseness=1.25] (0) to (1);
		\draw [style={infinity_weight}, bend right=15] (0) to (2);
		\draw [style={infinity_weight}, bend right=15] (0) to (3);
		\draw [style={infinity_weight}, bend right=15] (0) to (5);
	\end{pgfonlayer}
\end{tikzpicture}
         \caption{Weighted auxiliary graph $H$}
         \label{subfig:auxil_DAG_cut}
     \end{subfigure}
     \caption{(a): Weakly connected DAG $G$ with unit weights on all arcs. $C_1$ is a non-trivial DAG-cut of weight 5; $C_2$ is a smallest non-trivial DAG-cut of weight 4; $C_3$ is a trivial DAG-cut; $C_4$ is a directed cut that is not a DAG-cut, because it does not cut off a sinkset; $C_5$ is a directed cut that is not a DAG-cut, because the cut is not minimal (in particular, $C_3$ is a cut with a cut-set contained in the cut-set of $C_5$). (b): The corresponding weighted auxiliary graph $H$, where all dashed arcs have weight $\infty$ and the black arcs have unit weights. The grey vertices are in the set $U$ and the black vertices are in the set $W$. The smallest DAG-cut $C_2$ of $G$ corresponds to a minimum directed $s$-$t$ cut in $H$ with the same weight, where $s \in U$ and $t \in W$.}
\end{figure}
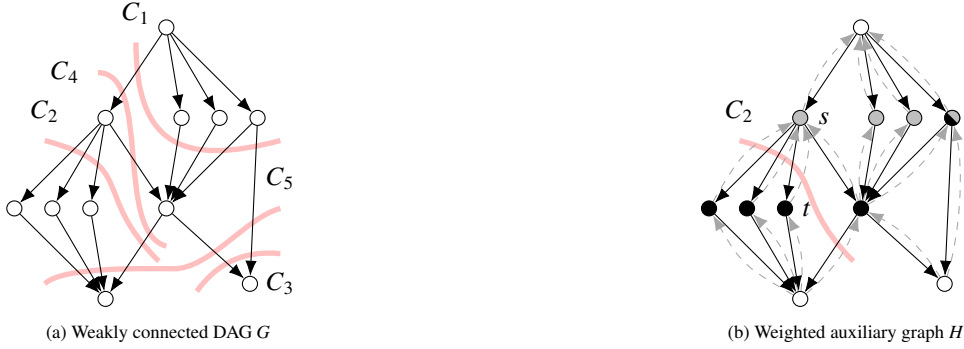

It might not immediately be clear that the sets $\SW^\sigma_i$ of an extension $\sigma$ correspond to the cut-sets of DAG-cuts, and vice versa. The next lemma formally proves this, thus solidifying the idea of using DAG-cuts to split a graph.
\begin{lemma}
Let $G=(V,E)$ be a weakly connected DAG with $n \geq 2$ vertices. Then, a set $F\subseteq E$ is the cut-set of some DAG-cut $C$ if and only if $F=\SW_i^\sigma$ for some position $i < n$ of an extension $\sigma$ of $G$.
\end{lemma}
\begin{proof}
($\Rightarrow$) Let $F$ be the cut-set of some DAG-cut $C=(S,T)$. By \cref{def:dagcut}, $C$ is now a minimal cut. This means that no other cut exists with a cut-set contained in $F$. This implies that $G[T]$ is a weakly connected graph. Now let $\sigma_1 \in \Pi[T]$ and $\sigma_2 \in \Pi[S]$. Since $T$ is a sinkset, $\sigma = \sigma_1 \circ \sigma_2$ is then an extension of $G$. Using that $G[T]$ was weakly connected, we have that $\SW_{|T|}^\sigma = F$, where $|T| < n$ because $|S| > 0$.

($\Leftarrow$) Let $\sigma$ be an extension of $G$ and $i < n$ a position of $\sigma$. Now let $T$ be the vertex set of the weakly connected component of $G[1 \ldots i]$ that contains $\sigma (i)$. We let $S = V \setminus T$. It should come as no surprise that $\SW_i^\sigma = \{uv \in E: u \in S, v \in T \}$. In other words, $\SW_i^\sigma$ is the cut-set of $C = (S, T)$, with $|S| > 0$ because $|T| \leq i < n$. Since $T$ is weakly connected, $C$ must be minimal. As $T$ was a component of $G[1 \ldots i]$ and $\sigma$ was an extension, it must also be a sinkset. Thus, $C$ is a DAG-cut.
\end{proof}

Our main focus will be to find non-trivial DAG-cuts, specifically the smallest such cut(s). Finding these minimum-weight non-trivial cuts is not obvious, but the following lemma, illustrated in \cref{subfig:auxil_DAG_cut}, will clarify the approach. In particular, we prove that finding non-trivial DAG-cuts of finite weight is equivalent to finding certain $s$-$t$ cuts in an auxiliary graph $H$. This is helpful, as several algorithms are known that find minimum $s$-$t$ cuts. The idea to use reverse arcs of infinite weight in the auxiliary graph is inspired by  Ravi, Agrawal, and Klein \cite{ravi1991ordering}, who employed this technique for the closely related \emph{DAG edge-separators}.

\begin{lemma}
\label{lem:DAG-cut_reduction}
Let $G = (V, E, w)$ be a weighted, weakly connected DAG with weight function $w : E \rightarrow \mathbb{R}_{> 0}$. Let $H$ be the weighted directed graph obtained from $G$, by adding for each arc a reverse arc with infinite weight. Denote by $U$ (resp. $W$) the set of children (resp. parents) of the roots (resp. leaves) of $G$. Then,
\begin{enumerate}[label={(\alph*)},noitemsep,topsep=3pt]
\item $C=(S, T)$ is a non-trivial DAG-cut in $G$ with weight $k < \infty$ if and only if for some $s \in U$ and $t \in W\setminus \{ s\}$, $C$ is a minimal directed $s$-$t$ cut in $H$ with weight $k<\infty$ .

\item No non-trivial DAG-cut in $G$ exists if and only if for all $s \in U$ and $t \in W\setminus \{ s\}$ no minimal directed $s$-$t$ cut in $H$ with finite weight exists.
\end{enumerate}
\end{lemma}
\begin{proof}
($a, \Rightarrow$) By definition, $C$ is a minimal cut in $G$, $T$ is a sinkset, and $|S|, |T| \geq 2$. We must then have that $S$ (resp. $T$) contains at least one root (resp. leaf) of $G$ and a child $s$ (resp. parent $t$) of this vertex. (Note that they can not be the same.) If this were not the case, either $T$ would not be a sinkset, or the cut would not be minimal (e.g. if $T$ consists of just two leaves). We now have that $s \in U$ and $t \in W \setminus \{ s \}$. Thus, $C$ is clearly an $s$-$t$ cut in $H$. Furthermore, it has exactly weight $k$, because the arcs going from $S$ to $T$ in $H$ are exactly the arcs in the original cut-set in $G$. We have no infinite weight arcs going from $S$ to $T$ in $H$, since $T$ was a sinkset in $G$. Lastly, $C$ is also minimal in $H$, since it was minimal in~$G$.

($a, \Leftarrow$) Because $C$ has finite weight in $H$, it does not have any of the infinite weight arcs in its cut-set. This must mean that no infinite weight arc goes from $S$ to $T$ in $H$, or equivalently no arc in $G$ goes from $T$ to $S$. Thus $T$ must be a sinkset in $G$. This also means that the parent of $s$ in $G$ (which by definition is a root of $G$) is in the set $S$, otherwise $T$ would be no sinkset in $G$ any more. Similarly, the child of $t$ in $G$ (which is a leaf of $G$) must be in the set $T$. Thus, we have that $|S|, |T| \geq 2$. Again, the only arcs going from $S$ to $T$ in $H$ are exactly the arcs in the cut-set of $C$ in $G$. So, the weights also coincide. Lastly, $C$ is minimal in $G$, because it was minimal in $H$.

($b$) This follows directly from (a), and the fact that $G$ only has finite weight cuts.
\end{proof}

By the previous lemma, we can find a minimum-weight non-trivial DAG-cut by finding a minimum directed $s$-$t$ cut in the auxiliary graph $H$ for all $s \in U$ and $t \in W\setminus\{s\}$. A minimum weight directed $s$-$t$ cut in a directed graph with $n$ vertices and $m$ arcs can be found in $O (n\cdot m)$ time with the max-flow algorithm from \cite{orlin2013max}. By repeating this algorithm for $O (n^2)$ choices of $s$ and $t$, we can therefore find a minimum-weight non-trivial DAG-cut in $O (n^3 \cdot m)$ time.

The idea behind our scanwidth-heuristic is to recursively split the graph at a smallest non-trivial DAG-cut. Consequently, we obtain an upper and a lower subgraph. However, when considering scanwidth we can not just `forget' about the arcs in the cut, as they might also be counted at vertices lower or higher in the graph. Thus, for both created graphs, we merge the other part of the graph into one `supervertex'. This ensures the arcs in the DAG-cut are still accounted for. It also explains why we look for non-trivial DAG-cuts: otherwise, the merging operation will not decrease the size of our graph. Whenever no non-trivial DAG-cut exists, the graph is very small, and we simply take an arbitrary extension. This leads to the heuristic described in \cref{alg:DAGcut_heuristic}.

\begin{algorithm}[htb]
\caption{Repeated DAG-cut-splitting heuristic to find an extension.}
\label{alg:DAGcut_heuristic}
\Input{Weakly connected DAG $G = (V, E)$.}
\Output{Extension $\sigma_G$.}
\SetKwFunction{antichain}{MinDAGCutSplit}
$G'\gets$ weighted version of $G$ with unit weights\\
$\sigma_G \gets \antichain (G')$ \\
\Return{$\sigma_G$}\\
\setcounter{AlgoLine}{0}
\nonl \Proc(\tcp*[f]{$H$ is a weighted graph.}){$\antichain(H)$}{
\If{$H$ has no non-trivial DAG-cut}{
$\sigma \gets$ arbitrary extension of $H$ \tcp*{Use reverse BFS traversal.}
}
\Else{
$ C = (S, T) \gets$ minimum-weight non-trivial DAG-cut of $H$ \\
$H_1 \gets H[S]$; $H_2 \gets H[T]$\\
add a vertex $x$ to $H_1$ and a vertex $y$ to $H_2$\\
\For{$uv  \in E(H): u \in S, v \in T$}{
add an arc $ux$ to $H_1$ (if it already exists increase the weight by 1)\\
add an arc $yv$ to $H_2$ (if it already exists increase the weight by 1)
}
$\sigma_1 \gets \antichain (H_1)$ \\
$\sigma_2 \gets \antichain (H_2)$ \\
$\sigma \gets  \sigma_2[T] \circ \sigma_1[S] $\\
}
\Return{$\sigma$}}
\end{algorithm}

In \cref{subfig:cut_split1,subfig:cut_split2} the first iteration of the algorithm is visualized. \cref{subfig:cut_split_bad} shows the canonical tree extension corresponding to the extension resulting from the algorithm. We indeed see the cut $C$ reappearing. \cref{subfig:cut_split_opt} shows the optimal tree extension which has a smaller scanwidth than the one in \cref{subfig:cut_split_bad}. Thus, in general, it is not necessarily true that the smallest non-trivial DAG-cut appears in an optimal extension. Consequently, the algorithm is not necessarily optimal. Nonetheless, the algorithm remains practical, running in polynomial time, as formalized in the following theorem.

\begin{figure}[htb]
     \centering
    \parbox{.35\textwidth}{
     \begin{subfigure}[b]{0.98\linewidth}
         \centering
         \begin{tikzpicture}[xscale=0.45, yscale=0.35]
	\begin{pgfonlayer}{nodelayer}
		\node [style={graph_node}, label=below left:{$b$}] (30) at (0, 3) {};
		\node [style={graph_node}, label=above left:{$a$}] (31) at (2, 4) {};
		\node [style={graph_node}, label=below left:{$c$}] (32) at (2, 2) {};
		\node [style={graph_node}, label=below left:{$d$}] (33) at (4, 1) {};
		\node [style={graph_node}, label=above:{$\rho$}] (34) at (4, 5) {};
		\node [style={graph_node}, label=below left:{$e$}] (35) at (6, 4) {};
		\node [style={graph_node}, label=above:{$f$}] (36) at (8, 3) {};
		\node [style={graph_node}, label=above right:{$g$}] (37) at (6, 0) {};
		\node [style=none,label=right:{$C$}] (38) at (6.5, 5.25) {};
		\node [style=none] (39) at (5.5, -0.5) {};
	\end{pgfonlayer}
\draw [style={cut_line}, in=120, out=-165, looseness=1.25] (38.center) to (39.center);
	\begin{pgfonlayer}{edgelayer}
		\draw [style={graph_edge}, bend right=15, looseness=1.25] (34) to (30);
		\draw [style={graph_edge}, bend right=15, looseness=1.25] (30) to (33);
		\draw [style={graph_edge}, bend left=15, looseness=1.25] (34) to (36);
		\draw [style={graph_edge}] (34) to (31);
		\draw [style={graph_edge}] (31) to (30);
		\draw [style={graph_edge}] (30) to (32);
		\draw [style={graph_edge}] (32) to (33);
		\draw [style={graph_edge}] (33) to (37);
		\draw [style={graph_edge}] (34) to (33);
		\draw [style={graph_edge}] (31) to (32);
		\draw [style={graph_edge}] (34) to (35);
		\draw [style={graph_edge}] (35) to (36);
		\draw [style={graph_edge}] (35) to (37);
	\end{pgfonlayer}
\end{tikzpicture}
         \caption{Weakly connected DAG $G$}
         \label{subfig:cut_split1}
     \end{subfigure}
     \begin{subfigure}[b]{0.98\linewidth}
         \centering
         \begin{tikzpicture}[xscale=0.45, yscale=0.35]
	\begin{pgfonlayer}{nodelayer}
		\node [style={graph_node}, label=below left:{$b$}] (30) at (0, 3) {};
		\node [style={graph_node}, label=above left:{$a$}] (31) at (2, 4) {};
		\node [style={graph_node}, label=below left:{$c$}] (32) at (2, 2) {};
		\node [style={graph_node}, label=below left:{$d$}] (33) at (4, 1) {};
		\node [style={graph_node}, label=above:{$\rho$}] (34) at (4, 5) {};
		\node [style={graph_node}, label=below left:{$x$}, fill=black] (35) at (6, 0.5) {};
		\node [style={graph_node}, label=below left:{$e$}] (41) at (7.75, 4) {};
		\node [style={graph_node}, label=above:{$f$}] (42) at (9.75, 3) {};
		\node [style={graph_node}, label=above right:{$g$}] (43) at (7.75, 0) {};
		\node [style={graph_node}, label=above:{$y$}, fill=black] (44) at (5.75, 4.75) {};
	\end{pgfonlayer}
	\begin{pgfonlayer}{edgelayer}
		\draw [style={graph_edge}, bend right=15, looseness=1.25] (34) to (30);
		\draw [style={graph_edge}, bend right=15, looseness=1.25] (30) to (33);
		\draw [style={graph_edge}] (34) to (31);
		\draw [style={graph_edge}] (31) to (30);
		\draw [style={graph_edge}] (30) to (32);
		\draw [style={graph_edge}] (32) to (33);
		\draw [style={graph_edge}] (34) to (33);
		\draw [style={graph_edge}] (31) to (32);
		\draw [style={graph_edge}] (34) to node[midway, above,sloped]{\footnotesize{$2$}} (35) ;
		\draw [style={graph_edge}] (33) to (35);
		\draw [style={graph_edge}] (41) to (42);
		\draw [style={graph_edge}] (41) to (43);
		\draw [style={graph_edge}] (44) to (41);
		\draw [style={graph_edge}] (44) to (43);
		\draw [style={graph_edge}, bend left, looseness=1.25] (44) to (42);
	\end{pgfonlayer}
\end{tikzpicture}
         \caption{Graphs $H_1$ and $H_2$}
         \label{subfig:cut_split2}
     \end{subfigure}
     }
     \parbox{.64\textwidth}{
     \begin{subfigure}[b]{0.49\linewidth}
         \centering
         \begin{tikzpicture}[scale=0.55]
\begin{scope}[yscale=1,xscale=-1]
	\begin{pgfonlayer}{nodelayer}
		\node [style={graph_node}, label=above left:{$b$}] (62) at (2, 6) {};
		\node [style={graph_node}, label=above left:{$a$}] (63) at (1, 7.5) {};
		\node [style={graph_node}, label=above left:{$c$}] (64) at (3, 4.5) {};
		\node [style={graph_node}, label=above left:{$d$}] (65) at (4, 3) {};
		\node [style={graph_node}, label=above left:{$\rho$}] (66) at (0, 9) {};
		\node [style={graph_node}, label=above left:{$e$}] (67) at (5, 1.5) {};
		\node [style={graph_node}, label=above right:{$f$}] (68) at (4, 0) {};
		\node [style={graph_node}, label=above left:{$g$}] (69) at (6, 0) {};
		\node [style=none] (70) at (4.5, 1.75) {};
		\node [style=none, label=right:{$C$}] (71) at (3.25, 1.75) {};
		\node [style=none] (72) at (5.25, 3) {};
	\end{pgfonlayer}
		\draw [style={extension_edge}] (66.center) to (63.center);
		\draw [style={extension_edge}] (63.center) to (62.center);
		\draw [style={extension_edge}] (62.center) to (64.center);
		\draw [style={extension_edge}] (64.center) to (65.center);
		\draw [style={extension_edge}] (65.center) to (67.center);
		\draw [style={extension_edge}] (67.center) to (68.center);
		\draw [style={extension_edge}] (67.center) to (69.center);
		\draw [style={cut_line}, bend left=15] (71.center) to (72.center);	
	\begin{pgfonlayer}{edgelayer}
		\draw [style={graph_edge}, in=150, out=-75, looseness=0.75] (66) to (62);
		\draw [style={graph_edge}, bend right=15, looseness=0.75] (62) to (65);
		\draw [style={graph_edge}] (66) to (63);
		\draw [style={graph_edge}] (63) to (62);
		\draw [style={graph_edge}] (62) to (64);
		\draw [style={graph_edge}] (64) to (65);
		\draw [style={graph_edge}, bend left=15] (65) to (69);
		\draw [style={graph_edge}, in=90, out=-45, looseness=0.50] (66) to (65);
		\draw [style={graph_edge}, bend left=15] (63) to (64);
		\draw [style={graph_edge}, in=135, out=-90, looseness=0.50] (66) to (67);
		\draw [style={graph_edge}] (67) to (68);
		\draw [style={graph_edge}] (67) to (69);
		\draw [in=135, out=-105, looseness=0.50] (66) to (70.center);
		\draw [style={graph_edge}, in=75, out=-30, looseness=0.75] (70.center) to (68);
	\end{pgfonlayer}
\end{scope}
\end{tikzpicture}
         \caption{Heuristic tree extension $\Gamma_1$}
         \label{subfig:cut_split_bad}
     \end{subfigure}     
     \hfill
     \begin{subfigure}[b]{0.49\linewidth}
         \centering
         \begin{tikzpicture}[scale=0.55]
\begin{scope}[yscale=1,xscale=-1]
	\begin{pgfonlayer}{nodelayer}
		\node [style={graph_node}, label=above left:{$b$}] (54) at (-6.5, 4.5) {};
		\node [style={graph_node}, label=above left:{$a$}] (55) at (-7.25, 6) {};
		\node [style={graph_node}, label=above left:{$c$}] (56) at (-5.75, 3) {};
		\node [style={graph_node}, label=above left:{$d$}] (57) at (-5, 1.5) {};
		\node [style={graph_node}, label=above left:{$\rho$}] (58) at (-8.25, 9) {};
		\node [style={graph_node}, label=above left:{$e$}] (59) at (-8.25, 7.5) {};
		\node [style={graph_node}, label=left:{$f$}] (60) at (-9.25, 6) {};
		\node [style={graph_node}, label=above left:{$g$}] (61) at (-4.25, 0) {};
	\end{pgfonlayer}
		\draw [style={extension_edge}] (58.center) to (59.center);
		\draw [style={extension_edge}] (59.center) to (60.center);
		\draw [style={extension_edge}] (59.center) to (55.center);
		\draw [style={extension_edge}] (55.center) to (54.center);
		\draw [style={extension_edge}] (54.center) to (56.center);
		\draw [style={extension_edge}] (56.center) to (57.center);
		\draw [style={extension_edge}] (57.center) to (61.center);
	\begin{pgfonlayer}{edgelayer}
		\draw [style={graph_edge}, in=150, out=-90, looseness=0.50] (58) to (54);
		\draw [style={graph_edge}, bend right=15] (54) to (57);
		\draw [style={graph_edge}, in=60, out=-105] (58) to (60);
		\draw [style={graph_edge}, in=120, out=-75, looseness=0.75] (58) to (55);
		\draw [style={graph_edge}] (55) to (54);
		\draw [style={graph_edge}] (54) to (56);
		\draw [style={graph_edge}] (56) to (57);
		\draw [style={graph_edge}] (57) to (61);
		\draw [style={graph_edge}, in=90, out=-75, looseness=0.75] (58) to (57);
		\draw [style={graph_edge}, bend left=15] (55) to (56);
		\draw [style={graph_edge}] (58) to (59);
		\draw [style={graph_edge}] (59) to (60);
		\draw [style={graph_edge}, in=150, out=-90, looseness=0.25] (59) to (61);
	\end{pgfonlayer}
\end{scope}
\end{tikzpicture}
         \caption{Optimal tree extension $\Gamma_2$}
         \label{subfig:cut_split_opt}
     \end{subfigure}
     }
        \caption{(a): Weakly connected DAG $G$, with the unique minimum non-trivial DAG-cut $C$. (b): Weighted graphs $H_1$ and $H_2$ created after one iteration of \cref{alg:DAGcut_heuristic}. The arc $\rho x$ has a weight of 2, while the other arcs have unit weights. The two `supervertices' $x$ and $y$ are in black. (c): Canonical tree extension $\Gamma_1$ corresponding to the extension obtained by \cref{alg:DAGcut_heuristic}. The cut $C$ appears again. The tree extension is not optimal, and has a scanwidth of 6. (d): Optimal tree extension $\Gamma_2$ (of scanwidth 5) which does not contain the cut $C$.}
\end{figure}
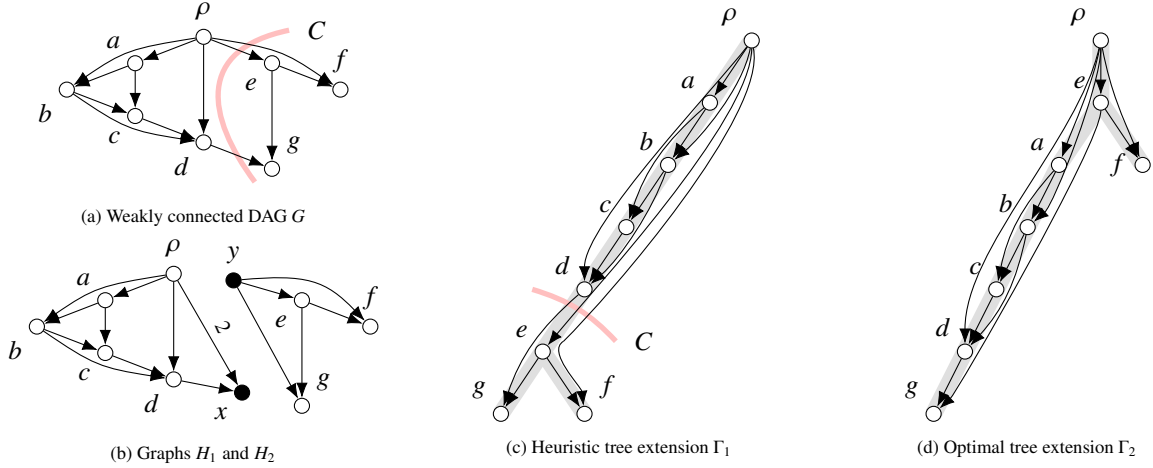

\begin{theorem}
Let $G=(V,E)$ be a weakly connected DAG of $n$ vertices and $m$ arcs. Then, \cref{alg:DAGcut_heuristic} returns an extension of $G$ and runs in $O(n^4 \cdot m)$ time.
\end{theorem}
\begin{proof}
\emph{Correctness:} We will show that the subroutine \texttt{MinDAGCutSplit} always returns an extension of the graph $H$, which will imply that the complete algorithm returns an extension of $G$. We will do this by strong induction on the number of vertices $k$ of $H$.

\emph{Base case:} Whenever $k \in \{1,2,3\}$, $H$ never has a non-trivial DAG-cut, because such a cut needs at least 2 vertices on either side. Then, the procedure returns an extension by reversing a BFS traversal.

\emph{Induction step:} Let $k\geq 4$ be arbitrary, and assume that the statement holds for all $1\leq \ell \leq k-1$. We can furthermore assume that $H$ has a non-trivial DAG-cut, otherwise the procedure will automatically return an extension. We have $|S|, |T| \leq k-2$, since the DAG-cut is non-trivial and cuts off at least 2 vertices on either side. After adding the `supervertices', we thus have that $H_1$ and $H_2$ both have at most $k-1$ vertices. By the induction hypothesis, $\sigma_1$ is then an extension of $H_1$. Because $H[S]$ is a subgraph of $H_1$, $\sigma_1[S]$ (which restricts $\sigma_1$ to $S$) is an extension of $H[S]$. Analogously, we can obtain that $\sigma_2[T]$ is an extension of $H[T]$. Since $C$ was a DAG-cut, $T$ is a sinkset of $G$. This then means that $\sigma = \sigma_2[T] \circ \sigma_1[S]$ is an extension of $H$.

\emph{Time complexity:} Let $T(n, m)$ be the time the algorithm takes to run on a graph with $n$ vertices and $m$ arcs. According to the discussion after \cref{lem:DAG-cut_reduction}, finding the minimum non-trivial DAG-cut of such a graph (or showing that no non-trivial DAG-cut exists) can be done in $O(n^3 \cdot m)$ time. As the other non-recursive parts in the procedure are also dominated by this complexity and graphs with less than four vertices are not split any further, we obtain the following recurrence relation for some constant $c>0$:
$$\begin{dcases}
  T(3, m) \leq c,  & \text{ if } n = 3 \,(\text{and so } m \leq 3); \\
  T(n, m) \leq \max_{2\leq k \leq n-2} T(k + 1, m) + T(n - k +1, m) + c \cdot n^3 \cdot m  , & \text{ if } n \geq 4.
\end{dcases}$$

We prove by strong induction on $n\geq 1$ that $T(n,m) \leq d \cdot n^4 \cdot m$ for some constant $d>0$. The base cases for $n \leq 3$ are trivial. Now assume that the induction hypothesis holds for $1 \leq \ell < n$. Then, using the recurrence relation, we get 
\begin{align*}
T(n, m) &\leq  \max_{2\leq k \leq n-2} \left\{ d\cdot (k+1)^4 \cdot m +d \cdot (n-k+1)^4 \cdot m\right\} +c\cdot n^3 \cdot m \\
&=  d\cdot 3^4 \cdot m +d \cdot (n-2)^4 \cdot m +c\cdot n^3 \cdot m.
\end{align*}
In the last equality we used that the above maximum is attained at the extreme value of $k=2$ (or equivalently, $k=n-2$) for fixed $m$ and $n$, which can easily be shown with a concavity argument. We can now choose $d$ large enough (and independent of $m$ and $n$) such that this expression is at most $d \cdot n^4 \cdot m$. This proves the theorem.
\end{proof}

\subsection{Greedy heuristic and simulated annealing}
\label{sec:greedy_heuristic}
In this subsection we discuss two other heuristic methods. Since simulated annealing and greedy algorithms are well known and intuitive, we will only briefly describe the general ideas and refer the interested reader to \cite[Ch.\,5]{holtgrefe2023scanwidth} for the (algorithmic) details and (complexity) analyses.

\paragraph{Greedy heuristic}
The idea behind our \emph{greedy algorithm} is to create an extension of seemingly small scanwidth by adding vertices one by one, each time adding the vertex that increases the scanwidth the least. This is achieved by maintaining a set $S$ that keeps track of the vertices that still need to be added. In each iteration, the algorithm chooses the leaf of $G[S]$ that increases the scanwidth the least. The procedure is repeated until the set $S$ is empty (or equivalently, $\sigma$ contains all vertices of $G$). This consecutive `picking' of leaves could lead to an optimal extension if the correct leaf would be chosen in each iteration. This is because we can create any extension by consecutively picking leaves of a graph.

The greedy rule will not necessarily pick the correct leaf each time. One can construct instances where the algorithm will perform very badly. An example is depicted in \cref{fig:bad_greedy_scan}. The figure shows a class of graphs that all have scanwidth 5, yet the greedy algorithm will construct a solution of scanwidth $n$, with $n$ half the number of vertices in the graph.

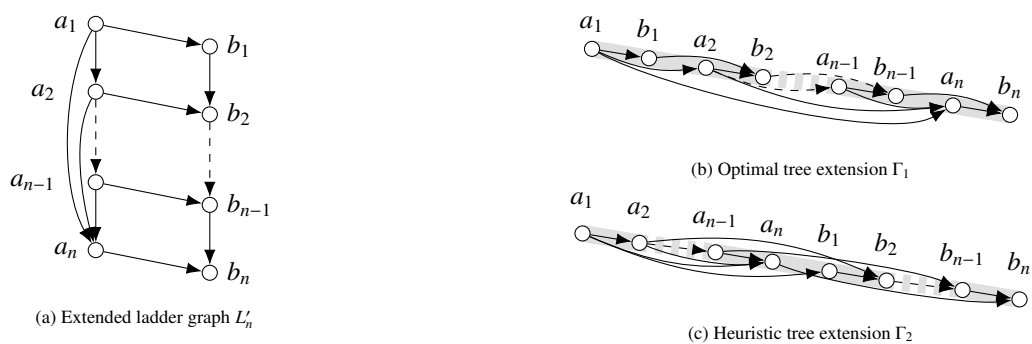
\begin{figure}[htb]
\centering
\valign{#\cr
  \hsize=0.35\columnwidth
  \begin{subfigure}{0.35\columnwidth}
  \centering
  \begin{tikzpicture}[scale=0.6]
	\begin{pgfonlayer}{nodelayer}
		\node [style={graph_node},label=left:{$a_1$}] (1) at (0, 5) {};
		\node [style={graph_node},label=right:{$b_1$}] (2) at (2.5, 4.5) {};
		\node [style={graph_node},label=right:{$b_2$}] (3) at (2.5, 3) {};
		\node [style={graph_node},label={[label distance=.3cm]left:{$a_2$}}] (4) at (0, 3.5) {};
		\node [style={graph_node},label={[label distance=.3cm]left:{$a_{n-1}$}}] (5) at (0, 1.5) {};
		\node [style={graph_node},label=right:{$b_{n-1}$}] (6) at (2.5, 1) {};
		\node [style={graph_node},label=right:{$b_n$}] (7) at (2.5, -0.5) {};
		\node [style={graph_node},label=left:{$a_n$}] (8) at (0, 0) {};
	\end{pgfonlayer}
	\begin{pgfonlayer}{edgelayer}
		\draw [style={graph_edge}] (2) to (3);
		\draw [style={graph_edge}] (1) to (4);
		\draw [style={graph_edge}] (1) to (2);
		\draw [style={graph_edge}] (4) to (3);
		\draw [style={graph_edge}] (6) to (7);
		\draw [style={graph_edge}] (5) to (8);
		\draw [style={graph_edge}] (8) to (7);
		\draw [style={graph_edge}] (5) to (6);
		\draw [style={graph_edge_dotted}] (3) to (6);
		\draw [style={graph_edge_dotted}] (4) to (5);
		\draw [style={graph_edge}, in=120, out=-120, looseness=0.75] (1) to (8);
		\draw [style={graph_edge}, in=105, out=-120, looseness=0.75] (4) to (8);
	\end{pgfonlayer}
\end{tikzpicture}
  \caption{Extended ladder graph $L_n'$}
  \end{subfigure}
  
  \cr\noalign{\hfill}
  \hsize=0.6\columnwidth
  
  \begin{subfigure}{0.6\columnwidth}
  \centering
  \begin{tikzpicture}[scale=0.5]
	\begin{pgfonlayer}{nodelayer}
		\node [style={graph_node},label=above:{$a_1$}] (1) at (1.5, -0.25) {};
		\node [style={graph_node},label=above:{$b_1$}] (2) at (3, -0.5) {};
		\node [style={graph_node},label=above:{$b_2$}] (3) at (6, -1) {};
		\node [style={graph_node},label=above:{$a_2$}] (4) at (4.5, -0.75) {};
		\node [style={graph_node},label=above:{$a_{n-1}$}] (5) at (8, -1.25) {};
		\node [style={graph_node},label=above:{$b_{n-1}$}] (6) at (9.5, -1.5) {};
		\node [style={graph_node},label=above:{$b_n$}] (7) at (12.5, -2) {};
		\node [style={graph_node},label=above:{$a_n$}] (8) at (11, -1.75) {};
	\end{pgfonlayer}
	\begin{pgfonlayer}{edgelayer}
		\draw [style={graph_edge}, bend left=15] (2) to (3);
		\draw [style={graph_edge}, bend right=15] (1) to (4);
		\draw [style={graph_edge}] (1) to (2);
		\draw [style={graph_edge}] (4) to (3);
		\draw [style={graph_edge}, bend left=15] (6) to (7);
		\draw [style={graph_edge}, bend right=15] (5) to (8);
		\draw [style={graph_edge}] (8) to (7);
		\draw [style={graph_edge}] (5) to (6);
		\draw [style={graph_edge_dotted}, bend left=15] (3) to (6);
		\draw [style={graph_edge_dotted}, bend right=15] (4) to (5);
	\end{pgfonlayer}		
		\draw [style={extension_edge}] (1.center) to (2.center);
		\draw [style={extension_edge}] (2.center) to (4.center);
		\draw [style={extension_edge}] (4.center) to (3.center);
		\draw [style={extension_edge}] (5.center) to (6.center);
		\draw [style={extension_edge}] (6.center) to (8.center);
		\draw [style={extension_edge_dashed}] (3.center) to (5.center);
		\draw [style={extension_edge}] (8.center) to (7.center);
		\draw [style={graph_edge}, bend right=15] (4) to (8);
		\draw [style={graph_edge}, in=-150, out=-30, looseness=0.50] (1) to (8);
\end{tikzpicture}
  \caption{Optimal tree extension $\Gamma_1$}
  \end{subfigure}
  
  \vfill
  
  \begin{subfigure}{0.6\columnwidth}
  \centering
  \begin{tikzpicture}[scale=0.5,label distance=3]
	\begin{pgfonlayer}{nodelayer}
		\node [style={graph_node},label=above:{$a_1$}] (1) at (1.5, -0.25) {};
		\node [style={graph_node},label=above:{$b_1$}] (2) at (8, -1.25) {};
		\node [style={graph_node},label=above:{$b_2$}] (3) at (9.5, -1.5) {};
		\node [style={graph_node},label=above:{$a_2$}] (4) at (3, -0.5) {};
		\node [style={graph_node},label=above:{$a_{n-1}$}] (5) at (5, -0.75) {};
		\node [style={graph_node},label=above:{$b_{n-1}$}] (6) at (11.5, -1.75) {};
		\node [style={graph_node},label=above:{$b_n$}] (7) at (13, -2) {};
		\node [style={graph_node},label=above:{$a_n$}] (8) at (6.5, -1) {};
	\end{pgfonlayer}
	\begin{pgfonlayer}{edgelayer}
		\draw [style={graph_edge}] (2) to (3);
		\draw [style={graph_edge}] (1) to (4);
		\draw [style={graph_edge}, bend right=15] (1) to (2);
		\draw [style={graph_edge}, bend left=15] (4) to (3);
		\draw [style={graph_edge}] (6) to (7);
		\draw [style={graph_edge}] (5) to (8);
		\draw [style={graph_edge}, bend right=15, looseness=0.50] (8) to (7);
		\draw [style={graph_edge}, bend left=15, looseness=0.50] (5) to (6);
		\draw [style={graph_edge_dotted}] (3) to (6);
		\draw [style={graph_edge_dotted}] (4) to (5);
	\end{pgfonlayer}		
		\draw [style={extension_edge}] (1.center) to (4.center);
		\draw [style={extension_edge}] (5.center) to (8.center);
		\draw [style={extension_edge}] (8.center) to (2.center);
		\draw [style={extension_edge}] (2.center) to (3.center);
		\draw [style={extension_edge}] (6.center) to (7.center);
		\draw [style={extension_edge_dashed}] (4.center) to (5.center);
		\draw [style={extension_edge_dashed}] (3.center) to (6.center);
		\draw [style={graph_edge}, bend right=15] (1) to (8);
		\draw [style={graph_edge}, bend right=15, looseness=0.75] (4) to (8);
\end{tikzpicture}
  \caption{Heuristic tree extension $\Gamma_2$}
  \end{subfigure}
  
  \cr
}
\caption{(a): The \emph{extended ladder-graph} $L_n'$ (with $n\geq 3$), which is a weakly connected DAG. (b): An optimal tree extension $\Gamma_1$ of $L_n'$ with scanwidth 5. (c): The worst-case tree extension $\Gamma_2$ of $L_n'$ with scanwidth $n$. This is the canonical tree extension of the extension returned by the greedy heuristic.
}
        \label{fig:bad_greedy_scan}
\end{figure}

\paragraph{Simulated annealing} A different approach to generate an extension with small scanwidth is to use \emph{simulated annealing}: a metaheuristic that efficiently guides itself through the space of possible solutions. The algorithm starts with some initial extension, either created at random, or utilizing another heuristic. It then randomly selects a neighboring extension by swapping two vertices that are adjacent in the extension but that have no arc connecting them in the graph. This extension is accepted as the new solution if it has a smaller scanwidth. If the neighbor has a larger scanwidth, we still have a probability of accepting it depending on the decreasing \emph{temperature} parameter: the lower the temperature, the less likely it becomes to accept a worse solution. 

As a result of the high starting temperature, the algorithm will initially perform a fairly global and random search. Consequently, almost any solution will be accepted as the next state. Due to the `cooling' of the temperature, the heuristic will slowly start to limit its choices to better solutions. Thus, it nudges itself to a (hopefully global) minimum.

\section{Experimental results}\label{sec:experiments}
In this section we conduct an experimental study to evaluate the performance of our exact algorithms, heuristics and reduction rules on phylogenetic networks. In \cref{subsec:network_generation} we will describe the networks that are used in the experiments and compare their level, scanwidth and treewidth. The subsequent subsections offer analyses of the reduction scheme from \cref{sec:reductions}, the exact methods from \cref{sec:exact_methods}, and the heuristics from \cref{sec:heuristics}, respectively.

All algorithms are implemented in Python and are publicly available as an easily installable package at \begin{center}
\url{https://github.com/nholtgrefe/scanwidth}.
\end{center} This repository also provides the used networks and complete numerical results. All experiments in this section were conducted on an Intel Core i7-8750H CPU @ 2.20 GHz with 16 GB RAM.

\subsection{Network generation}\label{subsec:network_generation}
Closely following the experimental study from \cite{vIersel2023making}, we utilize both a dataset of real-world networks and a synthetically created dataset. The real data is made up of 27 real phylogenetic networks found in the literature, collected on \cite{phylo_online}. Among these \emph{real} networks, 15 are binary, while the remaining 12 are non-binary. The number of leaves ranges from 6 to 39, while the number of reticulations ranges from 1 to 9, except for one outlier with 32 reticulations.

To augment our dataset we use the birth-hybridization network generator from \cite{zhang2018bayesian}, often called the \emph{ZODS generator}. The method takes two input parameters: the speciation rate $\lambda$, and the hybridization rate~$\nu$. Following the computational experiments from \cite{janssen2021comparing, vIersel2023making, bernardini2023constructing}, we set $\lambda = 1$ and sample $\nu$ uniformly at random from the interval $[0.0001, 0.4]$ for each individual network. We adapt the implementation from \cite{vIersel2023making} to generate 100 binary networks for each pair of $(r, \ell)$, where $r \in \{ 10, 20, 30 \}$ denotes the number of reticulations, and $\ell \in \{20, 50 , 100 \}$ the number of leaves. In total, this gives rise to a dataset comprising 900 \emph{synthetic} networks. 

Scornavacca and Weller \cite{Scornavacca2022} requested a comparison of the reticulation number, level, scanwidth and treewidth for different network classes. \cref{fig:parameter_boxplots} functions as a partial answer to this call. It depicts boxplots that show the spread of the level, the treewidth and the scanwidth within each dataset. 

\begin{figure}[htb]
         \centering
	\includegraphics[width=0.99\textwidth]{{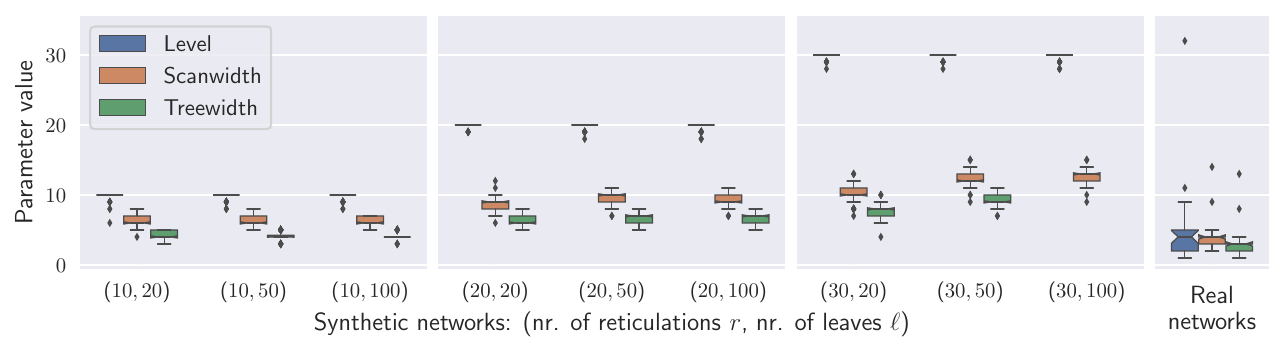}}
     \caption[.]{Variation in level, scanwidth and treewidth within each dataset. The figure displays a boxplot for each of the nine subsets of the synthetic data and one for the real dataset. The boxplots show the quartiles of the data and its outliers. Different colours indicate the three different parameters. The treewidth values of the last synthetic dataset are not presented due to computational constraints, since obtaining these values required excessive computation times surpassing 19 hours for some networks.\protect\footnotemark}
     \label{fig:parameter_boxplots}
\end{figure}

\footnotetext{The values of the scanwidth are calculated using our exact algorithms, whose performance is discussed in \cref{sec:experiment_exact}. The values of the treewidth are calculated on a different CPU using a Java implementation by Tamaki of one of the fastest known exact algorithms \cite{tamaki2022heuristic, treewidth_repo}.}

For the synthetic data, we fixed the number of reticulations, explaining why the level has a sharp cutoff at those values in the figure. Moreover, it is evident that the ZODS generator favours networks with levels very close to the reticulation number. Notably, we observe that the level + 1 is greater than or equal to the scanwidth, which in turn is greater than or equal to the treewidth. This aligns with the bounds from \cref{lem:treewidth-bound,lem:level_bound}. As mentioned in the introductory section, scanwidth was proposed as an alternative for treewidth in parametrized algorithms. Although the level and the scanwidth exhibit a considerable difference in value, the treewidth is not much smaller than the scanwidth. This observation strengthens our belief in the practical value of scanwidth as a parameter.

Regarding the real networks, we see a somewhat different trend. The values of the three parameters are closer together in this case. We attribute this to the fact that most levels of the real networks are fairly small. As a consequence, there is limited `room' for the scanwidth and the treewidth. The scanwidth, therefore, takes on predominantly small values, which further suggests its practicality. This is particularly promising since we have an algorithm capable of computing the scanwidth in polynomial time for fixed scanwidth, which becomes efficient when the scanwidth is small.

\subsection{Reductions}
\label{sec:reduction_experiments}
To test the effect the reduction rules from \cref{sec:reductions} have on the size of the networks, we employed the decomposition method (\cref{alg:reduction}) on each network. \cref{subfig:reduction_perc} showcases the percentage of the original vertices that remain after decomposition. Networks with fewer reticulations demonstrate greater potential for reduction. This is to be expected, as such networks are more tree-like, and many of their blocks thus have scanwidth 1 or 2. The decomposition algorithm effectively `deletes' these blocks, leading to a significant reduction in the overall network size.

\begin{figure}[htb]
     \begin{subfigure}[b]{0.495\textwidth}
         \centering
	\includegraphics[height=.495\textwidth]{{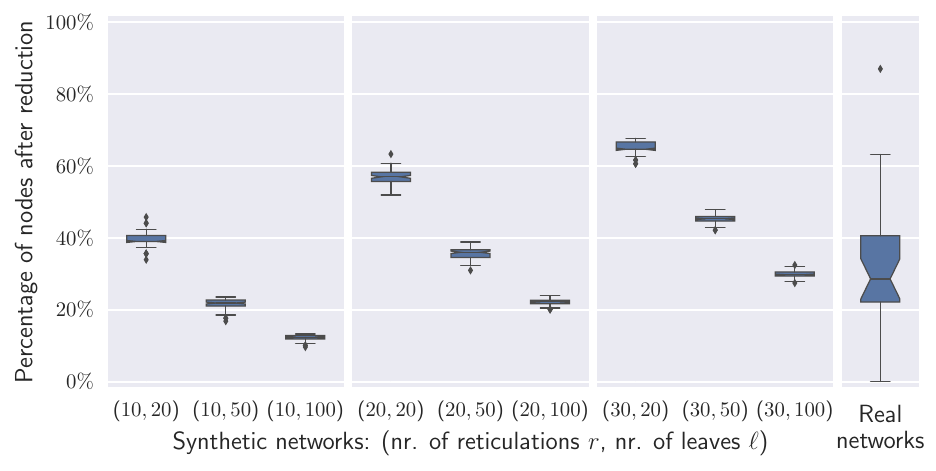}}
         \caption{Relative network size after reduction}
         \label{subfig:reduction_perc}
     \end{subfigure}
\hfill
     \begin{subfigure}[b]{0.495\textwidth}
         \centering
	\includegraphics[height=.495\textwidth]{{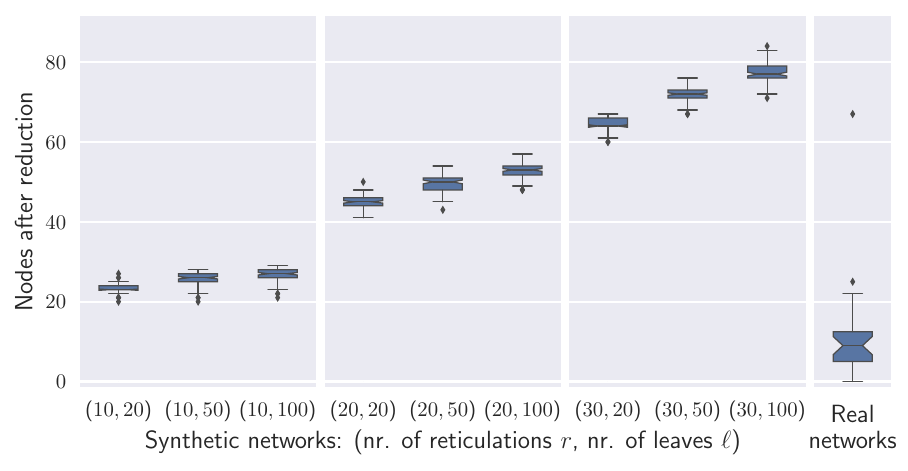}}
         \caption{Absolute network size after reduction}
         \label{subfig:reduction_absolute}
     \end{subfigure}
     \caption{Performance of the decomposition method (\cref{alg:reduction}). Figure (a) depicts how many vertices remain after the decomposition, as a percentage of the number of vertices of the original network. Figure (b) shows the absolute number of vertices after the decomposition. Both subfigures contain boxplots - showing quartiles and outliers of the data - for each of the nine subsets of the synthetic data and one for the real dataset.}
     \label{fig:reduction_performance}
\end{figure}

Regarding the number of leaves, we see a different relationship: a larger number of leaves corresponds to a greater reduction in size. This can be attributed to the fact that our reduction rules delete leaves of networks. Since the real networks vary in terms of reticulation numbers and number of leaves, their reduction percentages exhibit a wider range of values.

\cref{subfig:reduction_absolute} shows the absolute number of vertices after the decomposition of the networks. Logically, networks with more leaves and more reticulations are larger, even after decomposing them. We also see that most of the real networks are relatively small after decomposition.

The time to reduce the networks is extremely small. This outcome is not surprising given that \cref{lem:alg_reduction} proved the quadratic time complexity of the decomposition algorithm. The largest computation time for any of the instances was 0.056 seconds, while the average computation time remained below 0.005 seconds. All in all, the reduction rules prove to be beneficial without imposing substantial computational overhead. Therefore, we certainly recommend incorporating these reduction rules in practice.

\subsection{Exact algorithms}
\label{sec:experiment_exact}
In \cref{sec:exact_methods} we explored multiple exact algorithms and their respective time complexities. A brute force solution runs in $\tilde{O}(n!)$ time, while the recursive \cref{alg:recursion} runs in $\tilde{O} (4^n)$ time. The theoretically superior algorithm repeatedly applies the fixed-parameter \cref{alg:dynamic_program_XP}, as outlined in the proof of \cref{cor:sw_is_xp}. Since phylogenetic networks have a single root by definition (and thus $r=1$), the time complexity from \cref{cor:sw_is_xp} simplifies to $O (k \cdot m \cdot n^k)$, where $k$ represents the scanwidth.

While we have proved that these algorithms yield optimal solutions, it is interesting to assess how fast they run in practice. For each of the above-described algorithms, \cref{fig:exact_performance} provides insight into the percentage of networks for which the scanwidth can be determined within 60 seconds, after first applying the decomposition algorithm. Throughout the different data (sub)sets, the order of the algorithms concerning their completion rates aligns with the theoretical time complexities of the algorithms. The brute-force solution has the smallest completion rate, while (the repeated application of) \cref{alg:dynamic_program_XP} achieves the highest.

\begin{figure}[htb]
         \centering
	\includegraphics[width=0.98\textwidth]{{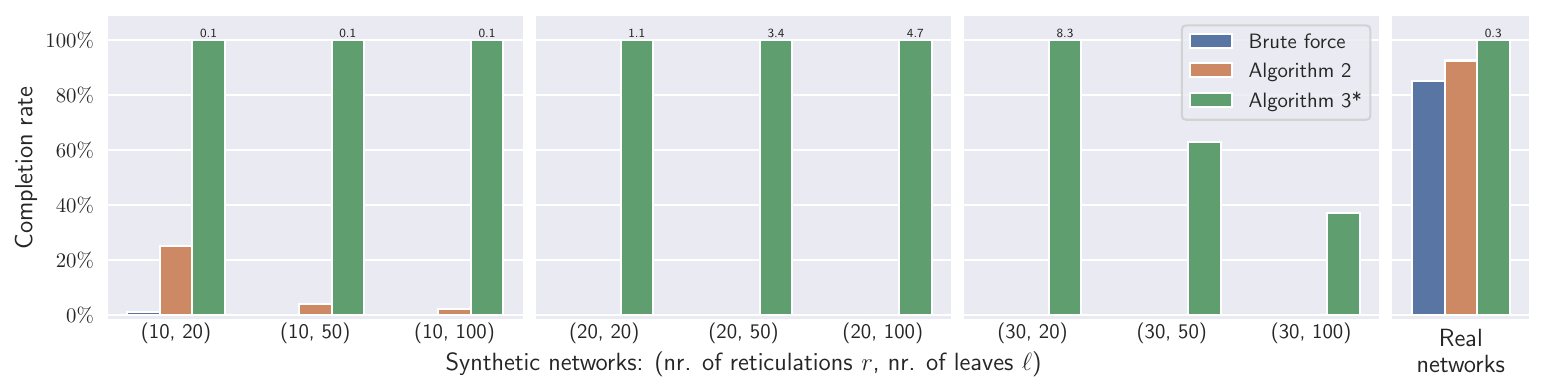}}
     \caption{Performance of the exact algorithms. For each algorithm, the completion rate, calculated as the number of instances per data set where the algorithm successfully computed the scanwidth within 60 seconds, is shown. If the rate is $100 \%$, the average running time in seconds is also depicted. The figure contains results for each of the nine subsets of the synthetic data and one for the real dataset. The different colours indicate the different algorithms. In all cases, the decomposition algorithm was also applied. The asterisk in the legend indicates that we repeatedly applied \cref{alg:dynamic_program_XP} as outlined in the proof of \cref{cor:sw_is_xp}, since the algorithm itself only solves the fixed-parameter version of the problem.}
     \label{fig:exact_performance}
\end{figure}

Interestingly, for the real networks, (the repeated application of) \cref{alg:dynamic_program_XP} achieves a $100\%$ completion rate and an average computation time of just $0.3$ seconds. Thus, on the real dataset the algorithms perform extremely well. However, the brute-force solution also attains a significant completion rate of $85 \%$, indicating that most of these real instances are not too hard after applying the decomposition algorithm.

In general, the completion rates drop as the number of leaves and reticulations increases. This is in line with our discussion from the previous two subsections, where it is noted that these networks are larger and have a higher scanwidth and level.

Looking at the complete synthetic dataset, the best-performing algorithm (i.e. repeated application of \cref{alg:dynamic_program_XP}) had a completion rate of $88.9 \%$ within 60 seconds. Additionally, we allowed this algorithm to run indefinitely to get the scanwidth values for all networks. After 300 seconds, the overall completion rate increased to $98.9 \%$. The maximum computation time of any of the networks using this algorithm turned out to be $453.52$ seconds. Hence, for all generated networks, with up to 30 reticulations and 100 leaves, the scanwidth could be computed exactly within 8 minutes.

\subsection{Heuristics}
\label{sec:experiment_heuristics}
In the previous section, we observed that our fastest algorithm is able to find all optimal tree extensions within 500 seconds. If less time is allowed, we need to turn our attention to heuristics. We evaluate the performance of the greedy heuristic (see \cref{sec:greedy_heuristic}) and the cut-splitting heuristic (\cref{alg:DAGcut_heuristic}) developed in \cref{sec:heuristics}. Additionally, we apply simulated annealing (see \cref{sec:greedy_heuristic}) to both results and compare the performances. Note that we have not proved any theoretical bounds on the approximation ratios of the heuristics. We refer to \cite{holtgrefe2023scanwidth} for the used parameter-settings of the simulated annealing algorithm.

\cref{fig:heuristic_performance} presents the results of the experiment. In \cref{subfig:heuristic_ratio} the practical approximation ratios obtained by the heuristics are shown for the different datasets. First of all, we observe that for the real networks, the practical approximation ratios are very close to 1. In fact, the whiskers of the boxplots are not even shown, indicating that for most networks all heuristics can find an optimal (tree) extension. Applying the cut-splitting heuristic together with simulated annealing even solves all but one network to optimality.

\begin{figure}[htb]
     \begin{subfigure}[b]{\textwidth}
         \centering
         \includegraphics[width=0.98\textwidth]{{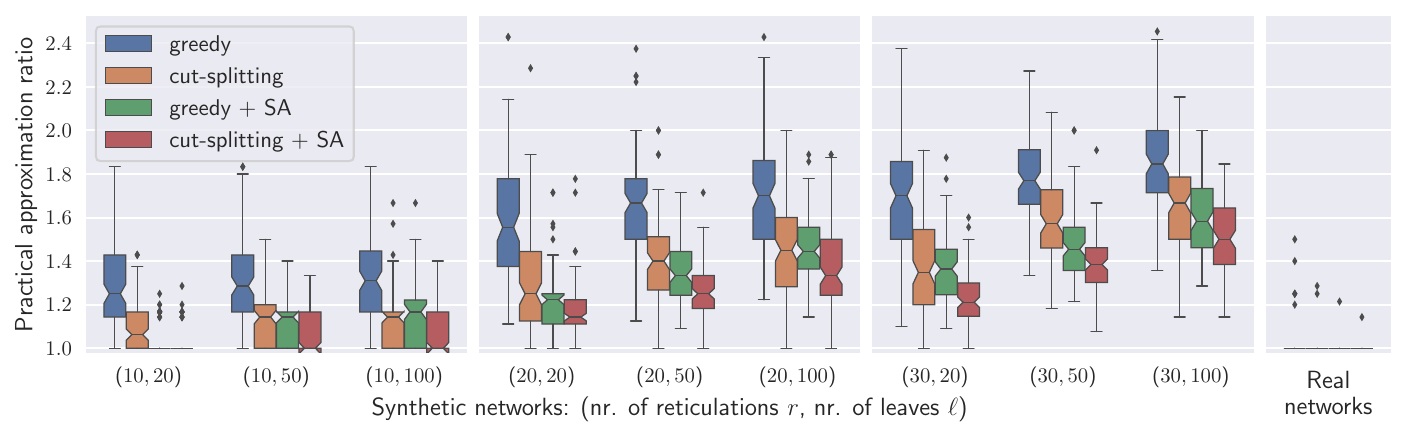}}
         \caption{Practical approximation ratios}
         \label{subfig:heuristic_ratio}
     \end{subfigure}
\hfill
     \begin{subfigure}[b]{\textwidth}
         \centering
	\includegraphics[width=0.98\textwidth]{{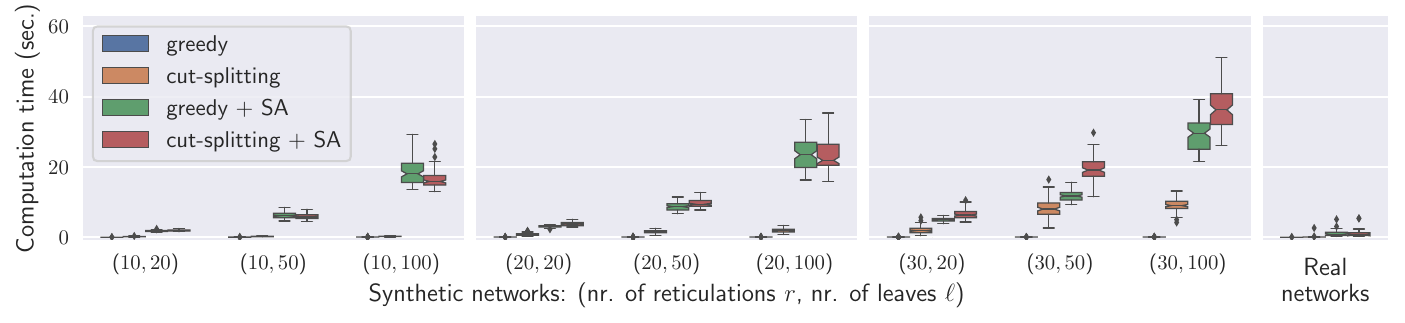}}
         \caption{Computation times}
         \label{subfig:heuristic_time}
     \end{subfigure}
     \caption{Performance of the different heuristics. Boxplots are shown for each of the nine subsets of the synthetic data and the real dataset. The boxplots show the quartiles of the data and its outliers. Figure (a) displays the practical approximation ratios, while Figure (b) depicts the computation times. The computation times for simulated annealing (SA) include the computation time to obtain the initial tree extension. We also applied the decomposition algorithm in all cases.}
     \label{fig:heuristic_performance}
\end{figure}

For the synthetic data, the approximation ratios increase with the number of reticulations and leaves. This is attributed to the fact that such networks are inherently more complex and thus more challenging. It is fairly clear that the greedy heuristic performs the worst, although the practical approximation ratio stays below 2.5. In all cases, the cut-splitting heuristic consistently shows significant improvement over the greedy algorithm. Lastly, applying simulated annealing indeed improves the solution quality even more, albeit at the cost of more computation time.

The computation times of the heuristics are visualized in \cref{subfig:heuristic_time}. It is apparent that simulated annealing takes considerably more computation time compared to just applying the heuristics on their own. Furthermore, we see that the better-performing heuristics require more time than the less effective ones.

\section{Conclusions}\label{sec:conclusion}
The main contribution of this paper is threefold: we have presented and implemented a relatively fast exact algorithm to compute the scanwidth, we have provided insight into the parametrized complexity of the \textsc{Scanwidth} problem, and we developed a heuristic to compute the scanwidth for instances that are too large for the exact method.

Regarding the first topic, we first improve upon a brute-force approach by a recursive exponential time algorithm. The recursive relations in this method form the foundation of an exact algorithm that can compute the scanwidth of DAGs with a fixed number of roots in slicewise polynomial time for fixed scanwidth. This algorithm iterates from top-to-bottom through different vertex subsets of a graph while storing intermediate results. The time complexity can be bounded by $O((k+r-1) \cdot m \cdot n^{k+r-1})$ for DAGs with $r$ roots, where $k$ is the scanwidth. Furthermore, in combination with a decomposition algorithm, the time complexity of the algorithm can be bounded by $O(2^{4 \ell - 1} \cdot \ell \cdot n + n^2)$ for level-$\ell$ networks.

The worst-case time complexity of this algorithm shows that for DAGs with a fixed number of roots \textsc{Scanwidth} is part of the complexity class XP with respect to its natural parametrization (i.e. with respect to the parameter scanwidth itself). Moreover, the above time complexity for level-$\ell$ networks shows that \textsc{Scanwidth} is fixed-parameter tractable with respect to the level of a network. Specifically, it takes quadratic time to calculate the scanwidth of a network when the level is fixed. It remains open whether \textsc{Scanwidth} lies in XP with respect to its natural parametrization when the number of roots of the DAG is unbounded. An even stronger open question is whether the problem is FPT with respect to its natural parametrization (with, or possibly without, a bounded number of roots), as is known to be the case for cutwidth~\cite{bodlaender2009derivation} and treewidth~\cite{bodlaender1993linear}.

We observe that the cuts in a tree extension are of a specific type: DAG-cuts. Using the fact that one can find a smallest (non-trivial) DAG-cut in polynomial time, we are able to efficiently keep splitting a graph into smaller subgraphs at these cuts. Although not necessarily optimal, we showed that this heuristic performs great in practice, especially when enhanced with simulated annealing.

Tested on a set of 27 real-world rooted phylogenetic networks, the exact XP algorithm performs best. The algorithm is able to compute the scanwidth of any network within 7.86 seconds, averaging a computation time of just 0.30 seconds. On a synthetic dataset of networks, the algorithm struggles with the harder instances, albeit still able to compute any scanwidth within 500 seconds. On these fairly hard instances---with 30 reticulations and 100 leaves---the earlier described heuristic attains an average practical approximation ratio of 1.5.

These experimental results show that computing the scanwidth exactly or finding a near-optimal solution can surely be done in a reasonable time. Additionally, we show that in practice, the treewidth---scanwidth's main competitor when it comes to parametrized algorithms---is not much smaller for networks. Scanwidth is also more intuitive for phylogenetic networks than treewidth. Together, these observations motivate the use of scanwidth (and the corresponding tree extensions) when designing parametrized algorithms in phylogenetics. No treewidth algorithm is known for \textsc{Hybridization Number} \cite{BORDEWICH2007914}, thus making it a possible candidate for a scanwidth-based approach. On the other hand, \textsc{Tree Containment} \cite{iersel2023embedding} and \textsc{Small Parsimony} \cite{Scornavacca2022} can be parametrized by treewidth, but might be solved faster in practice when parametrized by scanwidth.

Another direction of possible further research would be the transferability of some of our results to edge-treewidth, introduced in \cite{magne2023edge}. As this parameter is very closely related to scanwidth, it seems that translating our algorithms to edge-treewidth is not far-fetched. Perhaps, it is also possible to adapt our algorithms to \emph{node-scanwidth}, where instead of arcs we are interested in the tails of the arcs (similar to \cref{def:treewidth} of treewidth). A possibly simpler generalization would be the translation of our results to arc-weighted DAGs or multigraphs.

Furthermore, it is not known whether scanwidth can be approximated efficiently. Recent research into the inapproximability of width parameters suggests that treewidth and undirected cutwidth are inapproximable up to a constant factor within polynomial time  \cite{wu2014inapproximability}. It is not unthinkable that these inapproximability results translate to scanwidth. On the positive side, treewidth can be approximated within a ratio only linearly dependent on the treewidth itself (see~\cite{korhonen2022single} for the state-of-the-art and an overview of other approximation algorithms).

Lastly, we address our experimental study. Our synthetic networks were created with a single generating method~\cite{zhang2018bayesian}. A sensible next step is then to extend the study to other network generators (see \cite{janssen2021comparing}). Moreover, we could look into specific types of networks, both from a computational and a theoretical standpoint. We already considered networks with fixed reticulation numbers and levels, but other classes also exist (see \cite{kong2022classes} for a recent survey). During our experiments, we observed that scanwidth was a lot faster to compute than treewidth. We would thus welcome efforts towards a thorough comparison of the practical computability of the two parameters, both from an exact and a heuristic point of view.

\section*{Acknowledgements}
We extend our sincere appreciation to Mathias Weller and Norbert Zeh for insightful discussions on the topic of this paper.

\appendix

\section{Omitted proofs}\label{sec:omitted_proofs}

\begin{lemma}\label{lem:GW_unique}
Let $G=(V,E)$ be a weakly connected DAG, and let $\Gamma$ and $\Omega$ be two tree extensions of $G$. Then, $\GW_v^\Gamma = \GW_v^\Omega$ for all $v \in V$ if and only if $\Gamma = \Omega$. Therefore, a tree extension is uniquely determined by the sets $\GW$.
\end{lemma}
\begin{proof}
The `if direction' is trivial, so it remains to prove the `only-if direction'. Assume that $\Gamma$ and $\Omega$ are two different tree extensions of the graph $G$. For any $v \in V$, we write $\Gamma_v$ (resp. $\Omega_v$) for the subtree of $\Gamma$ (resp. $\Omega$) rooted at $v$. As $\Gamma \neq \Omega$, we can assume without loss of generality that $V (\Gamma_u) \not\subseteq V( \Omega_u )$ for some $u \in V$.

Let $y$ be a vertex in $V (\Gamma_u) \setminus V(\Omega_u)$ that is highest in $\Gamma$. Assuming that $G$ is non-trivial, we get three cases. (i): $y$ has a parent $x$ in $G$. Then $x \notin V( \Omega_u)$ and $x \notin V(\Gamma_u)$ (otherwise, $x$ would be a higher vertex in $V (\Gamma_u) \setminus V(\Omega_u)$). Then, $xy \in \GW_u^\Gamma \setminus \GW_u^\Omega$. (ii): $y$ does not have a parent in $G$ and some successor of $y$ in $G$ is in $V(\Omega_u)$. (A \emph{successor} of $y$ is a vertex $z <_G y$.) Let $z$ be the highest such successor in $G$, then there exist some arc $xz \in E$ and $xz \in \GW_u^\Omega \setminus \GW_u^\Gamma$. (iii): $y$ does not have a parent in $G$ and all successors of $y$ in $G$ are not in $V(\Omega_u)$. Since $V(\Gamma_u) \neq V$, there must be some arc~$xz$ entering $V(\Gamma_u)$. Using that $y$ is a root in $G$, we can choose $xz$ such that $z$ has a successor that is also a successor of~$y$. But then, $x$ and $z$ are not in $V(\Omega_u)$,
and so $xz \in \GW_u^\Gamma \setminus \GW_u^\Omega$.

In all cases, $\GW_u^\Gamma \neq \GW_u^\Omega$. This proves the lemma, and consequently, the fact that the sets $\GW$ uniquely determine a tree extension.
\end{proof}

\propcanon*
\begin{proof}
The `only-if direction' follows directly from \cite[Lem.\,5f]{Berry2022b}, thus it remains to prove the `if direction'. Let $\Gamma$ be a tree extension that is not canonical for $\sigma$. We now claim that there exists a $v\in V$ such that $\SW_v^{\sigma} \subset \GW_v^{\Gamma}$.

\begin{claimproof}
Let $v\in V$ and $xy \in \SW_v^{\sigma}$ be arbitrary. By definition, $x >_\sigma v \geq_\sigma y$ and $y \connect{G[1 \ldots v]} v$. As $\sigma$ is an extension of $\Gamma$, together this gives that $v \geq_\Gamma y$. Furthermore, $x$ and $v$ can not be incomparable in $\Gamma$, since $y \connect{G[1 \ldots v]} v$ and $xy$ is an arc of $G$. Combining with the fact that $x >_\sigma v$ and that $\sigma$ is an extension of $\Gamma$, we must then have $x >_\Gamma v$. This means that $xy \in \GW_v^{\Gamma}$. So, $\SW_v^{\sigma} \subseteq \GW_v^{\Gamma}$.

According to \cref{lem:GW_unique}, there exists some $v \in V$ such that $\GW_v^\Gamma \neq \GW_v^{\Gamma^\sigma}$. By \cite[Lem.\,5i]{Berry2022b}, we then get that $\SW_v^\sigma = \GW_v^{\Gamma^\sigma} \neq \GW_v^\Gamma$. We already had that $\SW_v^{\sigma} \subseteq \GW_v^{\Gamma}$, so $\SW_v^{\sigma}  \subset  \GW_v^\Gamma  $.
\end{claimproof}

Let $v$ be as in the claim, then there exists an arc $xy \in \GW_v^{\Gamma}$ that is not in $\SW_v^{\sigma}$. But since $\sigma$ is an extension of $\Gamma$, we must have that $x >_\sigma v \geq_\sigma y$. Then, $xy \notin \SW_v^{\sigma}$ implies that $y$ is not weakly connected to $v$ in $G[\sigma[1 \ldots v]]$. Clearly, $V(\Gamma_v) \subseteq \sigma[1 \ldots v]$, as $\sigma$ is an extension of $\Gamma$. But then, $y$ is also not weakly connected to $v$ in $G[V(\Gamma_v)]$. This means that $G[V(\Gamma_v)]$ is not weakly connected.
\end{proof}

\twbound*
\begin{proof}
Let $\Gamma$ be some tree extension of $G$ and $v \in V$ arbitrary. Clearly, $\Gamma$ is also a tree layout for $\tilde{G}$. By definition, $\GW_v^\Gamma (G)= \{xy \in E: x>_\Gamma v \geq_\Gamma y\}$. We now define the mapping $\phi : \GW_v^\Gamma (G) \rightarrow \TW_v^\Gamma (\tilde{G})$ as $\phi (xy) = x$. We first show that $\phi$ indeed maps all elements of $\GW_v^\Gamma (G)$ to $\TW_v^\Gamma (\tilde{G})$. For any $xy \in \GW_v^\Gamma (G)$, we surely have that $x \in V$ and that $x >_\Gamma v$. If we now set $y = w$, we immediately find the $w$ with $w\leq_\Gamma v$, such that $xw \in E$. Thus, $x$ is indeed in $\TW_v^\Gamma (\tilde{G})$.

Secondly, we prove that $\phi$ is surjective. To this end, let $u \in \TW_v^\Gamma (\tilde{G})$ be arbitrary. Then, there exists at least one~$w$ such that $u >_\Gamma v \geq_\Gamma w$ and $uw \in E$. But then, it holds that $uw \in \GW_v^\Gamma (G)$, and so $\phi(uw) = u$, proving that that $\phi$ is surjective.

Since $\phi$ is a surjective mapping from $\GW_v^\Gamma (G)$ to $\TW_v^\Gamma (\tilde{G})$, it must hold that $|\GW_v^\Gamma (G)| \geq |\TW_v^\Gamma (\tilde{G})|$. Then also, $\tw (\Gamma, G) = \max_v |\TW_v^\Gamma (\tilde{G})| \leq \max_v |\GW_v^\Gamma (G)| = \sw (\Gamma, G)$. Finally, letting $\Gamma$ be such that $\sw (G) = \sw (\Gamma, G)$ proves the bound.
\end{proof}

Recall that we write $W \sqsubseteq V$ to indicate that $W$ is a sinkset of a DAG $G=(V,E)$.

\begin{lemma}
\label{lem:sinkset_reticulation_nr}
Let $G=(V,E)$ be a network with reticulation number $k$. Then, for all $W\sqsubseteq V$ such that $G[W]$ is weakly connected, we have that $\deltain (W) \leq k + 1.$
\end{lemma}
\begin{proof}
For any DAG $H$ and $W\subseteq V(H)$, let $r_W (H)= \sum_{v \in W: \deltain(v) \geq 2} (\deltain (v) -1 )$. We will now prove the stronger statement that for any DAG $G=(V,E)$ and for all $W\sqsubseteq V$ with $G[W]$ weakly connected, $\deltain (W) \leq r_W (G) + 1$. For any such set $W$, we will prove this by induction on $r_W(G)$. This will immediately imply the lemma, since any network is a DAG and $r_W (G)$ is at most the reticulation number $k$ of a network $G$.

\emph{Base case:} ($r_W (G) =0$). Let $G$ be a DAG with $W\sqsubseteq V$ such that $G[W]$ is weakly connected and $r_W (G) = 0$. Then, $W$ contains no reticulation vertices, otherwise $r_W (G) > 0$. As $W$ is a weakly connected sinkset, $G[W]$ must then either be a pendent tree of $G$, or $G[W] = G$. In both cases, $\deltain (W) \leq 1 = r_W (G) + 1$.

\emph{Induction step:} ($r_W (G) \geq 1$). Let $G$ be a DAG with $W\sqsubseteq V$ such that $G[W]$ is weakly connected and $r_W (G) \geq 1$. Assume that the induction hypothesis holds for all DAGs $H$ and sets $W' \sqsubseteq V(H)$ such that $r_{W'} (H) < r_W (G)$. As $r_W (G) \geq 1$, we know that $W$ contains a reticulation vertex. We can then pick an arc $uv\in E$ such that $v$ is a reticulation vertex in $W$. It is easy to see that $W$ is still a sinkset of $G' = G - \{ uv \}$, and that $r_W(G') = r_W (G) - 1$. We now consider two cases and show that $\deltain_G (W) \leq r_W (G) + 1$, which will conclude the proof.

\emph{Case 1:} If $G'[W]$ is weakly connected, the induction hypothesis shows that $\deltain_{G'} (W) \leq r_W (G') + 1$. But then, if we add $uv$ back to $G'$ to obtain $G$ again, we can only increase the indegree of $W$ by 1, and it follows that $\deltain_G (W) \leq \deltain_{G'} (W) + 1  \leq r_W (G') + 2 =  r_W(G)+1$.

\emph{Case 2:} If $G'[W]$ is not weakly connected, the deletion of $uv$ must have disconnected $G[W]$. Deletion of one arc can only increase the number of weakly connected components by one. Therefore, $G'[W]$ consists of two weakly connected components: $G'[W_1]$ and $G'[W_2]$. Because $G'[W_1]$ and $G'[W_2]$ are disconnected in $G'$, and $W$ was a sinkset in $G'$, we must have that $W_1$ and $W_2$ are sinksets of $G'$. Since $W_1$ and $W_2$ partition $W$, we have $r_W (G') = r_{W_1} (G') + r_{W_2} (G')$.

For both $i=1$ and $i=2$, we now get that $r_{W_i} (G') \leq r_W (G') = r_W (G) -1$. Furthermore, $G'[W_i]$ is a weakly connected sinkset, as discussed above. Thus, we can apply the induction hypothesis, and get $\deltain_{G'} (W_i) \leq r_{W_i} (G') + 1$ for $i=1, 2$. Using this and the fact that no arc exists between $W_1$ and $W_2$ in $G'$, we get
$\deltain_{G'} (W) = \deltain_{G'} (W_1) + \deltain_{G'} (W_2) \leq  ( r_{W_1} (G') + 1 ) + ( r_{W_2} (G') + 1 )=  r_W (G') + 2$.
Since $uv$ disconnected $G[W]$, $u$ and $v$ must both be in $W$. Consequently, the arc $uv$ can not count towards the indegree of $W$. From this, it follows that $\deltain_G (W) = \deltain_{G'} (W) \leq r_W(G') + 2 = r_W(G)+1$.
\end{proof}

\levelbound*
\begin{proof}
Let $r$ be the reticulation number of $G$. We will prove that the scanwidth of $G$ is at most $r+1$. The result then follows from the definition of the level and from \cref{cor:split_blocks}, which shows that the scanwidth of a rooted weakly connected DAG is equal to the maximum scanwidth of its blocks.

Let $\sigma$ be an optimal extension of $G$. Let $v\in V$ be arbitrary, and consider the set $\SW_v^\sigma$. We can write $|\SW_v^\sigma| = \deltain (S)$, where $S \subseteq \sigma[1 \ldots v]$ contains all vertices that are weakly connected to $v$ in $G[1\ldots v]$. Clearly, $S$ is also a sinkset (as $\sigma$ is an extension of $G$). According to \cref{lem:sinkset_reticulation_nr}, we then have that $|\SW_v^\sigma| = \deltain(S) \leq r+1$. Since $v$ was arbitrary, $\sw(G) = \max_{v\in V} |\SW_v^\sigma| \leq r+1$.
\end{proof}

\bibliographystyle{elsarticle-num} 
\bibliography{references}

\end{document}